%% file: makespan.tex
\documentclass{article}
\input{envs}

\begin{document}

\title{Towards PTAS for Precedence Constrained Scheduling via Combinatorial Algorithms}
\author{Shi Li \thanks{Department of Computer Science and Engineering, University at Buffalo, Buffalo, NY 14260. The work is in part supported by NSF grant CCF-1844890.}}

\maketitle
\input{abstract}
\input{intro}
\input{prelim}
\input{dyadic}
\input{virtual}
\input{conversion}
\input{algo}

\bibliographystyle{plain}
\bibliography{reflist}

\appendix
\input{appendix}

\input{improved}
\input{discussion}

\end{document}

%% file: envs.tex
\usepackage{fullpage}
\usepackage{amsthm, amsmath, amsfonts, amssymb}
\usepackage{graphicx}
\usepackage{xcolor}
\usepackage{xspace}
\usepackage{thmtools} 
\usepackage{thm-restate}

\graphicspath{{figs/}}

\usepackage{enumitem}
\setlist{topsep=3pt,itemsep=0pt, parsep=0pt, leftmargin=*}

\usepackage{algorithm}
\usepackage{algpseudocode}
\algtext*{EndWhile}
\algtext*{EndIf}
\algtext*{EndFor}

\newtheorem{theorem}{Theorem}[section]
\newtheorem{lemma}[theorem]{Lemma}
\newtheorem{claim}[theorem]{Claim}

\newtheorem{obs}[theorem]{Observation}
\theoremstyle{definition}
\newtheorem{definition}[theorem]{Definition}
\newtheorem{remark}[theorem]{Remark}

\newcommand{\ceil}[1]{\left\lceil#1\right\rceil}
\newcommand{\floor}[1]{\left\lfloor#1\right\rfloor}

\DeclareMathOperator\union{\bigcup}
\mathchardef\mhyphen="2D

\newcommand{\Z}{{\mathbb{Z}}}
\newcommand{\calC}{{\mathcal{C}}}
\newcommand{\calI}{{\mathcal{{I}}}}
\newcommand{\calW}{{\mathcal{{W}}}}
\newcommand{\sfL}{{\mathsf{L}}}
\newcommand{\sfR}{{\mathsf{R}}}

\newcommand{\bbJ}{{\mathbb{J}}}

\newcommand{\bfT}{{\mathbf{T}}}

\newcommand{\depth}{\mathsf{depth}}
\newcommand{\dif}{\mathsf{dif}}

\newcommand{\sfleft}{{\mathsf{left}}}
\newcommand{\sfright}{{\mathsf{right}}}
\newcommand{\sfbegin}{{\mathsf{begin}}}
\newcommand{\sfcenter}{{\mathsf{center}}}
\newcommand{\sfend}{{\mathsf{end}}}
\newcommand{\side}{{\mathsf{side}}}

\newcommand{\sftop}{{\mathsf{top}}}

\newcommand{\sfmid}{{\mathsf{mid}}}

\newcommand{\sfbot}{{\mathsf{bot}}}

\newcommand{\len}{{\mathsf{len}}}

\newcommand{\discarded}{{\mathsf{disc}}}
\newcommand{\best}{{\mathsf{best}}}
\newcommand{\sfcap}{{\mathsf{cap}}}
\newcommand{\anc}{{\mathsf{anc}}}

\newcommand{\cut}{{\mathsf{cut}}}
\newcommand{\uncut}{{\mathsf{uncut}}}
\newcommand{\maxuncut}{{\mathsf{max\mhyphen uncut}}}

\newcommand{\poly}{{\mathrm{poly}}}

\newcommand{\cstrJStar}{\mathsf{construct\mhyphen\bbJ^*}}
\newcommand{\pushdown}{{\mathsf{push\mhyphen down}}}
\newcommand{\schedule}{\mathsf{schedule}}
\newcommand{\modified}{\mathsf{modified}}

\newcommand{\remove}[1]{}

\newcommand{\inorderless}{\mathbin{\overset{\mathsf{\scriptscriptstyle in\mhyphen ord}}{\scriptstyle <}}}
\newcommand{\Aless}{\mathbin{\overset{\mathsf{\scriptscriptstyle A}}{\scriptstyle <}}}
\newcommand{\notAless}{\mathbin{\overset{\mathsf{\scriptscriptstyle A}}{\scriptstyle \not<}}}
\newcommand{\sigmaless}{<_{\mathsf{side}_{\sigma'}}}

%% file: abstract.tex
\begin{abstract}
	We study the classic problem of scheduling $n$ precedence constrained unit-size jobs on $m = O(1)$ machines so as to minimize the makespan. In a recent breakthrough, Levey and Rothvoss \cite{LR16} developed a $(1+\epsilon)$-approximation for the problem with running time $\exp\Big(\exp\Big(O\big(\frac{m^2}{\epsilon^2}\log^2\log n\big)\Big)\Big)$, via the Sherali-Adams lift of the basic linear programming relaxation for the problem by $\exp\Big(O\big(\frac{m^2}{\epsilon^2}\log^2\log n\big)\Big)$ levels. Garg \cite{Garg18}  recently improved the number of levels to $\log ^{O(m^2/\epsilon^2)}n$, and thus  the  running time to $\exp\big(\log ^{O(m^2/\epsilon^2)}n\big)$, which is quasi-polynomial for constant $m$ and $\epsilon$. 
	
	In this paper we present a $(1+\epsilon)$-approximation algorithm for the problem with running time $n^{O\left(\frac{m^4}{\epsilon^3}\log^3\log n\right)}$, which is very close to a polynomial for constant $m$ and $\epsilon$.   Unlike the algorithms of Levey-Rothvoss and Garg, which are based on the linear-programming hierarchy, our algorithm is purely combinatorial.  We show that the conditioning operations on the lifted LP solution can be replaced by making guesses about the optimum schedule.  
	
	Compared to the LP hierarchy framework, our guessing framework has two advantages, both playing important roles in deriving the improved running time. First, we can guess any information about the optimum schedule, as long as it can be described using a few bits, while in the conditioning framework, we can only condition on the variables in the basic LP. Second, the guessing framework can save a factor of $\log n$ in the exponent of running time.  Roughly speaking, most of the time, the information we try to guess is binary and thus each nested guess only contributes to a multiplicative factor of $2$ in the running time. In contrast, each conditioning operation in a sequence incurs a multiplicative factor of $\poly(n)$. 
\end{abstract}

%% file: intro.tex
\section{Introduction}


The problem of scheduling $n$ precedence constrained unit length jobs on $m$ identical machines so as to minimize the makespan is a fundamental problem in scheduling theory. In the problem, we are given $m$ identical machines, a set $J^\circ$ of $n$ unit-size jobs, with precedence constraints given by a strict partial order $\prec$ over $J^\circ$.\footnote{A relation $\Aless$ over some set $A$ is a strict partial order if for every $a\in A$, we do not have $a\Aless a$, for every $a, b, c \in A$ with $a \Aless b$ and $b \Aless c$, we have $a \Aless c$, and for every $a, b \in A$ we can not have both $a \Aless b$ and $b \Aless a$. This definition will be used later.} If we have $j \prec j'$, then the job $j'$ can only start after job $j$ completes.  The goal of the problem is to schedule all jobs in $J^\circ$ so as to minimize the makespan of the schedule, which is defined as the time by which all jobs compete, assuming the schedule starts at time $0$.  Using the classic three-field notation, the problem is denoted as $P|\text{prec}, p_j=1|C_{\max}$.

Already in 1966, Graham \cite{Gra69} showed that any greedy non-idling schedule for the problem is $(2 - 1/m)$-approximate. When $m \geq 4$, a slightly better approximation ratio of $2 - 7/(3m + 1)$ can be achieved \cite{GR08}. Later, Svensson \cite{Sve10} proved that under a variant of the Unique Games Conjecture (UGC) introduced by Bansal and Khot \cite{BK09}, there is no polynomial time algorithm that can achieve an approximation factor of $(2 - \epsilon)$ for the  problem. Thus, under the conjecture and ignoring $o(1)$ terms in the approximation ratio, Graham's algorithm already gives the best ratio for the problem in polynomial time.

An important special case of the problem that has attract a lot of attention recently is when the number $m$ of machines is a small constant, as in most real world applications, $m$  is typically much smaller than $n$, the number of jobs.  A natural question is if one can design better approximation algorithms for this special case, denoted as $Pm|\text{prec}, p_j = 1|C_{\max}$ using the three-field notation, where $Pm$ indicates that the number $m$ of machines is a constant that is not a part of the input.  On the negative side, whether the problem $P3|\text{prec}, p_j = 1|C_{\max}$ is NP-hard or not is a long-standing open problem.  On the positive side, in a recent breakthrough, Levey and Rothvoss \cite{LR16} developed a $(1+\epsilon)$-approximation for $Pm|\text{prec}, p_j = 1|C_{\max}$ with running time $\exp\left(\exp \left(O(\frac{m^2}{\epsilon^2}\log^2 \log n)\right)\right)$, via the Sherali-Adams lift \cite{SA90} of the basic LP relaxation of the problem by $\exp \left(O(\frac{m^2}{\epsilon^2}\log^2 \log n)\right)$ levels. Later, Garg \cite{Garg18} reduced the number of levels of the Sherali-Adams hierarchy to $\log^{O(m^2/\epsilon^2)} n$, and thus improving the running time to $\exp\left(\log ^{O(m^2/\epsilon^2)}n\right)$, which is strictly quasi-polynomial.

An important open question that follows is whether the running time can be made strictly polynomial; that is, whether we can obtain a PTAS for $Pm|\text{prec}, p_j=1|C_{\max}$. Indeed, this is listed as the first open problem in both the influential survey by Schuurman and Woeginger \cite{SW99a}, and  the recent report by Bansal  \cite{Ban17} on approximate scheduling problems. Both Levey and Rothvoss \cite{LR16} and Garg \cite{Garg18} asked specifically whether a PTAS can be obtained using an $O(1)$-level Sherali-Adams lift of the basic LP relaxation. 

We note that even though the running time of the algorithm of Garg is quasi-polynomial, the  $O(1)$ exponent in the exponent $\log^{O(1)} n$ of the running time depends on $m$ and $\epsilon$. Moreover, both algorithms of Levey-Rothvoss and Garg are recursive: The initial instance has size $n$, and the algorithms reduce the instance to many sub-instances of size $n/\poly\log(n)$ and solve them recursively.  Therefore the number of levels in the recursion is at least $\Omega(\log n/\log \log n)$. Using the Sherali-Adams hierarchy framework, it seems hard to avoid a factor of $n^{\Omega(\log n/\log\log n)}$ in the running time, if the algorithm is recursive.

In this paper, we give a new framework that improves the running time to $n^{O\left(\frac{m^4}{\epsilon^3}\log^3 \log n\right)}$, thus making a big step towards obtaining a PTAS for the problem:
\begin{theorem}
	\label{thm:main}
	There is an $n^{O\left(\frac{m^4}{\epsilon^3}\log^3 \log n\right)}$-time $(1+\epsilon)$-approximation algorithm for  $Pm|\text{prec}, p_j=1|C_{\max}$, i.e, the problem of scheduling precedence-constrained unit-size jobs on $m$ identical machines to minimize the makespan. 
\end{theorem}

The exponent in our running time is only $\poly\left(m, \frac{1}{\epsilon}, \log \log n\right)$; thus we overcome the $n^{\Omega(\log n/\log\log n)}$ barrier mentioned above.   Moreover, our running time is only single exponential in $\poly\left(m, \frac1\epsilon\right)$, while that of Garg \cite{Garg18} has a double exponential dependence on $\poly\left(m, \frac1\epsilon\right)$.

 Unlike the algorithms of Levey-Rothvoss and Garg, which are based on the Sherali-Adams hierarchy, our algorithm is purely combinatorial. We show that the conditioning operations on the lifted LP solution can be replaced by \emph{making guesses} about the optimum schedule. The guessing framework has two advantages. First, it is more flexible in the sense that we can guess any information about the optimum solution, as long as it can be described using a few bits, while we can only condition on actual variables in the basic LP relaxation.  Second, the running time given by the framework depends on the number of possibilities for the combination of our guesses; for our result, the dependence gives a better running time than that given by the dependence on the number of Sherali-Adams levels. Our algorithm is recursive and we can not avoid an $\Omega(\log n)$ number of levels in the recursion.  However, most of the time the information we try to guess is binary. Roughly speaking, instead of losing an $n^{\Omega(\log n)}$ factor in the running time, we only lose a factor of $2^{\Omega(\log n)}$. Other than the improved running time, we believe our framework is conceptually simpler and interesting on its own.  To present more detail about our techniques, we first give an overview the algorithms of Levey-Rothvoss and Garg. 

\subsection{Overview of  Levey-Rothvoss and Garg} 
The algorithms of Levey-Rothvoss and Garg are based on the Sherali-Adams hierarchy of the basic LP relaxation for the problem, and we refer to \cite{Lau01, CT12, Rot13} for beautiful surveys of LP/SDP hierarchies and their applications. For this overview, it suffices to keep the following informal description in mind.  Given a basic LP relaxation of size $N$ for some problem, we can ``lift'' it by $r \geq 1$ levels to obtain a new LP relaxation of size $N^{O(r)}$. Solving the lifted LP gives us an $r$-level fractional solution $x$. An important operation defined over an $r$-level fractional solution $x$, which has been used in many hierarchy-based algorithms, is called ``conditioning'': Taking any variable $x_j$ in the basic LP relaxation with $x_j > 0$,  ``conditioning'' on the event that $x_j = 1$ yields a new LP solution $x'$ with $x'_j = 1$, but $x'$ is only an $(\ell-1)$-level fractional solution. Thus, we can only apply $r-1$ conditioning operations sequentially on an LP solution of level $r$.

It is known that the list-scheduling algorithm of Graham \cite{Gra69} gives a schedule with makespan at most $\Delta(J^\circ) + \frac{|J^\circ|}{m}$, where $\Delta(J)$ for any $J \subseteq J^\circ$  is the maximum length of a precedence chain of jobs in $J$. Since both $\Delta(J^\circ)$ and  $\frac{|J^\circ|}{m}$ are lower bounds on the optimum makespan $T$, the algorithm gives a $2$-approximation.  If $\Delta(J^\circ)$ is very small, say $\epsilon$ times the optimum makespan $T$, then the list scheduling algorithm already gives a $(1+\epsilon)$ approximation.  So intuitively, a set $J$ of jobs with a small $\Delta(J)$ value is easy to schedule. 

The Levey-Rothvoss algorithm \cite{LR16} exploited the intuition in the following way.  A basic structure used by the algorithm  is a dyadic tree $\bfT$ of intervals, with root being $[T]$, leaves being singular intervals, and the two children of an internal interval being its left and right-half sub-intervals. Levey-Rothvoss first solves an $r$-level lift of the basic LP relaxation for the problem to obtain a fractional solution $x$, for some integer $r = \exp(\frac{m^2}{\epsilon^2}\log^2\log T)$.  Then every job $j \in J^\circ$ is \emph{assigned} to the inclusion-wise minimal interval $I$ in $\bfT$ that contains all the time slots $t$ with $x_{j,t} > 0$, where $x_{j,t}$ is the variable in the basic LP relaxation indicating whether $j$ is scheduled at time $t$ or not. We say $I$ is the \emph{owning interval} of $j$. Let $J_I$ be set of all jobs with owning interval $I$, or equivalently, assigned to $I$.  So, every job $j$ is scheduled in its owning interval, according to the LP solution $x$. If  $\Delta(J_{[T]})$ is large, the algorithm  can take a long precedence chain in $J_{[T]}$, pick the middle job $j$ in the chain, choose an arbitrary time $t$ with $x_{j, t}>0$, and  condition on that $x_{j, t} = 1$. Thus in the new LP solution $x$, $j$ is scheduled at time $t$. If $t \leq T/2$, then $j$ and all its predecessors in $J_{[T]}$ must be scheduled in $(0, T/2]$ according to $x$. Thus the new owning intervals of these jobs become sub-intervals of $(0, T/2]$. Similarly, if $t > T/2$, the owning intervals of $j$ and its successors in $J_{[T]}$ will be changed to sub-intervals of $(T/2, T]$. In either case the algorithm is making a reasonable progress: the owning intervals of at least $\Delta(J_{[T]})/2$ jobs are shrunk.  The conditioning operation can then be repeated until $\Delta(J_{[T]})$ becomes very small.  The whole conditioning process is then repeated on $(0,T/2]$ and $(T/2, T]$ to make sure $\Delta(J_{(0, T/2]})$ and $\Delta(J_{(T/2, T]})$ are small, and then on levels $2, 3, \cdots, L_{\mathrm{LR}}$ of the dyadic tree $\bfT$, for some $L_{\mathrm{LR}} = O_{\epsilon,m}(\log^2\log T)$.

Then the Levey-Rothvoss algorithm carefully chooses three sets of levels from the $L_{\mathrm{LR}}$ levels: Top levels contain the topmost $a$ levels, middle levels contain the next $b$ levels below the top levels, and the bottom level is the level below the middle levels. It is guaranteed that $a + b + 1 \leq L_{\mathrm{LR}}$ and thus the top, middle and bottom levels all fall in the topmost $L_{\mathrm{LR}}$ levels in $\bfT$.  There are only a few jobs assigned to middle levels and thus they can be discarded immediately; jobs assigned to a bottom interval (called bottom jobs) are scheduled recursively by solving the instance defined over the bottom interval.  After obtaining a schedule of bottom jobs obtained from the recursive procedures, the top jobs (that is, jobs assigned to top intervals) are then inserted back. Levey-Rothvoss showed that only a few top jobs need to be removed due to two good properties: (i) $\Delta(J_I)$ for each top interval $I$ is small, which suggests that the dependence constraints among top jobs are easy to handle, and (ii) since $b$ is large, any top job has an owning interval that is much longer than the length of bottom intervals, allowing the algorithm to handle the precedence constraints between top and bottom jobs easily.    Overall the whole recursive algorithm discards at most $\epsilon T$ jobs, and inserting them back gives a schedule of makespan at most $(1+\epsilon)T$.

Garg \cite{Garg18} defined the owning intervals in a more flexible way: The owning interval of a job $j$ only needs to contain the time points $t$ with $x_{j, t} > 0$; it does not need to be the inclusion-wise minimal one.  By doing so, Garg can force some jobs to stay on top levels so that they do not contribute to the loss in sub-recursions, resulting in a better running time.

\subsection{Our Techniques}
As we mentioned, our algorithm is purely combinatorial. Similar to the algorithms of Levey-Rothvoss and Garg, we define the dyadic tree $\bfT$ and assign jobs to intervals in $\bfT$.  Let $J_I$ be the jobs assigned to $I$ (or with owning interval being $I$).  Instead of solving the Sherali-Adams lift of the basic LP relaxation for the problem and condition on variables to decide how to assign jobs to intervals, we \emph{guess} what happens in the optimum schedule. If all our guesses are correct, then we are sure that in the optimum solution every job is scheduled inside its owning interval.  Initially all jobs are assigned to the root interval $[T]$.  If $\Delta(J_{[T]})$ is big, we can then take the middle job $j$ in some length-$\Delta(J_{[T]})$ chain of jobs in $J_{[T]}$, as in Levey-Rothvoss. Now instead of referring to the fractional solution $x$ (which we do not have) what to do, we make a guess on whether $j$ is scheduled in $(0, T/2]$ or $(T/2, T]$ in the optimum schedule.  Suppose our guess is the former and it is correct. Then we are certain that $j$ and its predecessors are all scheduled in $(0, T/2]$; thus we can change their owning intervals to $(0, T/2]$. Similar to Levey-Rothvoss, by guessing repeatedly, we can make $\Delta(J_{[T]})$ small.

A natural way to proceed is to break the instance into two sub-instances over $(0, T/2]$ and $(T/2, T]$ respectively.  This requires us to split $J_{[T]}$ into two sets, one to the left and the other to the right.   However, once $\Delta(J_{[T]})$ becomes small, one guess can only yield a small progress and we can not afford to make guesses until $J_{[T]}$ becomes empty. To overcome this issue, we use the two ideas from Levey-Rothvoss. First, since $\Delta(J_{[T]})$ is small now, we can essentially ignore the precedence constraints among them.  Second, to take care of the precedence constraints between $J_{[T]}$ and $J^\circ \setminus J_{[T]}$,  we make guesses recursively to obtain information about the sets of jobs assigned to the first $h = \log \frac{O(1)\cdot m\log T}{\epsilon}$ levels of intervals in $\bfT$.  With this information, we have some rough knowledge on where a job in $J_{[T]}$ can be scheduled. This leads to the definition of a window for a job in $J_{[T]}$, and we impose the constraint that the job should be scheduled inside its window. We show that the precedence constraints between $J_{[T]}$ and $J^\circ \setminus J_{[T]}$ can be approximately captured by the window constraints; therefore they can be ignored.  A crucial property is that the boundaries of the windows are all multiplies of $2^{-h}T$, making the number of possible windows small.  Thus there are only a few different ways to split $J_{[T]}$.  By guessing how to split $J_{[T]}$, we can divide the instance into two separate sub-instances over $(0, T/2]$ and $(T/2, T]$, which are then solved independently and recursively.  Notice that our $h$ is small: we have $2^h = \Theta\left(\frac{m\log T}{\epsilon}\right)$. That means, we do not need to create a large gap between top and bottom intervals as in Levey-Rothvoss and Garg, allowing us to remove the $(\log n)^{\poly(m, 1/\epsilon)}$ factor in the exponent of the running time. 

However, the above framework can only lead to a running time of $n^{\tilde O_{m, \epsilon}(\log^2n)}$, where we hide a $\poly\log \log n$ factor in the $\tilde O$ notation, though this is already much better than the running time of Garg. The $\tilde O_{m, \epsilon}(\log^2n)$ term in the exponent comes from the need to guess how to assign  jobs to the first  $h$ levels of $\bfT$.  The flexibility of the guessing framework allows us to further improve the running time down to $n^{\tilde O_{m, \epsilon}(1)}$.  We show that we do not need the complete information for all the intervals in the first $h$ intervals of the dyadic tree $\bfT$.  Instead, we guess $\tilde O_{m, \epsilon}(1)$ critical intervals in the sub-tree of $\bfT$ at the first $h$ levels, and we only need the information relevant to the critical intervals.  This way the number of important intervals is reduced from $2^h$ to $\tilde O_{m, \epsilon}(h) = \tilde O_{m, \epsilon}(1)$. 

Thus, both advantages of the guessing framework play important roles in our improved running time. If we had to use the LP hierarchy and conditioning framework, we need to artificially introduce more variables in our LP to encode the information we need to guess (e.g, how to split $J_{[T]}$ between the two sub-instances), making the LP much more involved. The second advantage allows us to save a logarithmic factor in the exponent of the running time, which is critical in obtaining the running time of $n^{\tilde O_{m, \epsilon}(1)}$. From the above overview, we can see that most of the time we make guesses on whether a job is scheduled in the left or right half sub-interval of its current owning interval, which has a binary answer. So each guess will incur a multiplicative factor of $2$ in the running time.  If we use the LP hierarchy framework, we need 1-level in the LP hierarchy for each guess, which corresponds to a multiplicative factor of $\poly(n)$ in the running time. 

To deliver our techniques more smoothly, we first show how to obtain the $(1+\epsilon)$-approximation for $Pm|\text{prec}, p_j=1|C_{\max}$ in time $n^{O\left(\frac{m^4}{\epsilon^3}\log^2n\log\log n\right)}$, in which we make guesses to obtain the complete information on the first $h$ levels of the dyadic tree $\bfT$.  This already covers many essential techniques in our algorithm. Then we show how to further improve the running time to the claimed $n^{O\left(\frac{m^4}{\epsilon^3}\log^3\log n\right)}$ in the appendix.

\subsection{The Power of Linear Programming Hierarchy vs Guessing} In this paper we show that for the makespan minimization problem, the conditioning operations on the lifted LP solution can be replaced by making guesses on what happens in the optimum schedule. This phenomenon arises in some other problems as well.  Grandoni, Laekhanukit, and Li \cite{GLL19} recently gave a tight quasi-polynomial time $O(\log^2k/\log\log k)$-approximation for the Directed Steiner Tree problem, based on the Sherali-Adams hierarchy. Later, Ghuge and Nagarajan \cite{GN20} showed that the same result can be obtained using a combinatorial algorithm, based on guessing what happens in the optimum directed Steiner tree.  The guess-and-divide framework was also used in a recent result of Lokshtanov et al.\ \cite{LMM20} to obtain a tight $2$-approximation for the feedback vertex set on tournament graphs in polynomial time.  One can show that the $2$-approximation can be obtained via an $O(\log n)$-level lift of the Sherali-Adams hierarchy. But due to the recursiveness of the algorithm, it is not clear how one can avoid the $O(\log n)$ factor. So, for this problem, the combinatorial algorithm gives a better running time.  It is interesting to study for many other problems which admit LP hierarchy based algorithms, whether we can use the guessing framework to recover or improve upon these algorithms.

%% file: prelim.tex
\section{Preliminaries}
\label{sec:prelim}
	Throughout the paper, we use $J^\circ$ to denote the set of all jobs in the input instance, as $J$ will be used heavily. Let $n = |J^\circ|$. By binary search, we assume we know the optimum makespan $T$; notice that $T \leq n \leq mT$. We can assume $T$ is an integer power of $2$ using the reduction described in Appendix~\ref{appendix:manipulations}. To construct a schedule for $\bbJ^\circ$ with makespan at most $(1+\epsilon)T$, it suffices for us to construct a schedule of makespan $T$ with at most $\epsilon T$ jobs discarded, as explained in Appendix~\ref{appendix:manipulations}. Thus we set this as our new goal.  This transformation has also been used in Levey-Rothvoss and Garg.  Since we are allowed to discard jobs, we make the following definition:
	\begin{definition}
		\label{def:original-valid}
		A valid schedule for the input instance $(J^\circ, m, \prec)$ is a vector $\sigma \in ([T] \cup\{ \discarded\})^{J^\circ}$ satisfying:
		\begin{itemize}[itemsep=0pt]
			\item \textbf{capacity constraints}: for every $t \in[T]$, we have $|\sigma^{-1}(t)| \leq m$,  and
			\item \textbf{precedence constraints}: for every $j, j' \in J^\circ\setminus \sigma^{-1}(\discarded)$ with $j \prec j'$, we have $\sigma_j < \sigma_{j'}$.
		\end{itemize}
	\end{definition}
	 In the above definition we used the following shorthands. For every schedule $\sigma' \in (I \cup\{\discarded\})^J$ of some $J \subseteq J^\circ$ in some interval $I \subseteq [T]$, we define $\sigma'^{-1}(t) = \{j \in J: \sigma'_j = t\}$ for every $t \in I \cup \{\discarded\}$. We also define $\sigma'^{-1}(I'):= \{j \in J: \sigma'_j \in I'\} = \union_{t \in I'}\sigma'^{-1}(t)$,  for every sub-interval $I' \subseteq I$. We say jobs in $\sigma'^{-1}(\discarded)$ are discarded in the schedule $\sigma'$.  Our goal is then to find a valid schedule $\sigma \in ([T] \cup \discarded)^{J^\circ}$  with at most $\epsilon T$ jobs discarded.
	
	 \paragraph{Definitions and Notations Related to Precedence Constraints} Given two disjoint sets $J, J' \subseteq J^\circ$, we say there are no precedence constraints from $J$ to $J'$ if for every $j \in J, j' \in J'$, we have $j \not\prec j'$.  We say there are no precedence constraints between $J$ and $J'$ if for every $j \in J, j' \in J'$, we have $j \not \prec j'$ and $j' \not \prec j$.  If $J$ ($J'$, resp.) is a singleton set,  we can replace it with the job it contains in both definitions.  We say a sequence $J_1, J_2, \cdots, J_k$ of disjoint sets of jobs respects the precedence constraints if there are no precedence constraints from $J_{i'}$ to $J_i$ for any $1 \leq i < i' \leq k$.
	 
	Given a subset $J \subseteq J^\circ$ of jobs, we shall use $\Delta(J)$ to denote the length of the longest precedence chain $j_1 \prec j_2 \prec j_3 \prec \cdots \prec j_r$ with $j_1, j_2, \cdots, j_r \in J$. Notice that the $\Delta$ function is subadditive: We have $\Delta(J_1 \cup J_2 \cup \cdots \cup J_k) \leq \Delta(J_1) + \Delta(J_2) + \cdots + \Delta(J_k)$ for $k$ subsets $J_1, J_2, \cdots, J_k$ of $J^\circ$.   For every $J \subseteq J^\circ$ and some $j \in J$, we use $\depth_J(j)$ to denote the length of longest precedence chain $j_1 \prec j_2 \prec j_3 \prec \cdots \prec j_r$ with $j_1, j_2, \cdots, j_r \in J$ and $j_r = j$. It is easy to see that for two jobs $j, j' \in J$ with $j \prec j'$ we have $\depth_J(j) < \depth_J(j')$.  For every $J \subseteq J^\circ$ and $j \in J$, we use $N^-_J(j):=\{j' \in J: j' \prec j\}$ and $N^+_J(j):=\{j' \in J: j\prec j'\}$ to denote the set of predecessors and successors of $j$ in $J$ respectively.
	
	
	\paragraph{Global Parameters}
	Throughout the paper, we shall use the following important global parameters: $h = \left\lceil\log \frac{8m\log T}{\epsilon}\right\rceil = O(\log \log T ), L = \log T - h, h' = \ceil{\log \frac{4m}{\epsilon}}, \delta = \frac{\epsilon}{16\cdot2^hm^2} = \Theta\left(\frac{\epsilon^2}{m^3\log T}\right), \delta' = \frac{1}{2\cdot 2^{2h}}$ and $p = \floor{\frac2\delta\ln\frac{m}{\delta'}} + 1 = \Theta\left(\frac{m^3\log T\log \left(\frac m\epsilon\log T\right)}{\epsilon^2}\right)$, where the $\log$ function has base $2$.  
	For simplicity, we assume $T$ is sufficiently large compared to $m$ and $1/\epsilon$. 

	\subsection{Dyadic Interval Tree and Related Definitions and Notations}
	As in Levey-Rothvoss, a basic structure used in our algorithm is a dyadic tree of intervals in $[T]$. The tree is rooted at the whole interval $[T]$, and each internal interval $(b,e]$ has $(b, (b+e)/2]$ as its left child and $((b+e)/2, e]$ as its right child. All the bottom intervals have length $2^h = \Theta\left(\frac{m\log T}{\epsilon}\right)$; that is, we stop the partitioning if we obtained intervals of length $2^h$.  Since we assumed $T$ is an integer power of $2$, the starting and ending time of all intervals in the tree are integers. We view time intervals as intervals of integers and often use the left-open-right-closed form to denote them. That is, for integers $0 \leq b \leq e \leq T$, $(b, e]$ contains integers $b+1, b+2, \cdots, e$ and thus the size of $(b, e]$ is $e-b$. 

	We use $\bfT$ to denote the dyadic tree, and $\calI$ to denote the set of intervals in the tree. Thus $\bfT$ contains $\log \frac{T}{2^h}+1 = \log T - h + 1 =  L+1$ levels and we index them using $0$ to $L$ from top to bottom: For every $\ell \in [0, L]$, there are  $2^\ell$ intervals of  length $2^{-\ell}T$ at level $\ell$; we use $\calI_\ell$ to denote these intervals.    We say levels $0$ to $L-h'-1$ are top levels, levels $L-h'$ to $L-1$ are middle levels, and level $L$ is the bottom level.  The intervals at top, middle and bottom levels are called top, middle and bottom intervals respectively, and we use $\calI_\sftop := \union_{\ell \in [0, L-h'-1]}\calI_\ell$, $\calI_\sfmid := \union_{\ell \in [L-h', L-1]} \calI_\ell$ and $\calI_\sfbot = \calI_L$ to denote them. For any $\ell$ we use $\calI_{<\ell}, \calI_{\leq\ell}, \calI_{>\ell}$ and $\calI_{\geq\ell}$ for $\union_{\ell'<\ell}\calI_{\ell'},  \union_{\ell'\leq\ell}\calI_{\ell'}, \union_{\ell'>\ell}\calI_{\ell'}$ and $\union_{\ell'\geq\ell}\calI_{\ell'}$ respectively. For simplicity we assume if $\ell \notin [0, L]$ then $\calI_\ell = \emptyset$. See Figure~\ref{fig:tree} for the dyadic tree structure. 

	\begin{figure}
		\centering
		\includegraphics[width=0.5\textwidth]{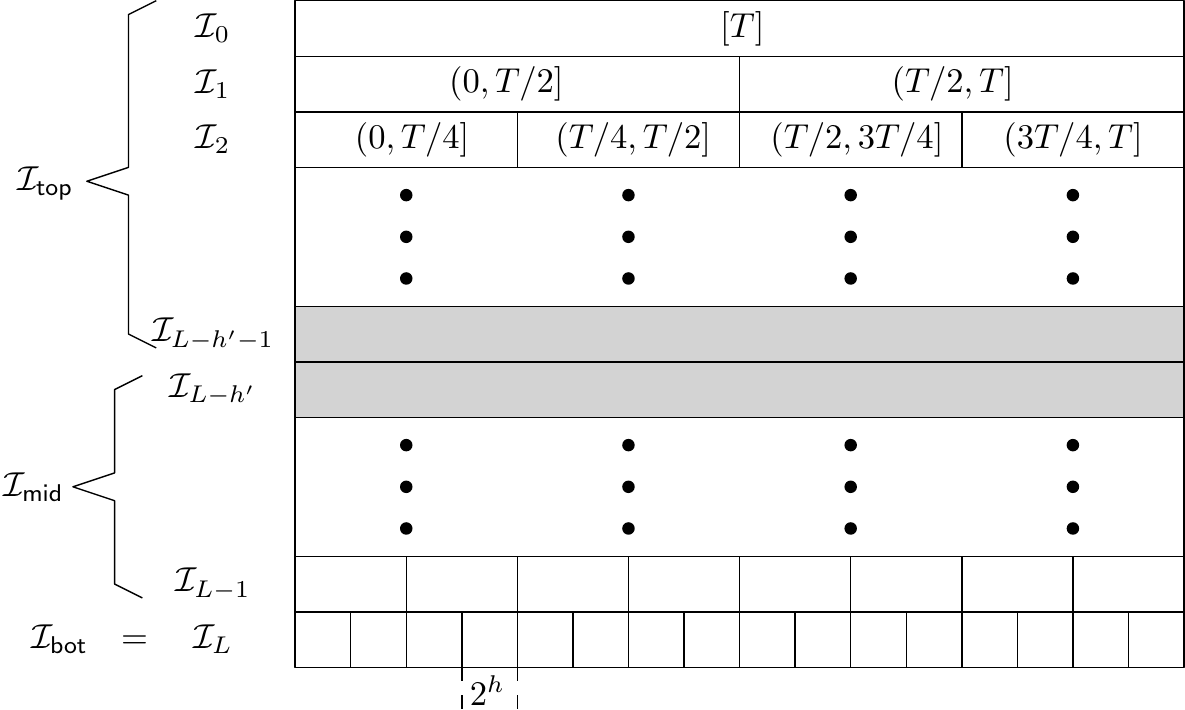} \caption{Dyadic Tree $\bfT$.} \label{fig:tree}
	\end{figure}
	
	Fix any $I\in\calI$, we use $\calI[I] = \{I' \in \calI: I' \subseteq I\}$ to denote the sub-intervals of $I$ in $\calI$, and define $\calI_\sftop[I] = \calI_\sftop \cap \calI[I]$, $\calI_\sfmid[I] = \calI_\sfmid \cap \calI[I]$ and $\calI_\sfbot[I] = \calI_\sfbot \cap \calI[I]$.  For every $\ell \in [0, L]$ and $I^* \in \calI$, we use $\calI_\ell[I^*]$ to denote the set of intervals in $\calI[I^*]$ of length $2^{-\ell}|I^*|$. Notice that in this definition $\ell$ is the \emph{relative} level:  Every interval in the set $\calI_\ell[I^*]$ has its level being $\ell$ plus the level of $I^*$.  Similarly, for every $\ell$, we use $\calI_{<\ell}[I^*], \calI_{\leq\ell}[I^*], \calI_{>\ell}[I^*]$ and $\calI_{\geq\ell}[I^*]$ for $\union_{\ell'<\ell}\calI_{\ell'}[I^*],  \union_{\ell'\leq\ell}\calI_{\ell'}[I^*], \union_{\ell'>\ell}\calI_{\ell'}[I^*]$ and $\union_{\ell'\geq\ell}\calI_{\ell'}[I^*]$ respectively. 
	
	For an interval $I \in \calI_\sftop \cup \calI_\sfmid$, we use $\sfleft(I)$ and $\sfright(I)$ to denote the left and right child intervals of $I$ respectively. If some $I \in \calI$ has $I= (b, e]$, we define $\sfbegin(I) = b, \sfend(I) = e$ and $\sfcenter(I) = (b+e)/2$ to denote the starting, ending and middle time points of $I$ respectively.  Notice that $\sfbegin(I) \notin I$ and $\sfend(I) \in I$.  Thus, $I = (\sfbegin(I), \sfend(I)], \sfleft(I) = (\sfbegin(I), \sfcenter(I)]$ and $\sfright(I) = (\sfcenter(I), \sfend(I)]$. 
	For two intervals $I, I' \in \calI$, we use $I \inorderless I'$ to denote that $I$ is before $I'$ in the in-order traversal of $\bfT$. Notice that this is equivalent to $\sfcenter(I) < \sfcenter(I')$, but we always use $I \inorderless I'$ since it emphasizes the relationship with the in-order traversal of $\bfT$.
	
	We compare our dyadic-tree structure with that in \cite{LR16} and \cite{Garg18}. The trees used in \cite{LR16} and \cite{Garg18} are very shallow:  the depth is $O_{m, \epsilon}(\log^2\log T)$ in \cite{LR16} and $O_{m, \epsilon}(\log\log T)$ in \cite{Garg18}.
	 So, their algorithms are recursive: the instance correspondent to a bottom interval has to be solved recursively. In contrast, our dyadic tree has depth $\log T - O(\log \log T)$, and each bottom instance can be solved directly by enumeration. Later, we shall see the recursiveness of our algorithm is in the construction of a ``{dyadic} system''.  Also, the number of middle levels in our tree is only $h' = O(\log \frac{m}{\epsilon})$, which is independent of $T$, while the number is $O_{m, \epsilon}(\log\log T)$ in both Levey-Rothvoss and Garg. 

	We use the following shorthands throughout the paper. Suppose we have a vector $(A_I)_{I \in \calI'}$ of sets for some subset $\calI' \subseteq \calI$ of intervals. Then for every interval $W \subseteq [T]$, we define $A_{\subseteq W} := \union_{I \in \calI': I \subseteq W} A_I, A_{\supseteq W} := \union_{I \in \calI': I \supseteq W} A_I$ and $A_{\supsetneq W} := \union_{I \in \calI': I \supsetneq W} A_I$.  In this definitions, $A$ can be replaced by other symbols.

	\subsection{Helper Lemmas}  Now we give some simple helper lemmas that will be used later. Their proofs can be found in Appendix~\ref{appendix:proof-helper}. The following lemma says that indeed a set $J$ of jobs with small $\Delta(J)$ value is easy to schedule:
	\begin{restatable}{lemma}{listscheduling}
		\label{lemma:list-scheduling}
		Given a set $J \subseteq J^\circ$ of jobs, a time interval $I \subseteq [T]$, and an integer capacity function $\sfcap: I \to [0, m]$ such that $\sum_{t \in I} \sfcap(t) \geq |J|$. Then we can efficiently find a schedule $\sigma \in (I \cup \{\discarded\})^{J}$ of $J$ in $I$ satisfying the precedence constraints and the capacity constraints w.r.t $\sfcap$: for every $t\in I$ we have $|\sigma^{-1}(t)| \leq \sfcap(t)$.  Moreover, the number of discarded jobs in $\sigma$ is at most $m\Delta(J)$.
	\end{restatable}

	To state the following lemma, we need a small definition. For four integers $z_1, z_2, z_3, z_4$, we define $\langle z_1, z_2\rangle < \langle z_3, z_4\rangle$ if $z_1 < z_3$ or $z_1 = z_3$ and $z_2 < z_4$.  Thus $<$ defines a total order over the 2-dimensional vectors $\langle z, z'\rangle \in \Z^2$. We define $\langle z_1, z_2\rangle \leq \langle z_3, z_4\rangle$ if $\langle z_1, z_2\rangle < \langle z_3, z_4\rangle$ or $\langle z_1, z_2\rangle = \langle z_3, z_4\rangle$.
	
	\begin{restatable}{lemma}{lengthaddup}
		\label{lemma:length-add-up}
		Let $J \subseteq J^\circ$ and $c: J \to Z \subseteq \Z$ be a function that maps $J$ to integers in $Z$.
		Let  $J_1, J_2, \cdots, J_k$  be disjoint subsets of $J$ (which do not necessarily form a partition of $J$). Assume the sequence $J_1, J_2, \cdots, J_k$ is consistent with the ordering of $\langle c(j), \depth_J(j)\rangle$: Formally, for every $j \in J_i, j' \in J_{i'}$ with $1\leq i < i' \leq k$ we have $\langle c(j), \depth_J(j)\rangle \leq \langle c(j'), \depth_J(j')\rangle$. Then
		\begin{align*}
			\Delta(J_1) + \Delta(J_2) + \cdots + \Delta(J_k) \leq |Z|\cdot \Delta(J) + k - 1.
		\end{align*}
	\end{restatable}
	
	The next lemma requires one definition.  Let $A$ be a set with a strict partial order $\Aless$. Let $\pi : A \to \Z$ be a function from $A$ to integers. We say an unordered pair $\{a, b\}$ in $A$ is an inversion in $\pi$ w.r.t the relation $\Aless$ if $a \Aless b$ but $\pi(b) < \pi(a)$, or $b \Aless a$ but $\pi(a) < \pi(b)$. The lemma says that swapping $\pi(a)$ and $\pi(b)$ for an inversion $\{a, b\}$ in $\pi$ will decrease the number of inversions.	
	\begin{restatable}{lemma}{reduceinversions}
		\label{lemma:reduce-inversions}
		Let $A$ be a set with a strict partial order $\Aless$ and $\pi : A \to \Z$ to be a function from $A$ to integers. 
	 	Let $\{a, b\}$ be an inversion in $\pi$ w.r.t $\Aless$, and $\pi'$ be obtained from $\pi$ by swapping $\pi(a)$ and $\pi(b)$: $\pi'(a) = \pi(b), \pi'(b) = \pi(a)$ and $\pi'(c) = \pi(c)$ for every $c \in A \setminus \{a, b\}$.  Then the number of inversions in $\pi'$ w.r.t $\Aless$ is strictly less than that in $\pi$. 
	\end{restatable}

\subsection{Overview of the Algorithm and Organization of the Paper}
To deliver our ideas more smoothly, we first prove Theorem~\ref{thm:main} with a worse running time of $n^{O\left(\frac{m^4}{\epsilon^3}\cdot\log^2n\cdot \log \log n\right)}$, which is already much better than the running time of $n^{\log^{O(m^2/\epsilon^2)}n}$ given by Garg \cite{Garg18}.  We show how to improve the running time to $n^{O\left(\frac{m^4}{\epsilon^3}\log^3 \log n\right)}$ in Appendix~\ref{sec:improved} and \ref{sec:improved-algo}.

In Section~\ref{sec:dyadic}, we define an important structure called a \emph{dyadic system} (Definition~\ref{def:dyadic-system}), which corresponds to an assignment $(J_I)_{I \in \calI}$ of $J^\circ$ to top and bottom intervals in $\bfT$. The assignment is consistent with the precedence constraints: If $I \inorderless I'$ then there are no precedence constraints from jobs assigned to $I'$ to jobs assigned to $I$.  Jobs assigned to a top interval $I$ has a small maximum chain length compared to $|I|$.   Then in a valid schedule for a dyadic system we require jobs assigned to some interval $I$ is scheduled in $I$ or discarded (Definition~\ref{def:valid}). Our algorithm will try to construct a dyadic system $\bbJ$ along with a valid schedule for it.  

In Section~\ref{subsec:optimum-dyadic} we show how to construct a dyadic system $\bbJ^*$ from the optimum schedule $\sigma^*$, and a set $\{g^*_I\}_{I \in \calI_\sftop \cup \calI_\sfmid}$ of short vectors in $\{\sfL, \sfR\}^*$.  Each $g^*_I$ gives us the answers to the guesses we made to reduce the chain length of jobs assigned to $I$. The $\bbJ^*$ and $g^*_I$ vectors are only used in the analysis since our algorithm does not know $\sigma^*$.  Roughly speaking, our algorithm tries to guess the vectors $\{g^*_I\}$, or equivalently, how to assign each job in the middle of a long chain to the left or right half of its owning interval, to recover the dyadic system $\bbJ^*$. 

To allow us to ignore the precedence constraints incident on top jobs  in a dyadic system $\bbJ$ (i.e, jobs assigned to top intervals in $\bbJ$), we define a window for each top job in Section~\ref{sec:virtual}. We then replace the precedence constraints incident on top jobs to window constraints: Each top job should be scheduled in its window or discarded. We call such a schedule a \emph{virtually-valid} schedule for $\bbJ$ (Definition~\ref{def:virtually-valid}).  

Then in Section~\ref{sec:conversion}, we show valid and virtually-valid schedules for a dyadic system can be converted to each other, up to the discarding of a few jobs. First  in Section~\ref{subsec:valid-to-virtually-valid} we show that the optimum schedule $\sigma^*$ can be converted to a virtually-valid schedule $\sigma'^*$ for the dyadic system $\bbJ^*$ with a few discarded jobs. Second, given any dyadic system $\bbJ$ and a virtually-valid schedule $\sigma''$ for $\bbJ$, we can efficiently construct a valid schedule $\sigma$ for $\bbJ$ with a few extra jobs discarded (Section~\ref{subsec:virtually-valid-to-valid}).  With the connection, the goal of our algorithm becomes to make guesses to recover $\bbJ^*$ and the virtually-valid schedule $\sigma'^*$.  

Finally in Section~\ref{sec:algo}, we present our recursive algorithm which constructs a dyadic system $\bbJ^\best$ along with a virtually-valid schedule $\sigma^\best$ by making guesses on $\bbJ^*$ and $\sigma'^*$.  The returned $\bbJ^\best$ and $\sigma^\best$ may be different from $\bbJ^*$ and $\sigma'^*$, but due to the existence of $\bbJ^*$ and $\sigma'^*$, we are guaranteed that $|(\sigma^\best)^{-1}(\discarded)| \leq |\sigma'^{*-1}(\discarded)|$.   To guarantee that our algorithm has a small running time, we need to break the problem into two separate sub problems with a few guesses. This is possible due to the following two properties. First, the window of a job $j$ in $J_{[T]}$, which is the set of jobs assigned to $[T]$ in $\bbJ$,  only depends on jobs assigned to the first $h$ levels of intervals in $\bfT$. Thus we only need to guess $g^*_I$ vectors for $I \in \calI_{<h}$ in order to define the windows of $J_{[T]}$. Second the windows for $J_{[T]}$ have boundaries being integer multiplies of $2^{-h}T$ and thus there are not too many possible windows. So we can afford to guess how to split $J_{[T]}$ into $(0, T/2]$ and $(T/2, T]$ and break the problem into two sub-problems.   

In Appendix~\ref{sec:improved} and \ref{sec:improved-algo}, we show how to improve the running time to $n^{O\left(\frac{m^4}{\epsilon^3}\log^3 \log n\right)}$ using a more careful guessing procedure.  The overview of the algorithm will be given at the beginning of Appendix~\ref{sec:improved}. 

We remark that many ingredients in our algorithm and analysis can also be found in Levey-Rothvoss \cite{LR16}; for example, the definition of window constraints, the ideas used to establish the connection between valid schedules and virtually-valid ones for a dyadic system are motivated by the techniques in \cite{LR16}. 

%% file: dyadic.tex
\section{Dyadic System}
	\label{sec:dyadic}
	In this section, we describe a core structure that our algorithm uses: (partial) {dyadic} systems. Some ingredients in the structure were used in Levey-Rothvoss \cite{LR16} and Garg \cite{Garg18}; for our algorithm and analysis, it is useful to define the structure explicitly. 
	\begin{definition}
		\label{def:dyadic-system}
		Given an interval ${I^*} \in \calI$, a \emph{{partial dyadic} system} $\bbJ$ over ${I^*}$ is a tuple $\big({J_\anc},  b^\anc \in [0, T]^{J_\anc}, e^\anc \in [0, T]^{J_\anc}, (J_I)_{I \in \calI[I^*]}\big)$ where
		\begin{enumerate}[label=(\ref{def:dyadic-system}\alph*)]
			\item ${J_\anc}$ and $J_I$'s are mutually disjoint subsets of $J^\circ$,  \label{property:dyadic-disjoint}
			\item for every $I \in \calI_\sftop[I^*]$,  we have $\Delta(J_I) \leq \delta|J_I| + \delta'|I|$, \label{property:dyadic-chain-length} 
			\item for every $I \in \calI_\sfmid[I^*]$, we have $J_I = \emptyset$, and \label{property:dyadic-middle-empty} 
			\item for every $I, I' \in \calI[I^*]$ with $I \inorderless I'$, there are no precedence constraints from $J_{I'}$ to $J_I$. \label{property:dyadic-precedence}
		\end{enumerate}		
		In the partial dyadic system $\bbJ$, we say jobs in $J_I$ \emph{are assigned to} the interval $I$, and $I$ is the \emph{owning interval} of jobs in $J_I$. The jobs assigned to $\calI_\sftop[I^*]$ are called \emph{top jobs}, and the jobs assigned to $\calI_\sfbot[I^*]$ are called \emph{bottom jobs}.  Notice that by Property~\ref{property:dyadic-middle-empty}, there are no middle jobs.
		Jobs in ${J_\anc}$ are called \emph{ancestor} jobs ($\anc$ stands for ``ancestor'').

		We simply say $\bbJ$ is a \emph{{dyadic} system} if additionally we have $I^* = [T], {J_\anc} = \emptyset$ and $J^\circ = J_{\subseteq [T]}$ (which is $\union_{I \in \calI}J_I$). We simply use $\bbJ = (J_I)_{I \in \calI}$ to denote a dyadic system.
	\end{definition}
		
	In this section we only focus on (non-partial) {dyadic} systems; we shall discuss partial ones when we need to use them.  In a {dyadic} system $\bbJ$, $(J_I)_{I \in \calI}$ form a partition of $J^\circ$ (Property~\ref{property:dyadic-disjoint} and that $J^\circ = J_{\subseteq [T]}$).  Property~\ref{property:dyadic-chain-length} requires that for a top interval $I \in \calI_\sftop$, the maximum chain length of jobs in $J_I$ is small. Property~\ref{property:dyadic-middle-empty} says that no jobs are assigned to middle levels.  
	Property~\ref{property:dyadic-precedence} requires that the sequence $(J_I)_{I \in \calI}$ according to the order $\inorderless$ respects the precedence constraints.


	\begin{definition}
		\label{def:valid}
		Given a {dyadic} system $\bbJ = (J_I)_{I \in \calI}$, a vector $\sigma \in \big({[T]} \cup \{\discarded\}\big)^{J^\circ}$  is said to be a valid schedule for $\bbJ$, if it satisfies the \textbf{capacity constraints}, \textbf{precedence constraints} as in Definition~\ref{def:original-valid}, and
		\begin{itemize}[leftmargin=*]
			\item \textbf{interval constraints}: for every $I \in \calI$ and $j \in J_I$, we have $\sigma_j \in I\cup \{\discarded\}$. 
		\end{itemize}		
	\end{definition}
 	So, for $\sigma$ to be valid schedule for a {dyadic} system, we additionally require each job $j$ to be scheduled inside its owning interval or discarded. 

\subsection{Dyadic System and Schedule from the Optimum Solution}
	\label{subsec:optimum-dyadic}
	In this section, we assume we are given an optimum valid schedule $\sigma^* \in [T]^{J^\circ}$ to the input instance (without discarded jobs). We shall construct a {dyadic} system $\bbJ^* = (J^*_I)_{I \in \calI}$ for which $\sigma^*$ is valid.  Notice that $\sigma^*$, $\bbJ^*$ and the procedure for constructing $\bbJ^*$ are only used in our analysis, instead of the algorithm. 

	In the recursive algorithm $\cstrJStar$ described in Algorithm~\ref{alg:cstrJstar},  we construct the {dyadic} system $\bbJ^* = (J^*_I)_{I \in \calI}$ for which the schedule $\sigma^*$ is valid. The algorithm also defines for every $I\in\calI$,  $K^*_I = J^*_{\subseteq I}$ to be the set of jobs assigned to sub-intervals of $I$ in the system $\bbJ^*$, and a vector $g^*_I \in \{\sfL, \sfR\}^*$ for every $I \in \calI_\sftop \cup \calI_\sfmid$.   Initially, we set $K^*_{[T]} = J^\circ$ and call $\cstrJStar([T])$. 
	\begin{algorithm}
		\caption{$\cstrJStar(I)$}
		\label{alg:cstrJstar}
		\begin{algorithmic}[1]
			\If{$I \in \calI_\sfbot$} $J^*_I \gets K^*_I$, \Return \EndIf \label{step:cstrJstar-bottom}
			\State $J^*_I \gets K^*_I,\ g^*_I \gets (),\ K^*_{\sfleft(I)} \gets \emptyset,\ K^*_{\sfright(I)} \gets \emptyset$, 
			define function $\len(x) = \begin{cases} 
				\delta x+\delta'|I| & \text{ if } I \in \calI_\sftop \\ 
				0 & \text{ if } I \in \calI_\sfmid
			\end{cases} $
			\label{step:cstrJstar-init}
			\While{$\Delta(J^*_I) > \len(|J^*_I|)$} \label{step:while}
				\State let $j \in J^*_I$ be a job with $|N^+_{J^*_I}(j)|,|N^-_{J^*_I}(j)| \geq \len(|J^*_I|)/2-1$, chosen in a deterministic way\label{step:cstrJstar-choose-j}
				\If{$\sigma^*_j \in \sfleft(I)$}
					\State append $\sfL$ to the vector $g^*_I$, and move $\{j\} \cup N^-_{J^*_I}(j)$ from $J^*_I$ to $K^*_{\sfleft(I)}$
				\Else \Comment{We must have $\sigma^*_j \in \sfright(I)$}
					\State append $\sfR$ to the vector $g^*_I$, and move $\{j\} \cup N^+_{J^*_I}(j)$ from $J^*_I$ to $K^*_{\sfright(I)}$
				\EndIf
			\EndWhile
			\State $\cstrJStar(\sfleft(I))$, $\cstrJStar(\sfright(I))$ \label{step:recurse}
		\end{algorithmic}
	\end{algorithm}

	At the beginning of any recursion $\cstrJStar(I)$, we have constructed the set $K^*_I$ and our goal is to assign $K^*_I$  to sub-intervals of $I$. It is guaranteed that all jobs in $K^*_I$ are scheduled in $I$ in $\sigma^*$. If $I \in \calI_\sfbot$, we set $J^*_I = K^*_I$ and return immediately (Step~\ref{step:cstrJstar-bottom}). Thus, we now assume $I \in \calI_\sftop \cup \calI_\sfmid$. Initially, all jobs are assigned to $I$ and thus we set $J^*_I = K^*_I$ and $K^*_{\sfleft(I)} = K^*_{\sfright(I)} = \emptyset$ (Step~\ref{step:cstrJstar-init}). We need to guarantee that $\Delta(J^*_I) \leq \delta|J^*_I| + \delta'|I|$ if $I \in \calI_\sftop$; when $I \in \calI_\sfmid$, we need $\Delta(J^*_I) = 0$, i.e, $J^*_I = \emptyset$.  This motivates the definition of the function $\len$. 
	Suppose at the beginning of some iteration in Loop~\ref{step:while}, we have $\Delta(J^*_I) > \len(|J^*_I|)$.  Then  there is a chain of jobs in $J^*_I$ of length at least  $\len(|J^*_I|)$. The bottom job $j$ in the chain has $|N^-_{J^*_I}(j)|, |N^+_{J^*_I}(j)| \geq \len(|J^*_I|)/2 - 1$ (this holds whenever $\len(|J^*_I|) \geq 0$).  Thus we can always find a job $j$ in Step~\ref{step:cstrJstar-choose-j} satisfying the condition.  Then we check whether $j$ is scheduled in $\sfleft(I)$ or $\sfright(I)$ in $\sigma^*$.  In the former case, all jobs in $N^-_{J^*_I}(j)$  are scheduled in $\sfleft(I)$ in $\sigma^*$ due to the precedence constraints and thus we can move $\{j\} \cup N^-_{J^*_I}(j)$  from $J^*_I$ to $K^*_{\sfleft(I)}$. Similarly in the latter case, all jobs in $N^+_{J^*_I}(j)$ are scheduled in $\sfright(I)$ in $\sigma^*$ and we can move $\{j\} \cup N^+_{J^*_I}(j)$  from $J^*_I$ to $K^*_{\sfright(I)}$.  Then the vector $g^*_I \in \{\sfL, \sfR\}^*$ will indicate whether each $j$ considered in the loop is scheduled in $\sfleft(I)$ or $\sfright(I)$.   
		
	So after the while loop, we are guaranteed that $\Delta(J^*_I) \leq \len(|J^*_I|)$, 
	$K^*_{\sfleft(I)}, K^*_{\sfright(I)}$ and $J^*_I$ form a partition of $K^*_I$, $K^*_{\sfleft(I)} \subseteq \sigma^{*-1}(\sfleft(I))$ and $K^*_{\sfright(I)} \subseteq \sigma^{*-1}(\sfright(I))$. If $I \in \calI_\sftop$, then $\Delta(J^*_I) \leq \delta|J^*_I| + \delta' |I|$ and if $I \in \calI_\sfmid$ then $J^*_I = \emptyset$.  Moreover, it is easy to see that during any moment in the while loop, the sequence $K^*_{\sfleft(I)}, J^*_I, K^*_{\sfright(I)}$ respects the precedence constraints: This is satisfied before the while loop, and it is maintained since whenever we move some $j$ from $J^*_I$ to $K^*_{\sfleft(I)}$, all its predecessors are moved, and whenever we move some $j$ from $J^*_I$ to $K^*_{\sfright(I)}$, all its successors are moved. Thus, $\bbJ^* = (J^*_I)_{I \in \calI}$ is indeed a {dyadic} system and $\sigma^*$ is a valid schedule for $\bbJ^*$ without discarded jobs.
	
	\begin{claim}
		\label{claim:number-guesses-small}
		Focus on a recursion of Algorithm~\ref{alg:cstrJstar} for some $I \in \calI_\sftop$. The number of iterations Loop~\ref{step:while} takes is at most ${p} = \floor{\frac{2}{\delta}\ln\frac{m}{\delta'} }+ 1$.  If $I \in \calI_\sfmid$, then the number is at most $m|I|$.
	\end{claim}
	\begin{proof}
		First consider the case where $I \in \calI_\sftop$. 
		In each iteration of the while loop, we move at least $(\delta|J^*_I| + \delta'|I|)/2 \geq  \delta|J^*_I|/2$ jobs out of $J^*_I$. Initially, $|J^*_I| \leq m|I|$ since they are scheduled in $I$ in $\sigma^*$.  At the beginning of the last iteration of the loop, we have $|J^*_I| \geq \delta'|I|$ since otherwise we would not have the loop. Thus, the number of iterations is at most $\floor{\log_{1/(1-\delta/2)}\frac{m|I|}{\delta'|I|}} + 1 = \floor{-\frac1{\ln{(1-\delta/2)}}\ln\frac{m}{\delta'}} + 1 \leq \floor{\frac{2}{\delta}\ln\frac{m}{\delta'}} + 1$, where the inequality is by that $\ln(1-x) < -x$ for every $x \in (0, 1)$.
		
		For the case where $I \in \calI_\sfmid$, in every iteration we moved at least 1 job out of $J^*_I$. Initially we have $J^*_I = K^*_I$ and thus the number of iterations is at most $|K^*_I| \leq m|I|$.
	\end{proof}
	The claim is crucial to our algorithm. In our actual algorithm we do not know $\sigma^*$. However, there are at most $2^{{p}}$ ($2^{m|I|}$, resp.) different vectors in $\{\sfL, \sfR\}^{{p}}$ ($\{\sfL, \sfR\}^{{m|I|}}$, resp.) and one of them must contain $g^*_I$ as a prefix.  Later our algorithm will guess the vector and run the while loop using the guess. This is the motivation of the procedure $\pushdown$ described in Algorithm~\ref{alg:pushdown}.  When calling the procedure, we guarantee that $I \in \calI_\sftop \cup \calI_\sfmid, K \subseteq J^\circ$, $|K| \leq m|I|$ and $g \in \{\sfL, \sfR\}^{{p}}$ if $I \in \calI_\sftop$ and $g \in \{\sfL, \sfR\}^{m|I|}$ if $I \in \calI_\sfmid$.
	\begin{algorithm}[h]
		\caption{$\pushdown(I, K, g)$}
		\label{alg:pushdown}		
		\textbf{Input}: $I \in \calI_\sftop \cup \calI_\sfmid, K\subseteq J^\circ, |K| \leq m|I|, g \in \{\sfL, \sfR\}^{{p}}$ if $I \in \calI_\sftop$ and $g \in \{\sfL, \sfR\}^{m|I|}$ if $I \in \calI_\sfmid$
		\begin{algorithmic}[1]
			\State $J \gets K, K_\sfL \gets \emptyset, K_\sfR \gets \emptyset, q \gets 1$, let function $\len(x) = \begin{cases} 
							\delta x+\delta'|I| & \text{ if } I \in \calI_\sftop \\ 
							0 & \text{ if } I \in \calI_\sfmid
						\end{cases} $
			\While{$\Delta(J) > \len(|J|)$}
				\State let $j \in J$ be a job with $|N^+_{J}(j)|,|N^-_{J}(j)| \geq \len(|J|)/2-1$, chosen using the same deterministic procedure as in Step~\ref{step:cstrJstar-choose-j} of Algorithm~\ref{alg:cstrJstar} \label{step:pushdown-choose-j} 
				\State \textbf{if} {$g_q = \sfL$} \textbf{then}
					move $\{j\} \cup N^-_{J}(j)$ from $J$ to $K_\sfL$
				\textbf{else} 
					move $\{j\} \cup N^+_{J}(j)$ from $J$ to $K_\sfR$
				\State $q \gets q + 1$
			\EndWhile
			\State \Return $(J, K_\sfL, K_\sfR)$
		\end{algorithmic}
	\end{algorithm}
	\begin{obs}
		\label{obs:pushdown}
		Assume $I \in \calI_\sftop \cup \calI_\sfmid, K\subseteq J^\circ, |K| \leq m|I|, g \in \{\sfL, \sfR\}^{{p}}$ if $I \in \calI_\sftop$ and $g \in \{\sfL, \sfR\}^{{m|I|}}$ if $I \in \calI_\sfmid$.  Assume $\pushdown(I, K, g)$ returns $(J, K_\sfL, K_\sfR)$. Then the following statements hold.
		\begin{enumerate}[label=(\ref{obs:pushdown}\alph*)]
			\item $J, K_\sfL$ and $K_\sfR$ form a partition of $K$. \label{property:pushdown-partition}
			\item If $I \in \calI_\sftop$, then 
			$\Delta(J) \leq  \delta|J|+\delta'|I|$. \label{property:pushdown-chain-length}
			\item If $I \in \calI_\sfmid$, then $J = \emptyset$. \label{property:pushdown-middle-empty}
			\item The sequence $K_\sfL, J, K_\sfR$ respects the precedence constraints. \label{property:pushdown-precedence}
			\item If $K = K^*_I$ and $g^*_I$ is a prefix of $g$, then $J = J^*_I, K_\sfL = K^*_{\sfleft(I)}$ and $K_\sfR = K^*_{\sfright(I)}$. \label{property:pushdown-returns-star}
		\end{enumerate}
	\end{obs}
	
	Although each $g^*_I$ is short, we can not afford to guess the combination of all $g^*_I$'s since there are too many intervals $I$. Later in each recursion of our algorithm, we  guess $g^*_I$'s only for a small set of intervals $I$.

%% file: virtual.tex
\section{Virtually-Valid Schedules for Dyadic Systems}
	\label{sec:virtual}
	As discussed in the introduction, to break an instance over some top interval $I^*$ into two sub-instances, we need to ignore the precedence constraints incident to jobs assigned to $I^*$. This leads to the definition of \emph{virtually-valid} schedules in Section~\ref{subsec:virtually-valid}, which in turn requires us to define a window $(b^\bbJ_j, e^\bbJ_j]$ for every top job $j$ in a partial dyadic system $\bbJ$ in Section~\ref{subsec:windows}. In the next Section (Section~\ref{sec:conversion}), we show a two-direction connection between valid schedules and virtually-valid ones for $\bbJ$. 
	
	To use mathematical inductions later in Section~\ref{sec:algo}, we need to define virtually-valid schedules for \emph{partial} dyadic systems.  Let us revisit Definition~\ref{def:dyadic-system}: We can treat a partial dyadic system as a dyadic system restricted to some interval $I \in \calI^*$, plus some ancestor jobs ${J_\anc} \subseteq J^\circ$, each $j \in {J_\anc}$ associated with a $b^\anc_j$ and $e^\anc_j$ value. We shall elaborate more on the set ${J_\anc}$ in Section~\ref{sec:algo}. For this section, it is only used in Definition~\ref{def:virtually-valid} and can be ignored in this section.  Till the end of this section, we fix a {partial dyadic} system $\bbJ = ({J_\anc},  b^\anc, e^\anc, (J_I)_{I \in \calI[I^*]})$ over some $I^* \in \calI$. All the definitions, claims and lemmas are w.r.t to this $\bbJ$. 
	
	\subsection{Windows for Top Jobs}
	\label{subsec:windows}
	\begin{definition}[$b^\bbJ$ and $e^\bbJ$ values for top jobs]
		\label{def:b-e}
	 	Given a top job $j\in J_I$ for some  $I \in \calI_\sftop[I^*]$, we define the window for $j$ in $\bbJ$ to be $(b^\bbJ_j, e^\bbJ_j]$, where
		\begin{itemize}[leftmargin=*]
			\item $b^\bbJ_j$ is the minimum integer multiply $b$ of $\max\{2^{-h}|I|, 2^h\}$ in $(\sfbegin(I), \sfcenter(I)]$ such that there are no precedence constraints from $J_{\subseteq(b, \sfcenter(I)]}$ to $j$, and
			\item $e^\bbJ_j$ is the maximum integer multiply $e$ of $\max\{2^{-h}|I|, 2^h\}$ in $[\sfcenter(I), \sfend(I))$ such that there are no precedence constraints from $j$ to $J_{\subseteq(\sfcenter(I),e]}$.
		\end{itemize}
	\end{definition}
	Notice that $b^\bbJ_j$ and $e^\bbJ_j$ are well-defined since $\sfcenter(I)$ is a candidate for both $b$ and $e$. The following claims are easy to prove:
	\begin{claim}
		\label{claim:simple-properties-of-b-e}
		For any top job $j \in J_I, I \in \calI_\sftop$, we have $\sfbegin(I) < b^\bbJ_j \leq \sfcenter(I)\leq e^\bbJ_j < \sfend(I)$. Moreover, there are no precedence constraints from $j$ to $J_{\subseteq (0, e^\bbJ_j]}$, or from $J_{\subseteq (b^\bbJ_j, T]}$ to $j$.
	\end{claim}
	\begin{proof}
		The first statement simply follows from the definitions of $b^\bbJ_j$ and $e^\bbJ_j$. To prove the second statement, we focus on any job $j' \in J_{I'}$ with $I' \subseteq (0, e^\bbJ_j]$. Then, either  $I'$ is disjoint from $I$ with $I' \inorderless I$, or $I' \subseteq \sfleft(I)$, or $I' \subseteq (\sfcenter(I), e^\bbJ_j] \subseteq \sfright(I)$.  In the first two cases, we have $j \not\prec j'$ by Property~\ref{property:dyadic-precedence}. In the third case, we have $j \not\prec j'$ by the definition of $e^\bbJ_j$.  Thus, there are no precedence constraints from $j$ to $J_{\subseteq(0, e^\bbJ_j]}$. Similarly we can show that there are no precedence constraints from $J_{\subseteq(b^\bbJ_j, T]}$ to $j$.		
	\end{proof}
	
	The following lemma shows that $b^\bbJ$ and $e^\bbJ$ values respect the precedence constraints.
	\begin{lemma}
		\label{lemma:b-e-respects-precedence}
		Let $j$ and $j'$ be two top jobs with $j \prec j'$.  Then we have $b^\bbJ_j \leq b^\bbJ_{j'}$ and $e^\bbJ_j \leq e^\bbJ_{j'}$. 
	\end{lemma}
	\begin{proof}
		Assume $j \in J_I$ and $j' \in J_{I'}$ for some $I, I' \in \calI_\sftop[I^*]$. If $I$ and $I'$ are disjoint, then the claim holds since $I$ must be to the left of $I'$ by Property~\ref{property:dyadic-precedence}, $b^\bbJ_j, e^\bbJ_j\in I$ and $b^\bbJ_{j'}, e^\bbJ_{j'} \in I'$. Now consider the case $I = I'$. Notice that $N^-_{J^\circ}(j) \subseteq N^-_{J^\circ}(j')$ and $N^+_{J^\circ}(j')\subseteq N^+_{J^\circ}(j)$. If there are no precedence constraints from $J_{\subseteq(b^\bbJ_{j'}, \sfcenter(I)]}$ to $j'$, then there are no such constraints to $j$ as well; thus $b^\bbJ_j \leq b^\bbJ_{j'}$.  If there are no precedence constraints from $j$ to $J_{\subseteq(\sfcenter(I),e^\bbJ_j]}$, then there are no such constraints from $j'$ as well; thus $e^\bbJ_{j'} \geq e^\bbJ_j$.
		
		Finally, we consider the case where one of the two intervals $I, I'$  is a strict sub-interval of the other. We only consider the case that $I \subsetneq I'$; the analysis for the other case is symmetric.  Since $I\inorderless I'$ by Property~\ref{property:dyadic-precedence}, we must have $I \subseteq \sfleft(I')$.  See Figure~\ref{fig:proof-of-time} for illustration of time points used in this proof.  Notice that $e^\bbJ_j <\sfend(I) \leq \sfcenter(I') \leq e^\bbJ_{j'}$.  Thus, it remains to prove that $b^\bbJ_j \leq b^\bbJ_{j'}$.
		\begin{figure}
			\centering
			\includegraphics[width=0.8\textwidth]{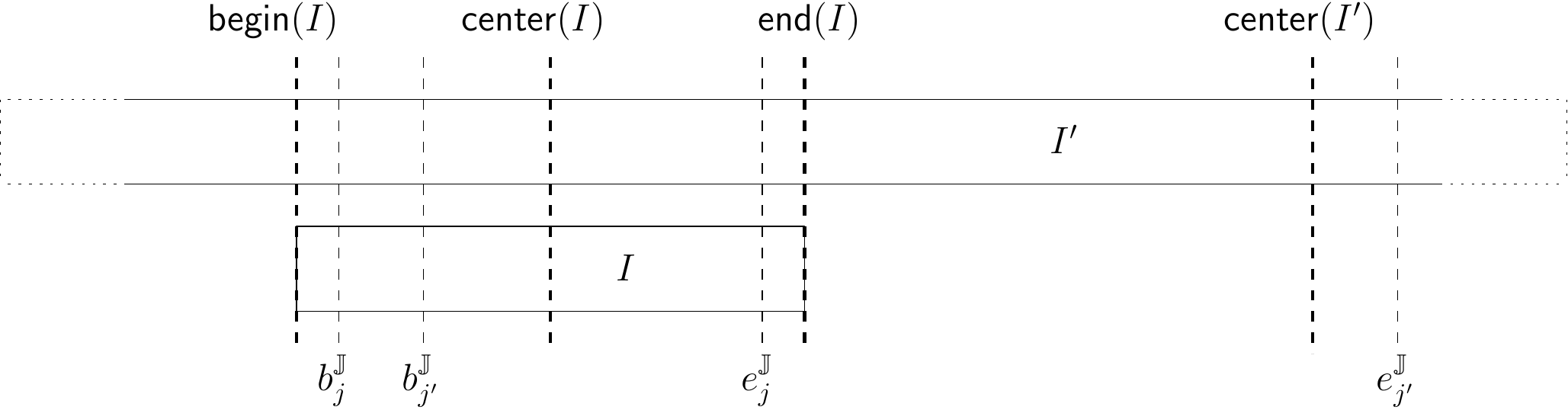}
			\caption{Time points used in the proof of Lemma~\ref{lemma:b-e-respects-precedence}.} \label{fig:proof-of-time}
		\end{figure}
				
		If $b^\bbJ_{j'} \geq \sfcenter(I)$, then we have $b^\bbJ_j \leq \sfcenter(I) \leq b^\bbJ_{j'}$ and we are done. So, assume $b^\bbJ_{j'} < \sfcenter(I)$.  Since $j \in J_I \subseteq J_{\subseteq (\sfbegin(I), \sfcenter(I')]}$ and $j \prec j'$, we have $b^\bbJ_{j'}  > \sfbegin(I)$ by its definition.\footnote{This holds regardless of whether $\sfbegin(I)$ is an integer multiply of $\max\{2^{-h}|I'|, 2^h\}$ or not.} So $b^\bbJ_{j'}$ is an integer multiply of $\max\{2^{-h}|I'|, 2^h\}$ strictly between $\sfbegin(I)$ and $\sfcenter(I)$. By the definition of $b^\bbJ_{j'}$, there are no precedence constraints from $J_{\subseteq (b^\bbJ_{j'}, \sfcenter(I')]}$ to $j'$. Since $J_{\subseteq(b^\bbJ_{j'}, \sfcenter(I')]} \supseteq J_{\subseteq(b^\bbJ_{j'}, \sfcenter(I)]}$, there will be no precedence constraints from $J_{\subseteq (b^\bbJ_{j'}, \sfcenter(I)]}$ to $j'$, implying that there will be no such constraints to $j$ as well.   As $b^\bbJ_{j'}$ is an integer multiply of $\max\{2^{-h}|I|, 2^h\}$ strictly between $\sfbegin(I)$ and $\sfcenter(I)$, we have that $b^\bbJ_j \leq b^\bbJ_{j'}$, by the definition of $b^\bbJ_j$. 
	\end{proof} 
	
\subsection{Virtually-Valid Schedules}
	\label{subsec:virtually-valid}
	With the windows for top jobs defined, we can now define what is a \emph{virtually-valid} schedule: 
	\begin{definition}
		\label{def:virtually-valid}
			We say $\sigma \in \big(I^* \cup \{\discarded\}\big)^{J_{\subseteq I^*} \cup {J_\anc}}$ is a \emph{virtually-valid schedule} for $\bbJ$ if it satisfies
		\begin{itemize}
			\item {\bf capacity constraints}: they are the same as in Definition~\ref{def:original-valid},
			\item {\bf precedence constraints for bottom jobs}: for every $j, j'\in \union_{I \in \calI_\sfbot[I^*]}J_I \setminus \sigma^{-1}(\discarded)$, we have $\sigma_j < \sigma_{j'}$,
			\item {\bf interval constraints for bottom jobs}: for every $j \in I, I \in \calI_\sfbot[I^*]$, we have $\sigma_j \in I \cup \{\discarded\}$,
			\item {\bf window constraints for top jobs}: for every $j \in \union_{I \in \calI_\sftop[I^*]}J_I$, we have $\sigma_j \in (b^\bbJ_j, e^\bbJ_j] \cup \{\discarded\}$, and
			\item {\bf window constraints for ancestor jobs}: for every $j \in {J_\anc}$, we have $\sigma_j \in (b^\anc_j, e^\anc_j] \cup \{\discarded\}$.
		\end{itemize}
	\end{definition}
		Again for intuition we first assume $\bbJ$ is a (non-partial) dyadic system; then the last set of constraints hold trivially. Compared to a valid one, in a virtually-valid schedule, we ignore the precedence constraints incident on top jobs. Instead, we require each top job $j$ is scheduled within its window $(b^\bbJ_j, e^\bbJ_j]$. Since $(b^\bbJ_j, e^\bbJ_j]$ is a sub-interval of the owning interval of $j$, the window constraint for a  $j$ implies the interval constraint for it;  thus we could require interval constraints to hold for both top and bottom jobs.  We remark that in spite of  the term ``virtually'', a valid schedule is not necessarily virtually-valid, as it may not satisfy the window constraints for top jobs.
		
		The first four sets of constraints in Definition~\ref{def:virtually-valid} naturally extend to partial dyadic systems.  The last set requires that each ancestor job $j\in {J_\anc}$ is either scheduled with its window $(b^\anc_j, e^\anc_j]$ or discarded.

%% file: conversion.tex
\section{Conversions between Valid and Virtually-Valid Schedules}
\label{sec:conversion}
	In this section, we show that valid and virtually-valid schedules can be converted to each other, up to discarding a few jobs.  In Section~\ref{subsec:valid-to-virtually-valid}, we show that the valid schedule $\sigma^*$ for $\bbJ^*$ can be converted to a virtually-valid schedule $\sigma'^*$ with a small number of discarded jobs. In Section~\ref{subsec:virtually-valid-to-valid}, we show that a virtually-valid schedule $\sigma''$ for any {dyadic} system $\bbJ$ can be converted to a valid schedule $\sigma$, with a small number of extra  jobs discarded. In this section, we only need to consider (non-partial) dyadic systems.
	\subsection{From $\sigma^*$ to a Virtually Valid Schedule $\sigma'^*$}	
	\label{subsec:valid-to-virtually-valid}
	 For convenience, from now on, we use $b^*_j$ to denote $b^{\bbJ^*}_j$ and $e^*_j$ to denote $e^{\bbJ^*}_j$  for every top job $j$ in $\bbJ^*$. The following simple claim is needed in our analysis.
	\begin{claim}
		\label{lemma:in-extended-window}
		For every top job $j \in J^*_I$ for some $I \in \calI_\sftop$, we have $\sigma^*_j\in \big(b^*_j - \max\{2^{-h}|I|, 2^h\}, e^*_j + \max\{2^{-h}|I|, 2^h\}\big]$.
	\end{claim} 
	\begin{proof}
		We prove that $\sigma^*_j > b^*_j - \max\{2^{-h}|I|,2^h\}$.  By the definition of $b^*_j$, either $b^*_j = \sfbegin(I) + \max\{2^{-h}|I|,2^h\}$, or $b^*_j > \sfbegin(I) + \max\{2^{-h}|I|,2^h\}$ and there is a job $j' \in J^*_{(b^*_j - \max\{2^{-h}|I|,2^h\}, \sfcenter(I)]}$ such that $j' \prec j$. In the former case, $\sigma^*_j > \sfbegin(I) = b^*_j - \max\{2^{-h}|I|,2^h\}$.   In the latter case, we must have $\sigma^*_{j'} > b^*_j - \max\{2^{-h}|I|, 2^h\}$. Since $j' \prec j$, we have $\sigma^*_j > b^*_j - \max\{2^{-h}|I|, 2^h\}$. Similarly, we can show that $\sigma^*_j\leq e^*_j + \max\{2^{-h}|I|,2^h\}$. 
	\end{proof}
	
	 The goal of this section is to prove the following lemma:
	\begin{lemma}
		\label{lemma:valid-to-virtually-valid}
		There is a virtually-valid schedule $\sigma'^*$ for $\bbJ^*$ with at most $\frac{3\epsilon T}{4}$ jobs discarded.
	\end{lemma}
	\begin{proof} 
		In ${\sigma'^*}$, we schedule the bottom jobs in exactly the same way as they are in $\sigma^*$:  For every bottom  job $j$ in $\bbJ^*$, we have $\sigma'^*_j = \sigma^*_j$.  Then in $\sigma'^*$ the bottom jobs satisfy interval and precedence constraints. It remains to show how to schedule  top jobs  to satisfy capacity and window constraints. (Notice that $\bbJ^*$ is non-partial and has ${J_\anc} =\emptyset$.) To this end, we fix some $I \in \calI_\sftop$ and focus on the set $J^*_I$. We schedule $J^*_I \cap \sigma^{*-1}(\sfleft(I))$ in $\sfleft(I)$ and $J^*_I \cap \sigma^{*-1}(\sfright(I))$ in $\sfright(I)$ in the schedule $\sigma'^*$. That means, we do not change the sides of top jobs when converting $\sigma^*$ to $\sigma'^*$.  We only show how to schedule $J^*_I \cap \sigma^{*-1}(\sfleft(I))$ since the set $J^*_I \cap \sigma^{*-1}(\sfright(I))$ can be handled symmetrically.
		
		In ${\sigma'^*}$ we schedule $J^*_I  \cap {\sigma^{*-1}}(\sfleft(I))$, using the slots allocated for them in $\sigma^*$.  In other words, at any time $t \in \sfleft(I)$, we say there are $\sfcap(t) := |J^*_I  \cap {\sigma^{*-1}}(t) \big|$ available slots. In $\sigma'^*$,  we only schedule $J^*_I \cap {\sigma^{*-1}}(\sfleft(I))$ using the available slots.  Let $\calC$ be the partition of $\sfleft(I)$ by integer multiplies of $\max\{2^{-h}|I|, 2^h\}$: For $I \in \calI_{\leq L-h}$, we have $\calI = \calI_{h-1}[\sfleft(I)]$ and for $I \in \calI_{ > L-h}$, we have $\calI = \calI_\sfbot[\sfleft(I)]$.	For every  $I' \in \calC$, let $\sfcap(I') := \sum_{t \in I'}\sfcap(t) =  \big|J^*_I  \cap {\sigma^{*-1}}(I')\big|$ be the number of available slots in $I'$.  
		
		In ${\sigma'^*}$, we schedule $J^*_I  \cap {\sigma^{*-1}}(\sfleft(I))$, via a simple procedure. Initially, let $\tilde J \gets \emptyset$. For every $I' \in \calC$ from left to right, we do the following: schedule $\min\{|\tilde J|, \sfcap(I')\}$  jobs in $\tilde J$ using the $\sfcap(I')$ available slots in $I'$, remove the scheduled  jobs from $\tilde J$, and add the $\sfcap(I')$ jobs $J^*_I  \cap {\sigma^{*-1}}(I')$ to $\tilde J$. Then we discard $\tilde J$ in the end.  Notice that if some $j \in J^*_I  \cap {\sigma^{*-1}}(I'), I' \in \calC$ is not discarded by ${\sigma'^*}$, then it must be scheduled in some $I'' \in \calC$ to the right of $I'$. Notice that $e^*_j  \geq \sfcenter(I) \geq \sfend(I'')$.  By Lemma~\ref{lemma:in-extended-window}, we have $b^*_{j} - \max\{2^{-h}|I|,2^h\} \leq \sfbegin(I')$, which implies $b^*_{j} \leq \sfbegin(I'')$.  So the window constraint for $j$ is satisfied.
		
		To count the number of discarded jobs, we consider the change of $|\tilde J|$ during the process. If $|\tilde J| \geq \sfcap(I')$ at the beginning of iteration $I'$, then $|\tilde J|$ will not change in the iteration. Otherwise, $|\tilde J|$ will be changed to $\sfcap(I')$.  Thus, in the end, we have that $|\tilde J|$ is the maximum of $\sfcap(I')$ over all $I' \in \calC$, which is at most $\max\{2^{-h}|I|,2^h\}m$ since $\sigma^*$ is a valid schedule for $\bbJ^*$. Thus we showed that we discarded at most $\max\{2^{-h}|I|,2^h\}m$  jobs in $J^*_I  \cap \sigma^{*-1}(\sfleft(I))$; using a similar procedure, we can schedule  $J^*_I  \cap \sigma^{*-1}(\sfright(I))$ with at most $\max\{2^{-h}|I|,2^h\}m$ jobs discarded. Thus we discarded at most $\max\{2^{1-h}|I|,2^{h+1}\}m \leq 2^{1-h}m|I| + 2^{h+1}m$  jobs in $J^*_{I}$.  
		
		So the total number of jobs we discarded is at most $\sum_{I \in \calI_\sftop}2^{1-h}|I|m + \sum_{I \in \calI_\sftop}2^{h+1}m$. The first term is $(L-h') 2^{1-h}Tm \leq L2^{1-h}Tm \leq 2mLT \cdot \frac{\epsilon}{8m \log T} \leq \frac{\epsilon T}{4}$. The second term is exactly $2^{h+1}m (2^{L-h'}-1) \leq 2\cdot 2^{L+h-h'}m = 2mT \cdot 2^{-h'}  \leq 2mT \cdot \frac{\epsilon}{4m} = \frac{\epsilon T}{2}$.  Thus we discarded at most $\frac{3\epsilon T}{4}$ jobs overall. 	\end{proof}

\subsection{Converting a Virtually-Valid Schedule to a Valid One}
	\label{subsec:virtually-valid-to-valid}
	Throughout this section, we focus on one {dyadic} system $\bbJ = (J_I)_{I \in \calI}$ and  a virtually-valid schedule $\sigma''$ for $\bbJ$. We show that $\sigma''$ can be efficiently converted to a valid schedule $\sigma$ with a few extra discarded jobs. This is done in two steps.  In the first step, we convert $\sigma''$ into another virtually-valid schedule $\sigma'$ with some good properties, then we convert $\sigma'$ to a valid schedule $\sigma$. The two steps are captured by Lemma~\ref{lemma:canonicalize} and \ref{lemma:virtually-valid-to-valid} respectively. 
	
	Some definitions are needed in order to describe and prove the two lemmas.  We define $J^{\#}$ to be the set of top jobs that are scheduled in $\sigma''$.  These are the set of interesting jobs: The bottom jobs in $\bbJ$ and discarded jobs of $\sigma''$ are handled in the same way in $\sigma'$ and $\sigma$ as they are in $\sigma''$. Moreover, our schedule $\sigma'$ does not discard extra jobs: $\sigma'^{-1}(\discarded) = \sigma''^{-1}(\discarded)$. $\sigma$ may discard extra jobs. For every $I \in \calI_{\sftop}$, we define $J^{\#}_I = J^{\#}  \cap J_I$ to be the set of jobs in $J_I$ that are scheduled in $\sigma''$.   For every $j \in J^{\#}_I$, we define $I(j) = I$ to be its owning interval.
	
	Now we assume we are given any virtually-valid schedule $\sigma'$ for $\bbJ$ with $\sigma'^{-1}(\discarded) = \sigma''^{-1}(\discarded)$, and $\sigma'_j = \sigma''_j$ for every bottom job $j$ in $\bbJ$. We make some definitions that depend on $\sigma'$. For any $j \in J^{\#}$, we define $\side_{\sigma'}(j) = \sfL$ if $\sigma'_j \in \sfleft(I(j))$ and $\side_{\sigma'}(j) = \sfR$ if $\sigma'_j \in \sfright(I(j))$.  We then define a partial order $\sigmaless$ over $J^{\#}$, which depends on the function $\side_{\sigma'}$. We have $j \sigmaless j'$ if and only if $I(j) = I(j') = I$ for some $I \in \calI_\sftop$, $\side_{\sigma'}(j) = \side_{\sigma'}(j')$ and the following happens. 
	\begin{itemize}
		\item If $\side_{\sigma'}(j) = \side_{\sigma'}(j') = \sfL$, then $\langle b^\bbJ_j, \depth_{J^{\#}_{I}}(j)\rangle < \langle b^\bbJ_{j'}, \depth_{J^{\#}_{I}}(j')\rangle$,
		\item If $\side_{\sigma'}(j) = \side_{\sigma'}(j') = \sfR$, then $\langle e^\bbJ_j, \depth_{J^{\#}_{I}}(j)\rangle < \langle e^\bbJ_{j'}, \depth_{J^{\#}_{I}}(j')\rangle$.
	\end{itemize}
	
	So we group jobs $j$ in $J^{\#}$ according to their $I(j)$ and $\side_{\sigma'}(j)$ values and only jobs in a same group are comparable by $\sigmaless$.  Within a same group, the jobs $j$  are compared using $\langle b_j, \depth_{J^{\#}_{I(j)}}(j) \rangle$ or $\langle e_j, \depth_{J^{\#}_{I(j)}}(j) \rangle$, depending on whether the group is a left group or a right group. 
%
	%
	With the notations defined, we can now give our Lemma~\ref{lemma:canonicalize}.	It says that we can find a schedule $\sigma'$ which ``weakly'' respects the $\prec$ and $\sigmaless$ order:
	\begin{lemma}
		\label{lemma:canonicalize}
		We can efficiently find another virtually-valid schedule $\sigma'$ for $\bbJ$ with $\sigma'^{-1}(\discarded) = \sigma''^{-1}(\discarded)$ and $\sigma'_j = \sigma''_j$ for every bottom job $j$ in $\bbJ$. Moreover, for every two jobs $j, j' \in J^{\#}$, the following holds. 
		\begin{enumerate}[label=(\ref{lemma:canonicalize}\alph*)]
			\item If $j \prec j'$, then $\sigma'_{j} \leq \sigma'_{j'}$.
			\item If $j \sigmaless j'$, then $\sigma'_{j} \leq \sigma'_{j'}$.
		\end{enumerate}
	\end{lemma}
		

	\begin{proof}
		Let $\sigma' = {\sigma''}$ initially.  While there exist some $j, j'$ that do not satisfy one of the two conditions, we swap $\sigma'_j$ and $\sigma'_{j'}$. This makes the condition satisfied, without breaking the window constraints for $j$ and $j'$:
		\begin{itemize}
			\item If (\ref{lemma:canonicalize}a) is not satisfied, then by Lemma~\ref{lemma:b-e-respects-precedence} and the window constraints for $j$ and $j'$ in $\sigma'$, we have $b^\bbJ_j \leq b^\bbJ_{j'} < \sigma'_{j'} < \sigma'_j \leq e^\bbJ_j \leq e^\bbJ_{j'}$. After swapping we still have $\sigma'_j \in (b^\bbJ_j, e^\bbJ_j]$ and $\sigma'_{j'} \in (b^\bbJ_{j'}, e^\bbJ_{j'}]$.
			\item Suppose (\ref{lemma:canonicalize}b) is not satisfied.  Assume w.l.o.g that $j, j' \in J^{\#}_I$ for some $I \in \calI_\sftop, \sigma'_j, \sigma'_{j'} \in \sfleft(I)$ and $\sigma'_j > \sigma'_{j'}$. The case $\sigma'_j, \sigma'_{j'} \in \sfright(I)$ can be handled in the same way. Since $j \sigmaless j'$, we have $b^\bbJ_j \leq b^\bbJ_{j'} < \sigma'_{j'} < \sigma'_j \leq \sfcenter(I) \leq \min\{e^\bbJ_j, e^\bbJ_{j'}\}$ and after swapping we still have $\sigma'_{j} \in (b^\bbJ_j, e^\bbJ_j]$ and $\sigma'_{j'} \in (b^\bbJ_{j'}, e^\bbJ_{j'}]$.
		\end{itemize}
				
		It is trivial that the swapping operations do not violate capacity constraints, precedence and interval constraints for bottom jobs.  So it remains to show that the procedure of swapping operations will terminate. This is done by carefully defining a vector $\dif := \big(\dif_1, \dif_2, \dif_3\big)$ for $\sigma'$ and showing that its lexicographic rank strictly decreases after each swapping. $\dif_1, \dif_2$ and $\dif_3$ are defined as follows:
		\begin{align*}
			\dif_1  &:= \sum_{j \in J^{\#}} |I(j)| \cdot \left|\sigma'_j - \sfcenter(I(j))\right|,\\
			\dif_2 &:= \text{number of inversions in funciton $\mathsf{side}_{\sigma'}$ w.r.t the partial order}  \prec,\\
			\dif_3 &:= \text{number of inversions in function $\sigma'$ w.r.t the partial order} \sigmaless.
		\end{align*}
		Notice that the $\side_{\sigma'}$ function maps $J^{\#}$ to $\{\sfL, \sfR\}$. When defining $\dif_2$, we treat $\sfL$ as $0$ and $\sfR$ as $1$ and so $\sfL < \sfR$. 
		
		First assume that $j, j'$ do not satisfy condition (\ref{lemma:canonicalize}a).   If $j$ and $j'$ are assigned to two disjoint intervals, then the window constraints for $j$ and $j'$ will guarantee (\ref{lemma:canonicalize}a).   So we assume $I(j)$ and $I(j')$ overlap, $j \prec j'$ and $\sigma'_{j} > \sigma'_{j'}$. We consider three different cases.  
 		\begin{enumerate}[label=(\Alph*), leftmargin=*]
			\item $I(j) \subseteq \sfleft(I(j'))$. By Lemma~\ref{lemma:b-e-respects-precedence}, we have $b^\bbJ_j \leq b^\bbJ_{j'} < \sigma'_{j'} < \sigma'_j \leq e^\bbJ_j < \sfend(I(j)) \leq \sfcenter(I(j')) \leq e^\bbJ_{j'}$.  Then, swapping $\sigma'_{j'}$ and $\sigma'_j$ will decrease $\big|\sigma'_{j'} - \sfcenter(I(j'))\big|$ by $|\sigma'_{j'} - \sigma'_j|$, and increase $\big|\sigma'_j - \sfcenter(I(j))\big|$ by at most $|\sigma'_{j'} - \sigma'_j|$. Thus $\dif_1$ will decrease since $ |I(j')| > |I(j)|$.
			\item  $I(j') \subseteq \sfright(I(j))$. The analysis is symmetric to that for (A). We have $b^\bbJ_j \leq \sfcenter(I(j))\leq \sfbegin(I(j')) < b^\bbJ_{j'} < \sigma'_{j'} < \sigma'_j \leq e^\bbJ_j \leq e^\bbJ_{j'}$. The swap decreases $|\sigma'_j - \sfcenter(I(j))|$ by $|\sigma'_j - \sigma'_{j'}|$ and increases $\big|\sigma'_{j'} - \sfcenter(I(j'))\big|$ by at most $|\sigma'_j - \sigma'_{j'}|$. Since $|I(j)|> |I(j')|$, $\dif_1$ will decrease.   
			\item $I(j) = I(j') = I$ for some $I \in \calI_\sftop$. In this case, swapping $\sigma'_j$ and $\sigma'_{j'}$ does not change $\dif_1$. Consider three cases.
			\begin{enumerate}[label=(C\roman*),leftmargin=20pt]
				\item $\sigma'_{j'} \in \sfleft(I)$ and $\sigma'_j \in \sfright(I)$.  In this case, $j \prec j'$, $\side_{\sigma'}(j) = \sfR$ and $\side_{\sigma'}(j')= \sfL$.  Swapping $\sigma'_j$ and $\sigma'_{j'}$ will swap $\side_{\sigma'}(j)$ and $\side_{\sigma'}(j')$.  Thus, by Lemma~\ref{lemma:reduce-inversions}, $\dif_2$ will strictly decrease.
				
				\item $\sigma'_{j}, \sigma'_{j'}\in \sfleft(I)$. In this case, we have $b^\bbJ_j \leq b^\bbJ_{j'}$ by Lemma~\ref{lemma:b-e-respects-precedence} and $\depth_{J^{\#}_{I}}(j) < \depth_{J^{\#}_{I}}(j')$, which implies $\langle b^\bbJ_j, \depth_{J^{\#}_{I}}(j)\rangle  < \langle b^\bbJ_{j'}, \depth_{J^{\#}_{I}}(j')\rangle$. Thus, we have $j\sigmaless j'$ and  this case will be covered by Condition~(\ref{lemma:canonicalize}b). 
				
				\item $\sigma'_{j}, \sigma'_{j'}\in \sfright(I)$. Similarly in this case we have $\langle e^\bbJ_j, \depth_{J^{\#}_{I}}(j)\rangle  < \langle e^\bbJ_{j'}, \depth_{J^{\#}_{I}}(j')\rangle$. So, we have $j\sigmaless j'$ and  the case  will be covered by Condition~(\ref{lemma:canonicalize}b).
			\end{enumerate}
		\end{enumerate}
		
		Now we assume Condition~(\ref{lemma:canonicalize}b) is not satisfied. That is $j \sigmaless j'$ and $\sigma'_j > \sigma'_{j'}$. So, $I(j) = I(j')$ and $\side_{\sigma'}(j) =\side_{\sigma'}(j')$.  Then swapping $\sigma'_j$ and $\sigma'_{j'}$ does not change $\dif_1, \dif_2$. Also it does not change the relation $\sigmaless$ itself since the $\side_{\sigma'}$ function is unchanged. By Lemma~\ref{lemma:reduce-inversions}, $\dif_3$ will strictly decrease.		
		
		Thus, we proved that each swapping operation will decrease the lexicographic rank of the vector $\dif$. Since $\dif_1, \dif_2$ and $\dif_3$ are non-negative integers upper bounded by $\textrm{poly}(n)$, the procedure will terminate in $\textrm{poly}(n)$ iterations.  Thus, the procedure is efficient and eventually the virtually-valid schedule $\sigma'$ will satisfy (\ref{lemma:canonicalize}a) and (\ref{lemma:canonicalize}b).
	\end{proof}
	
	The second lemma shows how to convert $\sigma'$ to a valid schedule $\sigma$ for $\bbJ$. 
	\begin{lemma}
		\label{lemma:virtually-valid-to-valid}
		Given a virtually-valid schedule $\sigma'$ for $\bbJ$ satisfying  properties in Lemma~\ref{lemma:canonicalize}, we can efficiently construct a valid schedule $\sigma$ for $\bbJ$ with $|\sigma(\discarded) \setminus \sigma'(\discarded)| \leq \frac{\epsilon T}{4}$.
	\end{lemma}
	
	\begin{proof}
		As we mentioned, $\sigma$ discards all jobs in $\sigma'^{-1}(\discarded) = \sigma''^{-1}(\discarded)$ and schedule every  bottom job $j$ at time $\sigma'_j = \sigma''_j$.  So $\sigma$ satisfies the precedence and interval constraints for bottom jobs.  Also, $\sigma$ will not change the scheduling bottom intervals of jobs in $J^{\#}$: For every $j \in J^{\#}$, if $\sigma'_j \in I'$ for some $I' \in \calI_\sfbot$, we must have $\sigma_j \in I' \cup \{\discarded\}$.  Any such $\sigma$ will satisfy interval constraints for top jobs, and precedence constraints  between top and bottom jobs, as well as precedence constraints between top jobs scheduled in different bottom intervals. To see this focus on any job $j \in J^{\#}$ with $\sigma'_j, \sigma_j \in I' \in \calI_\sfbot$. 
		\begin{itemize}
			\item We must have $\sigma'_j \in I' \subseteq (b^\bbJ_j, e^\bbJ_j] \subseteq I(j)$. $I' \subseteq (b^\bbJ_j, e^\bbJ_j]$ holds since $b^\bbJ_j$ and $e^\bbJ_j$ are multiplies of $2^h$. Thus $\sigma_j \in I(j)$ and the interval constraint for $j$ is satisfied.
		
			\item Notice that $I' \subseteq (b^\bbJ_j, e^\bbJ_j]$.   By Claim~\ref{claim:simple-properties-of-b-e}, there are no precedence constraints from $j$ to $J_{\subseteq (0, e_j]} \supseteq J_{\subseteq (0, \sfend(I')]}$, and there are no precedence constraints from $J_{\subseteq (b_j, T]} \supseteq J_{\subseteq (\sfbegin(I'),T]}$ to $j$.  So, if $j$ is scheduled in $I'$, then all the precedence constraints between $j$ and bottom jobs are satisfied.   
			
			\item By Property (\ref{lemma:canonicalize}a) for $\sigma'$, for every $I', I'' \in \calI_\sfbot$ with $I' \inorderless I''$, there are no precedence constraints from $\sigma'^{-1}(I'') \cap J^{\#}$ to $\sigma'^{-1}(I') \cap J^{\#}$.  Thus, $\sigma$ will satisfy the precedence constraints between top jobs scheduled in different bottom intervals.		
		\end{itemize}
		
		Thus it suffices for us to guarantee the precedence constraints among top jobs scheduled in the same bottom interval in $\sigma$, while maintaining the capacity constraints. For each $I'\in \calI_\sfbot$, we schedule jobs in $J^{\#} \cap \sigma'^{-1}(I')$ in $I'$ in $\sigma$ using Lemma~\ref{lemma:list-scheduling}: 
		By the lemma, we only need to discard at most $m\Delta(J^{\#} \cap \sigma'^{-1}(I'))$ jobs in $J^{\#} \cap \sigma'^{-1}(I')$. Thus, the total number of discarded jobs is at most $m\sum_{I'\in \calI_\sfbot}\Delta(J^{\#} \cap \sigma'^{-1}(I'))$. Therefore, the remaining task in the proof is to show
		\begin{align}
			\sum_{I'\in \calI_\sfbot}\Delta(J^{\#} \cap \sigma'^{-1}(I')) \leq \frac{\epsilon T}{4m}. \label{inequ:total-discard}
		\end{align}
		
		We fix $I \in \calI_\sftop$ and focus on the set $J^{\#}_I\cap\sigma'^{-1}(\sfleft(I))$ of jobs in $J^{\#}_I$ scheduled in $\sfleft(I)$ in $\sigma'$.   By Property~(\ref{lemma:canonicalize}b) of $\sigma'$, the jobs $j$ in the set are scheduled in non-decreasing order of $\langle b^\bbJ_j, \depth_{J^{\#}_{I}}(j)\rangle$. 
		Applying Lemma~\ref{lemma:length-add-up} with $J$ being $J^{\#}_I, c(j)$ being $b^\bbJ_j$, and the sequence being $(J^{\#}_I \cap \sigma'^{-1}(I'))_{I' \in \calI_\sfbot[\sfleft(I)]}$, we have 
		\begin{align}
			&\quad \sum_{I' \in \calI_\sfbot[\sfleft(I)]}\Delta(J^{\#}_I\cap \sigma'^{-1}(I'))
			\leq 2^{h-1}\Delta(J^{\#}_I) +\big|\calI_\sfbot[\sfleft(I)]\big|
			\leq 2^{h-1}(\delta|J_I|+\delta'|I|) + 2^{-{h-1}}|I|,
			\label{inequ:Delta-left}
		\end{align}
		where the first inequality is by Lemma~\ref{lemma:length-add-up} and that there are at most $2^{h-1}$ different values in $\{b^\bbJ_j: j \in J^{\#}_I\}$ (all the values are integer multiplies of $2^{-h}|I|$ in $(\sfbegin(I), \sfcenter(I)]$), and the second inequality is by that $\Delta(J^{\#}_I) \leq \Delta(J_I) \leq \delta|J_I|+\delta'|I|$, and $\big|\calI_\sfbot[\sfleft(I)]\big| = 2^{-h-1}|I|$.
		
		Using a similar argument for right jobs in $J^{\#}_I$ w.r.t $\sigma'$, we can show that  $\sum_{I' \in \calI_\sfbot[\sfright(I)]}\Delta(J^{\#}_I\cap \sigma'^{-1}(I'))\leq 2^{h-1}(\delta|J_I|+\delta'|I|) + 2^{-{h-1}}|I|$. 
		Together the two inequalities imply 
		\begin{align}
			\sum_{I' \in \calI_\sfbot[I]}\Delta(J^{\#}_I\cap \sigma'^{-1}(I'))\leq 2^h(\delta|J_I|+\delta'|I|) + 2^{-h}|I|. \label{inequ:discard-for-I}
		\end{align}
	
		Now summing up the bound over all $I \in \calI_\sftop$, we have 
		\begin{align*}
			&\quad \sum_{I' \in \calI_\sfbot}\Delta\left(J^{\#}\cap \sigma'^{-1}(I')\right) 
			\leq \sum_{I \in \calI_\sftop, I' \in \calI_\sfbot[I]}\Delta(J^{\#}_I \cap \sigma'^{-1}(I'))\leq \sum_{I \in \calI_\sftop}\left(2^h(\delta|J_I|+\delta'|I|) + 2^{-h}|I|\right)\\
			&\leq 2^h\delta mT + \sum_{\ell = 0}^{L-h' - 1}\sum_{I \in \calI_\ell}\left(2^h\delta'|I| + 2^{-h}|I|\right)
			= 2^h\delta mT + \sum_{\ell=0}^{L-h'-1}(2^h\delta'+2^{-h})T \leq 2^h\delta mT+ (2^h\delta'+2^{-h})LT\\
			&= 2^h \cdot \frac{\epsilon}{16\cdot 2^hm^2} mT + \frac32\cdot 2^{-h}LT \leq \frac{\epsilon T}{16m} + \frac32 \cdot \frac{\epsilon}{8m\log T}LT \leq \frac{\epsilon T}{16m} + \frac{3\epsilon T}{16m} = \frac{\epsilon T}{4m}.
		\end{align*}
		The first inequality in the first line used the subadditivity of $\Delta$, and the second inequality is by Inequality~\eqref{inequ:discard-for-I}. The first inequality in the second line used that $J_I$'s for top intervals $I$ are disjoint and $|J^\circ|\leq mT$. The inequalities in the third line used the definitions of $h, \delta$ and $\delta'$. This finishes the proof of \eqref{inequ:total-discard}.
	\end{proof}

%% file: algo.tex
\section{Construction of Dyadic System and Virtually-Valid Schedule}
\label{sec:algo}
	With the (partial) dyadic systems and virtually-valid schedules defined, we can now describe the final algorithm for the scheduling problem. By Lemma~\ref{lemma:canonicalize} and \ref{lemma:virtually-valid-to-valid},  it suffices for us to construct a {dyadic} system $\bbJ^\best$ and a virtually-valid schedule $\sigma^\best$ for it with a small number of discarded jobs.  The section is organized as follows: We give the recursive algorithm in Section~\ref{subsec:algo-recursive} and analyze its correctness and running time in Section~\ref{subsec:algo-correctness} and \ref{subsec:algo-running-time} respectively.

	Now it is the time to discuss the set ${J_\anc}$ of ancestor jobs and the vectors $b^\anc, e^\anc$ in a partial dyadic system.   They are indeed top jobs passed from upper levels, and $b^\anc_j$ and $e^\anc_j$ for a job $j \in {J_\anc}$ give the values of $b^\bbJ_j$ and $e^\bbJ_j$ in the dyadic system $\bbJ$ we try to construct.   Therefore in the definition of a virtually-valid schedule, we treat them in the same way as top jobs: we require the window constraints for ancestor jobs to hold, but ignore the precedence constraints incident on them. The partial dyadic systems are introduced so that we can conduct mathematical inductions easily.

\subsection{Recursive Algorithm}
\label{subsec:algo-recursive}
	At a high level, our recursive algorithm $\schedule$ (described in Algorithm~\ref{alg:schedule}) tries to guess $\bbJ^*$ and $\sigma'^*$. We show that with the information on the top $h$ levels of intervals below the scheduling interval $I^*$, we can seamlessly break the instance into two sub-instances correspondent to the left and right half and $I^*$. 
	
	Before describing the algorithm, it is convenient to make one more definition: 
	\begin{definition}
		\label{def:multisets-equiv}
		Given two multi-sets of $\calW$ and $\calW'$ of intervals and an interval $I\in\calI$, we say $\calW$ and $\calW'$ are \emph{equivalent} within $I$, denoted as $\calW \equiv_{I} \calW'$, if $\{ W \cap I: W \in \calW \} = \{W \cap I: W \in \calW'\}$, where both sets are treated as multi-sets.
	\end{definition}
	
	
		\begin{algorithm}
			\caption{$\schedule\big(I^*,  J_\anc,   \tilde b^\anc,  \tilde e^\anc, (\tilde J_{I})_{I \in \calI_{<h-1}[I^*]}, (\tilde K_{I})_{I \in \calI_{h-1}[I^*]}\big)$
			}\label{alg:schedule}
			\textbf{Output:} a dyadic system $\bbJ^\best$ and a virtually-valid schedule $\sigma^\best$ for $\bbJ^\best$\\
			\textbf{Remark:} When $I^*$ is below level $L-h$, then $\calI_{<h-1}[I^*] = \calI[I^*]$ and $\calI_{h-1}[I^*] = \calI_{\geq h-1}[I^*] =  \emptyset$.
				
				\begin{algorithmic}[1]
				\If{$|\tilde J_\anc| > m|I^*|$ or  $|\tilde J_{\subseteq I^*}| + |\tilde K_{\subseteq I^*}| > m|I^*|$} \Return $(\bot, \bot)$ \EndIf  \label{step:schedule-check-size}
				\If{$I^* \in \calI_\sfbot$}
					\Return $\big(\bbJ^\best := (\tilde J_\anc, \tilde b^\anc, \tilde e^\anc, (\tilde J_{I^*})), \text{best virtually-valid schedule $\sigma^\best$ for } \bbJ^\best\big)$  \label{step:schedule-terminating-start}
					\label{step:schedule-terminating-end}
				\EndIf
				\State $\sigma^\best \gets \bot$, copy $\tilde b^\anc$, $\tilde e^\anc$, $\tilde J_I$'s and $\tilde K_I$'s to $b^\anc, e^\anc$, $J_I$'s  and $K_I$'s \label{step:schedule-init}
				\For{every possible vector $(g_I)_{I \in \calI_{h-1}[I^*] \setminus \calI_\sfbot}$ s.t $g_I \in \{\sfL, \sfR\}^p$ if $I \in \calI_\sftop$ and $g_I \in \{\sfL, \sfR\}^{m|I|}$ if $I \in \calI_\sfmid$} \label{step:schedule-enumerate-g}
					\For{every $I \in \calI_{h-1}[I^*]$}
					\State \textbf{if} $I\in \calI_\sftop \cup \calI_\sfmid$ \textbf{then} $\big( J_I,  K_{\sfleft(I)},  K_{\sfright(I)}\big) \gets \pushdown(I, K_I, g_{I})$ \textbf{else} $J_I\gets K_I$ \label{step:schedule-call-pushdown}
					\EndFor
					\For{every $j \in  J_{I^*}$} \label{step:schedule-extend-b-e-up} define
						\begin{itemize}[leftmargin=25pt]
							\item $b^\anc_j$ to be the minimum integer multiply $b$ of $\max\{2^{-h}|I^*|,2^h\}$ in $(\sfbegin(I^*),\sfcenter(I^*)]$ such that there are no precedence constraints from $J_{\subseteq (b, \sfcenter(I^*)]} \cup K_{\subseteq (b, \sfcenter(I^*)]}$ to $j$, and
							\item $e^\anc_j$ to be the maximum integer multiply $e$ of $\max\{2^{-h}|I^*|, 2^h\}$ in $[\sfcenter(I^*),\sfend(I^*))$ such that there are no precedence constraints from $j$ to $J_{\subseteq ( \sfcenter(I^*), e]} \cup K_{\subseteq (\sfcenter(I^*), e]}$.
						\end{itemize}
					\EndFor
					\For{every partition of $ {J_\anc} \cup  J_{I^*}$ into $ J_\anc^\sfL, J_\anc^\sfR$ and $J_\discarded$, keeping only one partition in every equivalence class defined in Remark~\ref{remark:equivalence}} \label{step:schedule-partition-Janc}
						\State $(\bbJ^\sfL, \sigma^\sfL)\gets \schedule\big(\sfleft(I^*),  J_\anc^\sfL,    b^\anc|_{ J_\anc^\sfL}, e^\anc|_{ J_\anc^\sfL}, (J_{I})_{I \in \calI_{<h-1}[\sfleft(I^*)]}, ( K_{I})_{I \in \calI_{h-1}[\sfleft(I^*)]}\big)$ \label{step:schedule-recurse-left}
						\State $(\bbJ^\sfR, \sigma^\sfR) \gets\schedule \big(\sfright(I^*),  J_\anc^\sfR,    b^\anc|_{ J_\anc^\sfR}, e^\anc|_{ J_\anc^\sfR}, (J_{I})_{I \in \calI_{<h-1}[\sfright(I^*)]}, (K_{I})_{I \in \calI_{h-1}[\sfright(I^*)]}\big)$ \label{step:schedule-recurse-right}
						\If{$\sigma^\sfL, \sigma^\sfR\neq \bot$ and \big($\sigma=\bot$ or $|(\sigma^\sfL)^{-1}(\sfleft(I^*))|+|(\sigma^\sfR)^{-1}(\sfright(I^*))|  > |(\sigma)^{-1}(I^*)|$\big)} 
						\label{step:schedule-check-if-better}
							\State $J^\best_{I^*} \gets J_{I^*}$; $J^\best_I \gets J^\sfL_I, \forall I \in \calI[\sfleft(I^*)]$; $J^\best_I \gets J^\sfR_I, \forall I \in \calI[\sfright(I^*)]$ \label{step:schedule-update}
							\State let $\sigma^\best$ be obtained by merging $\sigma^\sfL$ and $\sigma^\sfR$ and discard $J_\discarded$	
							 \label{step:schedule-update-1}
						\EndIf
					\EndFor
				\EndFor
				\State \textbf{if} $\sigma \neq \bot$ \textbf{then} \Return $\Big(\bbJ^\best:=\big(J_\anc, \tilde b^\anc,  \tilde e^\anc, (J^\best_I)_{I \in \calI[I^*]}\big), \sigma^\best\Big)$ \textbf{else} \Return $(\bot, \bot)$
			\end{algorithmic}
		\end{algorithm}
		
		\begin{algorithm}
			\caption{The Main Algorithm}
			\label{alg:main}
			\begin{algorithmic}[1]
				\State $\sigma^\best \gets$ schedule discarding all jobs in $J^\circ$, $K_{[T]} \gets J^\circ$
				\For{every $(g_I \in \{\sfL, \sfR\}^{{p}})_{I \in \calI_{< h - 1}}$} \label{step:main-guess-g}
					\For{every $I \in \calI_{<h-1}$ from top to bottom}
						\State $(J_I, K_{\sfleft(I)}, K_{\sfright(I)}) \gets \pushdown(I, K_I, g_I)$
						\State \textbf{if} $|K_{\sfleft(I)}|> m|I|/2$ or  $|K_{\sfright(I)}| >m|I|/2$ \textbf{then continue} to the next iteration of Loop~\ref{step:main-guess-g} \label{step:main-check-size}
					\EndFor
					\State $(\bbJ,\sigma) \gets \schedule([T],\emptyset, (), (), (J_I)_{I \in \calI_{<h-1}}, (K_I)_{I \in \calI_{h-1}})$ \label{step:main-call-schedule}
					\State \textbf{if} $|\sigma^{-1}([T])| > |(\sigma^\best)^{-1}([T])|$ \textbf{then} $\bbJ^\best \gets \bbJ, \sigma^\best \gets \sigma$
				\EndFor
				\State \Return $(\bbJ^\best, \sigma^\best)$
			\end{algorithmic}
		\end{algorithm}
	
	\begin{figure*}
		\centering
		\includegraphics[width=0.85\textwidth]{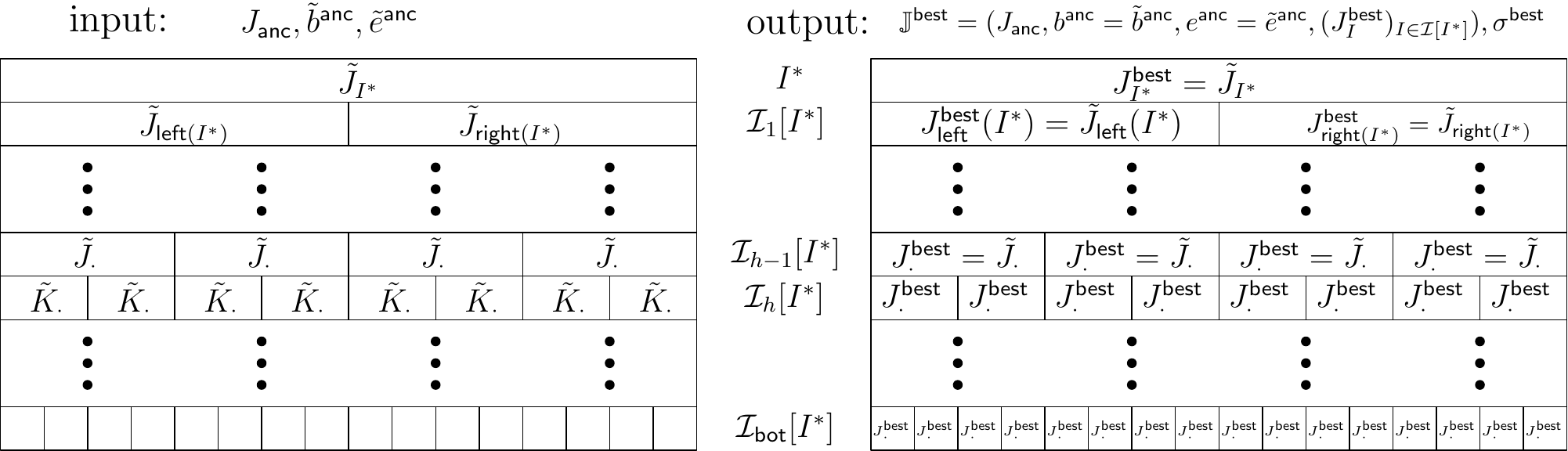}
		\caption{Input and Output of a Recursion of $\schedule$.} \label{fig:schedue}
	\end{figure*}
		
	In the algorithm $\schedule$, we are given an interval $I^* \in \calI$, a set ${J_\anc}$ of jobs, two vectors $b^\anc$ and $e^\anc$, a set $J_I$ of jobs for every $I\in \calI_{<h-1}[I^*]$ and a set $\tilde K_I$ of jobs for every $I \in \calI_{h-1}[I^*]$; notice that if $I^*$ is below level $L-h + 1$, then $\calI_{<h-1}[I^*] = \calI[I^*]$ and $\calI_{h-1}[I^*] = \emptyset$.  Our goal is to construct a dyadic system $\bbJ$ over $I^*$ and a virtually-valid schedule $\sigma$ for $\bbJ$.  $J_\anc, \tilde b^\anc$ and $\tilde e^\anc$ give the set of ancestor jobs in $\bbJ$ and their $b^\bbJ_j$ and $e^\bbJ_j$ values.  Each $\tilde J_I$ specifies the jobs assigned to $I$ in $\bbJ$, and each $\tilde K_I$ specifies the set of jobs assigned to sub-intervals of $I$. Therefore, we know exactly how jobs are assigned to the first $h-1$ levels of the tree rooted at $I^*$; for each job assigned to $\calI_{\geq h-1}[I^*]$, we only know the super interval in $\calI_{h-1}[I^*]$ of its owning interval. 
	
	The goal of the procedure is to construct the system $\bbJ^\best$ satisfying the requirements,  and construct a virtually-valid schedule $\sigma^\best$ for $\bbJ^\best$. See Figure~\ref{fig:schedue} for the illustration of the input and output for a recursion of $\schedule$.  In the algorithm, all variables in the procedure except the ones defined in Section~\ref{sec:prelim} are local. To avoid confusions in the analysis, we never change the input variables in the process; that is, they are read-only. 
	
	In Step~\ref{step:schedule-check-size} of Algorithm~\ref{alg:schedule}, we check if both the number of ancestors jobs and total number of top and bottom jobs are at most $m|I^*|$; if not, we return $(\bot, \bot)$ immediately, where $\bot$ stands for  ``not defined''. If $I^*\in \calI_\sfbot$, then $\bbJ^\best$ is decided.  We then find the best virtually-valid schedule $\sigma^\best$ for $\bbJ^\best$ by enumeration and return immediately (Step~\ref{step:schedule-terminating-start}). So we assume $I^* \in \calI_\sftop  \cup \calI_\sfmid$. 
	
	We initialize some variables in Step~\ref{step:schedule-init}. Then in Loop~\ref{step:schedule-enumerate-g}, we try to guess $g^*_I$ for all $I \in \calI_{h-1}[I^*] \setminus \calI_\sfbot$ (more precisely, we shall extend each $g^*_I$ so that it has length ${p}$ or $m|I|$ and guess the extension). Based on our guesses,  we expand the information about $\bbJ^*$ by one more level: for every $I \in \calI_{h-1}[I^*]$, in Step~\ref{step:schedule-call-pushdown}, we partition $K_I$ into $J_I, K_{\sfleft(I)}$ and $K_{\sfright(I)}$ by calling $\pushdown$; if $I$ is a bottom interval, we simply set $J_I = K_I$. The information we have now is sufficient to define the $b^\bbJ$ and $e^\bbJ$ values for jobs in $J_{I^*}$.   We then compute these values in Step~\ref{step:schedule-extend-b-e-up}: For every $j \in J_{I^*}$, $b^\anc_j$ and $e^\anc_j$ will be the same as the $b^\bbJ_j$ and $e^\bbJ_j$ values for the constructed partial dyadic system $\bbJ^\best$.  
	
	In Loop~\ref{step:schedule-partition-Janc}, we guess how jobs in ${J_\anc} \cup J_{I^*}$ are split into $\sfleft(I^*)$ and $\sfright(I^*)$. Since we only focus on virtually-valid schedules, we can then ignore the precedence  constraints incident to ${J_\anc} \cup J_{I^*}$. This is crucial in reducing the number of possibilities.  In the loop, we only keep one partition in every equivalence class defined as follows:
	\begin{remark}
		\label{remark:equivalence}
		In Step~\ref{step:schedule-partition-Janc}, we say two partitions $( J_\anc^\sfL,  J_\anc^\sfR, J_\discarded)$ and $( J_\anc'^\sfL,  J_\anc'^\sfR, J'_\discarded)$ of $ {J_\anc}\cup  J_{I^*}$ are equivalent if  $\{(b^\anc_j, e^\anc_j] : j \in J_\anc^\sfL\} \equiv_{\sfleft(I^*)} \{(b^\anc_j, e^\anc_j]: j \in J_\anc'^\sfL\}$ and $\{(b^\anc_j, e^\anc_j]: j \in J_\anc^\sfR\} \equiv_{\sfright(I^*)} \{(b^\anc_j, e^\anc_j]: j \in J_\anc'^\sfR\}$ where all sets are treated as multi-sets (thus, $| J_\anc^\sfL| = | J_\anc'^\sfL|$ and $| J_\anc^\sfR| = | J_\anc'^\sfR|$). 
	\end{remark}
	 Later we show that the number of equivalence classes is small.  Once we made the guess, we recursively and independently call $\schedule$ for $\sfleft(I^*)$ and $\sfright(I^*)$ (Step~\ref{step:schedule-recurse-left} and \ref{step:schedule-recurse-right}).   We assume the job sets in the two constructed systems $\bbJ^\sfL$ and $\bbJ^\sfR$ are automatically named $J^\sfL_I$'s and $J^\sfR_I$'s.   We maintain the best solution constructed so far (Step~\ref{step:schedule-update} and \ref{step:schedule-update-1}) and return it in the end.

	In the main algorithm (Algorithm~\ref{alg:main}), we enumerate all possible combinations of $(g_I)_{I \in \calI_{<h-1}}$ and use each of the combinations to obtain $(J_I)_{I \in \calI_{<h-1}}$ and $(K_I)_{I \in \calI_{h-1}}$. Then we use the information to call the algorithm $\schedule$ and return the best schedule constructed so far. 
	
\subsection{Analysis of Correctness}
	\label{subsec:algo-correctness}
	We now analyze the correctness of the algorithm. The following claim gives some simple properties about the input to each recursion of $\schedule$.
		\begin{claim}
			\label{claim:schedule-properties}
			The input parameters of $\schedule$ for some $I^* \in \calI$ satisfy the following.
			\begin{enumerate}[label=(\ref{claim:schedule-properties}\alph*)]
				\item All sets in $\{{J_\anc}\} \cup \{\tilde J_I\}_{I \in \calI_{<h-1}[I^*]}\cup\{\tilde K_I\}_{I \in \calI_{h-1}[I^*]}$ are mutually disjoint. \label{property:schedule-input-partition}
				\item For every $I \in \calI_{<h-1}[I^*] \cap \calI_\sftop[I^*]$, we have $\Delta(\tilde J_I) \leq \delta|\tilde J_I| + \delta'|I|$. \label{property:schedule-input-chain-length}
				\item For every $I \in \calI_{<h-1}[I^*] \cap \calI_\sfmid[I^*]$, we have $\tilde J_I = \emptyset$. \label{property:schedule-input-middle-empty}
				\item The sequence $(\tilde J_{I} \text{ or }  \tilde K_{I})_{I \in \calI_{\leq h-1}[I^*]}$ according to the order $\inorderless$ respects the precedence constraints, where $(\tilde J_{I} \text{ or }  \tilde K_{I})$ indicates either $\tilde J_I$ or $\tilde K_I$ depending on which one is given in the input. \label{property:schedule-input-precedence}
				\item For every $j \in {J_\anc}$, we have that $\tilde b^\anc_j $ and $\tilde e^\anc_j$ are integer multiplies of $\max\{2^{-h+1}|I^*|, 2^h\}$; moreover, for every $j \in {J_\anc}$, we have either $\tilde b^\anc_j \leq \sfbegin(I^*)$ or $\tilde e^\anc_j \geq \sfend(I^*)$. \label{property:schedule-input-b-e-multiplies}
			\end{enumerate}
		\end{claim}
		\begin{proof}
			Properties~\ref{property:schedule-input-partition}, \ref{property:schedule-input-chain-length}, \ref{property:schedule-input-middle-empty} 
			and~\ref{property:schedule-input-precedence} hold due to Properties \ref{property:pushdown-partition}, \ref{property:pushdown-chain-length}, \ref{property:pushdown-middle-empty}
			and~\ref{property:pushdown-precedence} for $\pushdown$. 
			For Property~\ref{property:schedule-input-b-e-multiplies}, note that the $\tilde b^\anc_j$ and $\tilde e^\anc_j$ values for each $j \in {J_\anc}$ must be computed in ancestor recursions of $\schedule$, for some interval $\hat I \in \calI, \hat I \supsetneq I^*$. Then $\tilde b^\anc_j$ and $\tilde e^\anc_j$ are both integer multiplies of $\max\{2^{-h}|\hat I|, 2^h\}$, which are integer multiplies of $\max\{2^{-h+1}|I^*|,2^h\}$. Also, if $I^* \subseteq \sfleft(\hat I)$, then $\tilde e^\anc_j \geq \sfcenter(\hat I) \geq \sfend(I^*)$; if $I^* \subseteq \sfright(\hat I)$, then $\tilde b^\anc_j \leq \sfcenter(\hat I) \leq \sfbegin(I^*)$.
		\end{proof}
	
	The following lemma shows the validity of the output for each recursion of $\schedule$.	 
		\begin{lemma}
			\label{lemma:schedule-valid}
			Suppose some recursion of $\schedule$ takes $\big(I^*,  {J_\anc},    \tilde b^\anc,  \tilde e^\anc, ( \tilde J_{I})_{I \in \calI_{<h-1}[I^*]}, (\tilde K_{I})_{I \in \calI_{h-1}[I^*]}\big)$ as input and returns $\big(\bbJ^\best,\sigma^\best\big) \neq (\bot, \bot)$.  Then $\bbJ^\best = (J_\anc, \tilde b^\anc, \tilde e^\anc, (J^\best_I)_{I \in \calI[I^*]})$ is a {partial dyadic} system over $I^*$ and $\sigma^\best$ is a virtually-valid schedule for $\bbJ^\best$.  Moreover, $J^\best_I = \tilde J_I$ for every $I \in \calI_{<h-1}[I^*]$, and $J^\best_{\subseteq I} = \tilde K_I$ for every $I \in \calI_{h-1}[I^*]$.
		\end{lemma}
		\begin{proof}
			Notice that $J^\best_I = \tilde J_I$ for every $I \in \calI_{<h-1}[I^*]$ trivially holds. We prove the other two statements by induction from bottom to top. Consider the case $I^* \in \calI_\sfbot$.  Then $\bbJ^\best$ is trivially a partial dyadic system and $\sigma^\best$ is a virtually-valid schedule for $\bbJ^\best$ (notice that the schedule that discards all jobs is always virtually-valid). Thus the lemma holds. 
			
			Now we assume $I^* \in \calI_\sftop \cup \calI_\sfmid$. Consider the last iteration of the two nested loops in which Step \ref{step:schedule-update} and \ref{step:schedule-update-1} are executed. Then the final $\bbJ^\best$ and $\sigma^\best$ are constructed in this iteration.  All the notations we use are w.r.t this the moment at the end of the iteration. The induction hypothesis for the sub-recursions of $\schedule$ made in Step~\ref{step:schedule-recurse-left} and~\ref{step:schedule-recurse-right} says
			\begin{itemize}
				\item $\bbJ^\sfL=\big(J_\anc^\sfL,    b^\anc|_{J_\anc^\sfL},  e^\anc|_{J_\anc^\sfL}, (J^\sfL_I)_{I \in \calI[\sfleft(I^*)]}\big)$ is a {partial dyadic} system over $\sfleft(I^*)$ and $\sigma^\sfL$ is a valid schedule for $\bbJ^\sfL$. Moreover $J^\sfL_{\subseteq I} =  K_{I}$ for every $I \in \calI_{h-1}[\sfleft(I^*)]$. 
				\item $\bbJ^\sfR =\big(J_\anc^\sfR,    b^\anc|_{J_\anc^\sfL},  e^\anc|_{J_\anc^\sfR}, (J^\sfR_I)_{I \in \calI[\sfright(I^*)]}\big)$ is a {partial dyadic} system over $\sfright(I^*)$ and $\sigma^\sfR$ is a valid schedule for $\bbJ^\sfR$. Moreover $J^\sfR_{\subseteq I} =  K_{I}$ for every $I \in \calI_{h-1}[\sfright(I^*)]$. 
			\end{itemize}
			We first  prove $J^\best_{\subseteq I} = \tilde K_I$ for every $I \in \calI_{h-1}[I^*]$. First consider the case that $I^* \in \calI_{\leq L-h}$.  Focusing on some $I \in \calI_{h-2}[\sfleft(I^*)] \subseteq \calI_{h-1}[I^*]$, and $\sfleft(I), \sfright(I) \in \calI_{h-1}[\sfleft(I^*)]$, we have
			\begin{align*}
				J^\sfL_{\subseteq I} &=  J^\sfL_I \ \cup J^\sfL_{\subseteq \sfleft[I]} \cup J^\sfL_{\subseteq \sfright[I]}
				=  J_I \cup  K_{\sfleft(I)} \cup  K_{\sfright(I)} =  K_I = \tilde K_I
			\end{align*}
			where the second equality used the induction hypothesis and the third equality used Property~\ref{property:pushdown-partition} for $\pushdown$.  Notice that $J^\sfL_I= J_I$ since $J_I$ is passed to the sub-recursion for constructing $\bbJ^\sfL$; and we copied $\tilde K_I$ to $K_I$. Similarly,  for every $I \in \calI_{h-2}[\sfright(I^*)] \subseteq \calI_{h-1}[I^*]$, we have $J^\sfR_{\subseteq I} =  \tilde K_I$.  Therefore, for every $I \in \calI_{h-1}[I^*]$ we have $J^\best_{\subseteq I} =  \tilde K_I$.
			
			Now consider the case $I^* \in \calI_{L-h + 1}$.  Then $\calI_{h-1}[\sfleft(I^*)] = \calI_{h-1}[\sfright(I^*)]=\emptyset$. For every $I\subseteq \calI_{h-1}[I^*]\subseteq \calI_\sfbot$, we have $J^\best_{\subseteq I} = J^\best_I = \tilde K_I$ by step~\ref{step:schedule-call-pushdown}.  Finally, if $I^* \in \calI_{>L-h + 1}$, then $\calI_{h-1}[I^*] = \emptyset$ and there is nothing to prove. 
			
			Then we prove that $\bbJ^\best$ is a partial dyadic system and $\sigma^\best$ is a virtually-valid schedule for $\bbJ^\best$. Notice that $\bbJ^\best$ can be viewed as obtained by merging $\bbJ^\sfL$ and $\bbJ^\sfR$, adding $J_\discarded$ to the set of ancestor jobs, and moving $\tilde J_{I^*}$ from ancestor jobs to $J^\best_{I^*}$.  $\sigma^\best$ is obtained by merging $\sigma^\sfL$ and $\sigma^\sfR$ and discarding all jobs in $J_\discarded$.   So, $\bbJ $ is indeed a {partial dyadic} system since $\bbJ^\sfL$ and $\bbJ^\sfR$ are both {partial dyadic} systems: Property~\ref{property:dyadic-disjoint} for  $\bbJ^\best$
			is implied by the same property for $\bbJ^\sfL$ and $\bbJ^\sfR$. Property~\ref{property:dyadic-chain-length} and \ref{property:dyadic-middle-empty} are implied by the properties for $\bbJ^\sfL$ and $\bbJ^\sfR$ and that  $\Delta(\tilde J_{I^*}) \leq \delta|\tilde J_{I^*}| + \delta'|I^*|$ if $I^* \in \calI_\sftop$ and $\tilde J_{I^*} = \emptyset$ if $I^* \in \calI_\sfmid$. Property~\ref{property:dyadic-precedence} is implied by the property for $\bbJ^\sfL$ and $\bbJ^\sfR$ and that the sequence $\tilde J_{\subseteq \sfleft(I^*)} \cup \tilde K_{\subseteq \sfleft(I^*)}, \tilde J_{I^*}, \tilde J_{\subseteq \sfright(I^*)} \cup \tilde K_{\subseteq \sfright(I^*)}$ respects the precedence constraints (implied by Property~\ref{property:schedule-input-precedence}).
			
			Notice that capacity constraints, precedence and interval constraints for bottom jobs for $\sigma^\best$ are implied by the same properties for $\sigma^\sfL$ and $\sigma^\sfR$.  For every top job $j$ in $\bbJ^\sfL$, we have $\tilde b^{\bbJ}_j = \tilde b^{\bbJ^\sfL}_j$, and for every top job $j$ in $\bbJ^\sfR$, we have $\tilde b^{\bbJ}_j = \tilde b^{\bbJ^\sfR}_j$.  Jobs in $J_\anc \cap J_\anc^\sfL$ and $J_\anc \cap J_\anc^\sfR$ will have consistent $\tilde b^\anc$ and $\tilde e^\anc$ values in the three {partial dyadic} systems.  For every $j \in  \tilde J_{I^*}$ has $\tilde b^{\bbJ}_j = \tilde b^\anc_j$ and $\tilde e^{\bbJ}_j = \tilde e^\anc_j$. This holds by our definitions of $ \tilde b^\anc_j$ and $ \tilde e^\anc_j$ in Step~\ref{step:schedule-extend-b-e-up}.  Thus, the virtual-validity of $\sigma^\best$ is implied by the virtual-validity of $\sigma^\sfL$ and $\sigma^\sfR$. 
		\end{proof}
		
		Now it remains to show that the number of discarded jobs in the returned virtually-valid schedule is small. This is guaranteed by the existence of $\bbJ^*$  and $\sigma'^*$.  Recall that $b^*_j = b^{\bbJ^*}_j$ and $e^*_j = b^{\bbJ^*}_j$ for every top job $j$ in $\bbJ^*$. The following lemma says that if our guesses about $\bbJ^*$ and $\sigma'^*$ are correct, then the number of discarded jobs in the returned schedule is small.
		\begin{lemma}
			\label{lemma:schedule-good}
			Suppose at the beginning of some recursion of $\schedule$, we have $\{(\tilde b^\anc_j, \tilde e^\anc_j] : j \in  {J_\anc}\} \equiv_{I^*} \{(b^*_j, e^*_j]: j \in \sigma'^{*-1}(I^*) \cap J^*_{\supsetneq I^*}\}$, $\tilde J_I = J^*_I$ for every $I\in \calI_{<h-1}[I^*]$ and $\tilde K_I = K^*_I =  J^*_{\subseteq I}$ for every $I \in \calI_{h-1}[I^*]$.  Then the returned schedule $\sigma^\best$ has at least $|\sigma'^{*-1}(I^*)|$ jobs scheduled.
		\end{lemma}
		\begin{proof}
			Notice in Step~\ref{step:schedule-check-size} of $\schedule$, we will not return $(\bot, \bot)$ immediately since $|{J_\anc}| = |\sigma'^{*-1}(I^*) \cap J^*_{\supsetneq I^*}|\leq |\sigma'^{*-1}(I^*)| \leq m|I^*|$ and $|J_{\subseteq I^*}| + |K_{\subseteq I^*}| = |K^*_{I^*}| \leq m|I^*|$. We prove the lemma by induction from bottom to top.  First consider the case $I^* \in \calI_\sfbot$. $\sigma'^{*-1}(I^*)$ was scheduled in $I^*$ in the schedule $\sigma'^*$.  Then the schedule $\sigma$ obtained from the $\sigma'^*$ restricted on $I^*$, with $\sigma'^{*-1}(I^*) \cap J^*_{\supsetneq I^*}$ replaced by jobs in ${J_\anc}$ using the equivalence mapping,  is a candidate schedule.  So, we now assume that $I^* \in \calI_\sftop \cup \calI_\sfmid$.
			
			Consider the iteration of Loop~\ref{step:schedule-enumerate-g} in which $g^*_I$ is a prefix of $g_I$ for every $I \in \calI_{h-1}[I^*] \setminus \calI_\sfbot$; such an iteration exists since the length of $g^*_I$ is at most ${p}$ of $I\in \calI_\sftop$ and at most $|K^*_I| \leq m|I|$ if $I \in \calI_\sfmid$.  By Property~\ref{property:pushdown-returns-star} for the procedure $\pushdown$, we have $ J_I = J^*_I$ for every $I \in \calI_{h-1}[I^*]$ and $ K_I = J^*_{\subseteq I}$ for every $I \in \calI_h[I^*]$ after Step~\ref{step:schedule-call-pushdown}.  Then after Step~\ref{step:schedule-extend-b-e-up}, we have that for every $j \in  J_{I^*} = J^*_{I^*}$, $b^*_j = b^\anc_j$ and $e^*_j = e^\anc_j$; this holds by the definitions of $b^*_j, e^*_j$, the definitions of $b^\anc_j$ and $e^\anc_j$ in the step,  and the conditions of the lemma. 
			
			Now we focus on Loop~\ref{step:schedule-partition-Janc}. By the conditions of the lemma and that $b^*_j = b^\anc_j, e^*_j = e^\anc_j$ for every $j \in J$, we have 
			\begin{align*}
				&\qquad \big\{(b^\anc_j, e^\anc_j]: j \in  {J_\anc} \cup J^*_{I^*}\}\equiv_{I^*} \{(b^*_j, e^*_j]: j \in \sigma'^{*-1}(I^*) \cap J^*_{\supsetneq I^*} \cup J^*_{I^*}\big\}\\
				&\equiv_{I^*} \Big\{(b^*_j, e^*_j]: j \in \sigma'^{*-1}(I^*) \setminus \big(K^*_{\sfleft(I^*)} \cup K^*_{\sfright(I^*)}\big) \cup \big(\sigma'^{*-1}(\discarded) \cap J^*_{I^*}\big)\Big\}\\
				&\equiv_{I^*} \Big\{(b^*_j, e^*_j]: j \in \left[\sigma'^{*-1}(\sfleft(I^*)) \setminus K^*_{\sfleft(I^*)}\right] \bigcup \left[ \sigma'^{*-1}(\sfright(I^*)) \setminus K^*_{\sfright(I^*)}\right] \cup \left[\big(\sigma'^{*-1}(\discarded) \cap J^*_{I^*}\big)\right]\Big\}.
			\end{align*}
			Then there is a partition $(J_\anc^\sfL, J_\anc^\sfR, J_\discarded)$ of ${J_\anc} \cup J^*_{I^*}$ such that
			\begin{enumerate}[label=(\roman*)]
				\item $\big\{(b^\anc_j, e^\anc_j]: j \in  J_\anc^\sfL\big\} \equiv_{I^*} \big\{(b^*_j, e^*_j]: j \in \sigma'^{*-1}(\sfleft(I^*)) \setminus K^*_{\sfleft(I^*)}\big\}$,
				\item $\big\{(b^\anc_j, e^\anc_j]: j \in  J_\anc^\sfR\big\} \equiv_{I^*} \big\{(b^*_j, e^*_j]: j \in \sigma'^{*-1}(\sfright(I^*)) \setminus K^*_{\sfright(I^*)}\big\}$, and 
				\item $\big\{(b^\anc_j, e^\anc_j]: j \in  J_\discarded\big\} \equiv_{I^*} \big\{(b^*_j, e^*_j]: j \in \sigma'^{*-1}(\discarded) \cap J^*_{I^*}\big\}$.
			\end{enumerate}
			Notice that (i) and (ii) hold with $\equiv_{I^*}$ replaced by $\equiv_{\sfleft(I^*)}$ and $\equiv_{\sfright(I^*)}$ respectively. By the definition of equivalence of partitions in Remark~\ref{remark:equivalence},  in Loop~\ref{step:schedule-partition-Janc} there will be an iteration where the two conditions hold. Therefore, in the iteration, the conditions of the lemma for the two sub-recursions of $\schedule$ hold. So we have that $\sigma^\sfL$ has at least $|\sigma'^{*-1}(\sfleft(I^*))|$ jobs scheduled, and $\sigma^\sfR$ has at least $|\sigma'^{*-1}(\sfright(I^*))|$ jobs scheduled. Thus, at the end of the iteration, we have that $\sigma^\best$ has least $|\sigma'^{*-1}(\sfleft(I^*))| + |\sigma'^{*-1}(\sfright(I^*))| = |\sigma'^{*-1}(I^*)|$ jobs scheduled.	 In the end, 	$\sigma^\best$ schedules at least $|\sigma'^{*-1}(I^*)|$ jobs.
		\end{proof}
		Now in the main algorithm, consider the iteration of Loop~\ref{step:main-guess-g} in which $g^*_I$ is a prefix of $g_I$ for every $I \in \calI_{<h-1}$. Then we shall have $J_I = J^*_I$ for every $I \in \calI_{<h-1}$ and $K_I = K^*_I$ for every $I \in \calI_{h-1}$.  Then the $\sigma$ returned in Step~\ref{step:main-call-schedule} will have at least $|\sigma'^{*-1}([T])|$ jobs scheduled.  So in the end, the main algorithm will return a {partial dyadic} system $\bbJ^\best$ and a virtually-valid schedule $\sigma$ for $\bbJ^\best$ with at least $|\sigma'^{*-1}([T])|$ jobs scheduled.

\subsection{Analysis of Running Time}
	\label{subsec:algo-running-time}
		Finally we analyze the running time of the algorithm.
		\begin{lemma}
			\label{lemma:schedule-run-time}
			The running time of $\schedule$ for $I^* = [T]$ is at most $\exp\left(O\left(\frac{m^4}{\epsilon^3}\log^3n\log \log n\right)\right)$.
		\end{lemma} 
		\begin{proof}
			For every $\ell \in [0, L)$, we define $R_\ell$ to be the maximum number of times we call $\schedule$ as sub-recursions in a recursion of $\schedule$ for some $I^* \in \calI_\ell$.  Let $R_{L}$ be the worst case running time for a recursion of $\schedule$ for some $I^* \in \calI_{L} = \calI_\sfbot$.   Notice that the running time of recursion of $\schedule$ for some $I^*\in \calI_\ell$, not counting the running time for sub-recursions, is at most $\poly(n) R_\ell$.  Then it is easy to see that the running time of $\schedule$ for $I^* = [T]$ is at most $\poly(n)\prod_{\ell=0}^{L}R_\ell = \poly(n)\exp\left(\sum_{\ell=0}^L \ln R_{\ell}\right)$.
			
			First, we bound $R_\ell$ for $\ell < L$ by focusing on any $I^* \in \calI_\ell$. If $\ell \leq L-h' - h$, $(g_I)_{I \in \calI_{h-1}[I^*]\setminus \calI_\sfbot}$ has total length $2^{h-1}\cdot {p}$.  If $\ell \geq L-h'-h+1$ but $\ell \leq L-h$, the total length is $2^{-\ell} mT$ since $\sum_{I \in \calI_{h-1}[I^*]}m|I| = m|I^*| = 2^{-\ell}mT$. If $\ell \geq L-h+1$, then the length is $0$. Now we consider the number of different ways to split ${J_\anc} \cup J_{I^*}$ into $J_\anc^\sfL, J_\anc^\sfR$ and $J_\discarded$. By Property~\ref{property:schedule-input-b-e-multiplies},  $b^\anc_j$ and $e^\anc_j$ values for $j \in {J_\anc} \cup J_{I^*}$ will be integer multiplies of $2^{-h}|I^*|$. Moreover, for each $j \in {J_\anc} \cup J_{I^*}$, we have $b^\anc_j \leq \sfbegin(\sfleft(I^*))$ or $e^\anc_j \geq \sfend(\sfleft(I^*))$, and we also have $b^\anc_j \leq \sfbegin(\sfright(I^*))$ or $e^\anc_j \geq \sfend(\sfright(I^*))$. Thus, there are at most $4\cdot 2^{h-1} = 2^{h+1}$ distinct elements in $\{(b^\anc_j, e^\anc_j] \cap \sfleft(I^*): j \in {J_\anc} \cup J_{I^*}\} \cup \{(b^\anc_j, e^\anc_j] \cap \sfright(I^*): j \in {J_\anc} \cup J_{I^*}\}$. Therefore, we have at most $n^{\cdot2^{h+1}}$ distinct equivalence classes for partitions $(J_\anc^\sfL, J_\anc^\sfR, J_\discarded)$. So, if $\ell \leq L-h' - h$, we have 
			\begin{align*}
				\log R_\ell &\leq 1  + 2^{h-1}\cdot {p} + 2^{h+1} \log n \leq O\left(2^h {p}\right)
				\leq O\left(\frac{m\log T}{\epsilon}\cdot \frac{m^3\log T\log \log T}{\epsilon^2}\right) = O\left(\frac{m^4\log^2n\log\log n}{\epsilon^3}\right).
			\end{align*}
			If $\ell \geq L - h' -h +1 = \log T - 2h - h' + 1$, we have 
			\begin{align*}
				\log R_\ell &\leq 1 + 2^{-\ell}mT + 2^{h+1} \log n \leq  1 + 2^{2h+h'-1}m + 2^{h+1} \log n \leq O(2^{2h+h'}m) = O\left(\frac{m^4\log^2n}{\epsilon^3}\right).
			\end{align*}
		
			Now we bound $\log R_{L}$ and focus on any $I^* \in \calI_\sfbot$. Since in Step~\ref{step:schedule-check-size} of $\schedule$, we guaranteed the sizes of ${J_\anc}$ and $J_{I^*}$ are at most  $m|I^*| = m2^h$. We have 
			\begin{align*}
				\log R_L \leq  O\left(m2^h\log 2^h\right) = O\left(\frac{m^2\log n \log\log n}{\epsilon}\right).
			\end{align*}		
			Thus, 
			\begin{align*}
				\sum_{\ell=0}^{L} \log R_\ell &\leq O\left(L\cdot \frac{m^4\log^2n\log\log n}{\epsilon^3}  + \frac{m^2\log n\log\log n}{\epsilon}\right) = O\left(\frac{m^4\log^3n\log\log n}{\epsilon^3}\right).
			\end{align*}
			This finishes the proof of lemma.
		\end{proof}
				
		Then the running time of the main algorithm is at most $\exp\left(O(1)\cdot 2^h \cdot {p}\right) \leq \exp\left(O\left(\frac{m^4\log^2n\log \log n}{\epsilon^3}\right)\right)$ times that of $\schedule$ for $[T]$. So overall the running time of main is at most $\exp\left( O\left(\frac{m^4}{\epsilon^3}\log^3n\log \log n\right)\right) = n^{O\left(\frac{m^4}{\epsilon^3}\cdot\log^2n\cdot \log \log n\right)}$.
		
		\paragraph{Wrapping Up} Running the main algorithm, we can obtain a {partial dyadic} system $\bbJ^\best$ and a virtually-valid schedule $\sigma''$ for $\bbJ^\best$ with  $|\sigma''^{-1}(\discarded)| \leq |\sigma'^{*-1}(\discarded)| \leq \frac{3\epsilon T}{4}$ by Lemma~\ref{lemma:valid-to-virtually-valid}. By Lemma~\ref{lemma:canonicalize} and \ref{lemma:virtually-valid-to-valid}, we can convert $\sigma'$ to a valid schedule $\sigma$ for $\bbJ^\best$ with $|\sigma^{-1}(\discarded) \setminus \sigma'^{-1}(\discarded)| \setminus \frac{3\epsilon T}{4}$. Thus, we have  $|\sigma^{-1}(\discarded)| \leq \frac{\epsilon T}{4} + \frac{3\epsilon T}{4} = \epsilon T$. By inserting the at most $\epsilon T$ jobs back to $\sigma$ using the procedure in Appendix~\ref{appendix:manipulations} we can obtain a schedule for all jobs in $J^\circ$ with makespan at most $(1+\epsilon)T$. The running time of the algorithm is $n^{O\left(\frac{m^4}{\epsilon^3}\cdot\log^2n\cdot \log \log n\right)}$. This finishes the proof of Theorem~\ref{thm:main} with running time replaced by $n^{O\left(\frac{m^4}{\epsilon^3}\cdot\log^2n\cdot \log \log n\right)}$.

%% file: appendix.tex
	\section{Simple Transformations}
	\label{appendix:manipulations}
	We show that without loss of generality, we can assume the optimum makespan $T$ is an integer power of $2$. Let $T'$ be the smallest integer power of $2$ that is at least $T$. We add $m(T'-T)$ jobs $\tilde J$ to the input set and add precedence constraints from all jobs in $J^\circ$ to all jobs in $\tilde J$. It is easy to see that the optimum makespan for the new instance is $T'$ and any schedule with makespan $(1+\epsilon)T'$ for the new instance can be converted to a schedule for the original instance with makespan $(1+\epsilon)T' - (T' - T) = T + \epsilon T' \leq (1+2\epsilon)T$. Thus a $(1+\epsilon)$-approximation for the new instance implies a $(1+2\epsilon)$-approximation for the original instance. 
	
	Then we show a valid schedule $\sigma$ of makespan $T$ with $a$ jobs discarded can be converted to a schedule of makespan $T+a$ for all jobs.  For every discarded job $j$, we insert $j$ to the schedule using the following procedure. We find the predecessor of $j$ that is scheduled latest in $\sigma$ and assume it is scheduled at time $t$; if no predecessor of $j$ is scheduled in $\sigma$, let $t = 0$. We then shift the schedule $\sigma$ starting from time $t+1$ to the right by 1 unit time. This will leave all the machines idle at time $t+1$. We then schedule $j$ at time $t+1$. All the precedence constraints to $j$ are satisfied by the definition of $t$.  All the precedence constraints from $j$ are satisfied since before inserting $t$, any successor of $j$ must be scheduled at time $t+1$ or later.   So inserting $j$ will increase the makespan of $\sigma$ by 1. The final schedule has makespan at most $T+a$.

\section{Proofs of Helper Lemmas}
	\label{appendix:proof-helper}
	\listscheduling*
	\begin{proof}
		We can assume $\sum_{t \in I} \sfcap(t) = |J|$ by decreasing the $\sfcap$ values.  We use a simple greedy algorithm to schedule the jobs.  For every $t$ from $\sfbegin(I)+1$ to $\sfend(I)$,  we try to schedule as many jobs as possible at time $t$ using any strategy. The number of jobs we can schedule at time $t$ is the minimum of the following two numbers: (a) the number of unscheduled jobs whose predecessors were all scheduled before $t$, and (b) the number $\sfcap(t)$ of available slots at time $t$.  In the end, all the jobs that are not scheduled will be discarded. 
		
		At any time, we say a job is a min-job if all its predecessors are scheduled and itself is not scheduled yet; these are the jobs that are ready for scheduling. During any time $t$, either all the $\sfcap(t)$ slots are being used, or all the min-jobs at the beginning of time $t$ are scheduled.  In the latter case, $\Delta(J_{\textsf{rem}})$ will be decreased by 1, where $J_{\textsf{rem}}$ denotes the set of jobs that are not scheduled yet; this holds since we will never run out of jobs. Thus the total number of time steps in which the latter case happens is at most $\Delta(J)$. There will be at most $m\Delta(J)$ available slots that are not used for scheduling. So, we discarded at most $m\Delta(J)$ jobs.
	\end{proof}

	\lengthaddup*
	\begin{proof}
		Focus on any maximum-length precedence chain of jobs in $J_i$ for any $i \in [k]$. The jobs must have distinct $\depth_J(\cdot)$ values: Two jobs with the same $\depth_J(\cdot)$ value can not have a precedence constraint between them. So, $\big\{\langle c(j), \depth_J(j)\rangle: J \in J_i\big\}$ must contain $\Delta(J_i)$ different vectors. Moreover, any vector in $\big\{\langle c(j),\depth_J(j)\rangle: j \in J_i\big\}$ is less than or equal to any vector in $\big\{\langle c(j), \depth_J(j)\rangle: J \in J_{i'}\big\}$ for $i' > i$.  It is straightforward to prove that the number of different vectors in $\{(c(j), \depth_J(j)): j \in J\}$ should be at least $\Delta(J_1) + \Delta(J_2) + \cdots + \Delta(J_k) - (k-1)$. Thus, we have $|Z|\cdot \Delta(J) \geq \Delta(J_1) + \Delta(J_2) + \cdots + \Delta(J_k) - (k-1)$, which finishes the proof of the lemma.
	\end{proof}

	\reduceinversions*
	\begin{proof}
		W.l.o.g we assume $a \Aless b$ and $\pi(a) > \pi(b)$. Let $M$ and $M'$ be the number of inversions in $\pi$ and $\pi'$ respectively.  $\{a, b\}$ contributed 1 to $M$ but not to $M'$.  Let $c \in A \setminus \{a, b\}$ we are going to consider how $\{a, c\}$ and $\{b, c\}$ contribute to $M$ and $M'$.  
		
		If $c \Aless a \Aless b$ or $a \Aless b \Aless c$, then $\{a, c\}$ is an inversion in $\pi$ if and only if $\{b, c\}$ is an inversion in $\pi'$, and $\{b, c\}$ is an inversion in $\pi$ if and only if $\{a, c\}$ is an inversion in $\pi'$.  In the two cases they contribute the same to $M$ and $M'$. 
		
		So assume we are not in the two cases; that is, we have $c \notAless a$ and $b\notAless c$.   Focus on the pair $\{a, c\}$. If we can not compare $a$ and $c$ using the order $\Aless$ then $\{a, c\}$ does not contribute to $M$ and $M'$. Otherwise we have $a \Aless c$. Since $\pi(a) > \pi(b) = \pi'(a)$, then $\pi'(a) > \pi'(c)$ implies $\pi(a) > \pi(c)$. Therefore if $\{a, c\}$ contributed 1 to $M'$, it must have contributed 1 to $M$ as well.   We can make the same argument for $\{b, c\}$. 
		
		Therefore we showed that the contribution of $\{a, c\}$ and $\{b, c\}$ makes to $M'$ is at most that to $M$. Therefore, $M' \leq M - 1$; that is, the number of inversions in $\pi'$ is at most the number of inversions in $\pi$ minus 1. 
	\end{proof}

%% file: improved.tex
\section{Improved Running Time: Modified Dyadic Systems, Valid and Virtually-Valid Schedules}
\label{sec:improved}

In this section and the next one, we show how the running time of the algorithm can be improved to $n^{O\left(\frac{m^4}{\epsilon^3}\log^3\log n\right)}$.  We first give some intuition on how this is done. Consider Lemma~\ref{lemma:valid-to-virtually-valid} on the existence of the virtually-valid schedule $\sigma'^*$ on $\bbJ^*$ with a few discarded jobs.  In the proof, we discard at most $m\cdot \max\{2^{1-h}|I|, 2^{1+h}\}$ jobs in $J^*_{I}$ for each top interval $I$.   For simplicity, let us only consider the case where $2^{-h}|I| \geq 2^h$. Roughly speaking, the $2^{-h}|I|$ term comes from the precision of the $b^*_j$ and $e^*_j$ values for $j \in J^*_I$: they are multiplies of $2^{-h}|I|$ and thus we do not need to know how jobs are assigned $h$ levels below the level of $I$. 

If we relate the number of jobs discarded in $J^*_{I}$ to $|J^*_{I}|$ (instead of $|I|$), then we can afford to discard $\Omega(\frac{\epsilon}{m}|J^*_{I}|)$ jobs in $J^*_I$ since $J^*_I$'s are disjoint.  With this in mind, we do not need to partition $I$ into sub-intervals of length $2^{-h}|I|$.  Instead, we only partition $I$ into a collection $\calC$ of intervals using integer multiplies of $2^{-h}|I|$ as cutting points, such that every interval in $\calC$ either has length $2^{-h}|I|$ or covers at most  $\Theta(\frac{\epsilon}{m}|J^*_{I}|)$ jobs in $J^*_I$, where a job is covered by $I'$ if its owning interval is a subset of $I'$.   Now we define $b^{\bbJ^*}_j$ and $e^{\bbJ^*}_j$ so that there are the cutting points we used to from $\calC$. One can modify the proof of Lemma~\ref{lemma:valid-to-virtually-valid} slightly to show that  the number of discarded jobs is still $O(\epsilon T)$.   But now $\calC$ only needs to contain $O(m/\epsilon)$ intervals and thus we have much less information to guess; this saves us a factor of $\log n$ in the exponent of the running time. Of course we have to guess the partition $\calC$ itself but there are not too many possibilities for $\calC$. At the same time, the number of different $b^\bbJ_j$ and $e^\bbJ_j$ values over all jobs $j$ assigned to an interval $I$ in a dyadic system $\bbJ$ is reduced to $O_{m, \epsilon}(1)$. This bound was used in the proof of Lemma~\ref{lemma:virtually-valid-to-valid} and eventually contributed to the exponent in the running time. Thus with the improvement, we can further remove the other factor of $\log n$ in the exponent.  A small technicality is that the partition $\calC$ for $I$ should be a refinement of the partition used for its parent, which incurs a factor of $h = O(\log \log T)$ in the size of $\calC$.  

For the sake of formality, we walk through the whole analysis again, but omit some details if they can be easily extended from the arguments for the basic algorithm.  

\paragraph{Global Parameters} We use a different set of global parameters now. $h = \ceil{\log \frac{16m \log T}{\epsilon}} = \log \log T + \log \frac{m}{\epsilon} + O(1)$, $L = \log T - h$, $h' = \ceil{\log\frac{16m}{\epsilon}}$, $\delta = \frac{\epsilon}{8shm^2}, \delta' = \frac{1}{2^hsh}$, $p = \floor{\frac2\delta\ln\frac{m}{\delta'}} + 1, s= \frac{16m}{\epsilon}$ (we assume $8m/\epsilon$ is an integer) and $\rho \leq \frac{\epsilon}{(Lsh)^2}$ is the largest number such that $1/\rho$ is an integer power of $2$.

Above, $h, L$ and $h'$ are defined the same as they were in the basic algorithm, except  with different constants.  $\delta$ and $\delta'$ are defined differently. In particular, $\delta$ is now of order $\Theta_{m,\epsilon}(\log \log T)$. As a result, $p = \floor{\frac2\delta\ln\frac{m}{\delta'}} + 1$ now becomes of order $O_{m, \epsilon}(\log \log T)$, which will lead to our improved running time.  The difference in the definition of $\delta'$ is not crucial since eventually only $\log (1/\delta')$ will appear in the exponent of the running time.   $s$ and $\rho$ are new variables we introduce in the improved algorithm. 

\subsection{Modified Dyadic Systems}
	This section corresponds to Section~\ref{sec:dyadic} in the basic algorithm. We define what is a (partial) \emph{modified dyadic system}. In addition to the parameters specified in Definition~\ref{def:dyadic-system} for a (partial) dyadic system, we further need to specify a set $S_I$ for every $I \in \calI_\sftop[I^*]$; this defines  the partition for $I$ we use to decide the $b^\bbJ_j$ and $e^\bbJ_j$ values of jobs assigned to $I$.
	\begin{definition}
		\label{def:modified-dyadic-system}[Counterpart of Definition~\ref{def:dyadic-system}]
		Given an interval ${I^*} \in \calI$, a \emph{{partial modified dyadic} system} $\bbJ$ over ${I^*}$ is a tuple $\big({J_\anc},  b^\anc \in [0, T]^{J_\anc}, e^\anc \in [0, T]^{J_\anc}, (S_{I})_{I \in \calI_\sftop[I^*]}, (J_I)_{I \in \calI[I^*]}\big)$ satisfying all the four properties in Definition~\ref{def:dyadic-system} (Property~\ref{property:dyadic-disjoint}-\ref{property:dyadic-precedence}) and 
		\begin{enumerate}[label=(\ref{def:modified-dyadic-system}\alph*)]
			\item for every $I \in \calI_\sftop[I^*]$, we have $\sfcenter(I) \in S_I \subseteq (\sfbegin(I), \sfend(I))$,  $|S_I| \leq s - 1$ and every number in $S_I$ is an integer multiply of $\max\{2^{-h}|I|, 2^h\}$.  \label{property:modified-dyadic-S}
		\end{enumerate}	
		Jobs being assigned to intervals, owning intervals, top, bottom and ancestor jobs are defined as in Definition~\ref{def:dyadic-system}.
		Again, we simply say $\bbJ$ is a modified dyadic system, if $I^* = [T]$, ${J_\anc} = \emptyset$ and $J_{\subseteq [T]} = J^\circ$. We use $\bbJ = ((S_{I})_{I \in \calI_\sftop}, (J_I)_{I \in \calI})$ to denote a modified dyadic system.
	\end{definition}
	
	The following observation is easy to see since each $S_{\hat I}$ in a partial modified dyadic system only contains integer multiplies of $2^{-h} |\hat I|$:
	\begin{obs}
		\label{obs:cutting-small}
		Assume $(S_I)_{I \in \calI_\sftop[I^*]}$ satisfies Property~\ref{property:modified-dyadic-S}.  If $(\sfbegin(I), \sfend(I)) \cap S_{\hat I} \neq \emptyset$ for some $I \in \calI[I^*]$ and some ancestor $\hat I \in \calI_\sftop[I^*]$ of $I$, then $\hat I$ is at most $h-1$ levels above $I$. 
	\end{obs}

	A valid-schedule for $\bbJ$ is defined in the same way as it was in the basic algorithm, since the sets $S_I$ do not play a role in the definition:
	\begin{definition}[Counterpart of Definition~\ref{def:valid}]
		\label{def:valid-modifed}
		Given a modified {dyadic} system $\bbJ = ((S_I)_{I \in \calI_\sftop}, (J_I)_{I \in \calI})$, a vector $\sigma \in \big({[T]} \cup \{\discarded\}\big)^{J^\circ}$  is said to be a valid schedule for $\bbJ$, if it  is valid to the input instance, and satisfies the interval constraints as in Definition~\ref{def:valid}.
	\end{definition}
		
\subsubsection{Modified Dyadic System from the Optimum Solution}
	This section corresponds to Section~\ref{subsec:optimum-dyadic} for the basic algorithm. We define a modified dyadic system $\bbJ^*$ from the valid schedule $\sigma^* \in [T]^{J^\circ}$ without discarded jobs.  As before, we construct $(J^*_I)_{I \in \calI}, (K^*_I:= J^*_{\subseteq I})_{I \in \calI}$ and $(g^*_I)_{I \in \calI_\sftop}$ by running $K^*_{[T]} \gets J^\circ$ and $\cstrJStar([T])$ (as in Algorithm~\ref{alg:cstrJstar}). The procedure $\pushdown$ is defined as in Algorithm~\ref{alg:pushdown}. Our modified system $\bbJ^*$ is then $\big((J^*_I)_{I \in \calI}, (S^*_I)_{I \in \calI_\sftop}\big)$ for some vector $(S^*_I)_{I \in \calI_\sftop}$ of sets. Then $\sigma^*$ is a valid schedule for $\bbJ^*$. 
	
	It remains  to specify the sets $S^*_I$ for all top intervals $I$. They are defined so that the number of discarded jobs is small when we convert $\sigma^*$ to a virtually-valid schedule $\sigma'^*$ for $\bbJ^*$ later. Focus on each $I \in \calI_\sftop$. We construct a partition $\calW_\sfL$ of $\sfleft(I)$ into many sub-intervals (which are not necessarily in $\calI$) using integer multiplies $\max\{2^{-h}|I|, 2^h\}$ as cutting points, via the following procedure. Initially, the partition $\calW_\sfL$ is the most refined one: It contains $\min\{2^{h-1}, 2^{-1-h}|I|\}$ intervals of length $\max\{2^{-h}|I|, 2^h\}$. While there are two adjacent intervals $W, W'$ in $\calW_\sfL$ such that $\big|J^*_I \cap \sigma^{*-1}(W \cup W') \big| \leq \frac{\epsilon}{8m}|J^*_I|$,  we merge $W$ and $W'$ in the $\calW_\sfL$. The procedure ends when no such $W$ and $W'$ can be found. 

We now make two simple observations about the partition $\calW_\sfL$. First, every interval in $W \in \calW_\sfL$ has either $|W| = \max\{2^{-h}|I|, 2^h\}$ or $\big|J^*_I \cap \sigma^{*-1}(W)\big| \leq \frac{\epsilon}{8m}|J^*_I|$: If $|W| >\max\{2^{-h}|I|, 2^h\}$, then it must be constructed by merging two intervals and the merging can be done only if $\big|J^*_I \cap \sigma^{*-1}(W)\big| \leq \frac{\epsilon}{8m}|J^*_I|$.  Second, every two adjacent intervals $W$ and $W'$ in $\calW_\sfL$ has  $\big|J^*_I \cap \sigma^{*-1}(W \cap W')\big| > \frac{\epsilon}{8m}|J^*_I|$ since otherwise they would have been merged.  This implies that $|\calW_\sfL| \leq 2 \cdot \ceil{\frac{|J^*_I \cap \sigma^{*-1}(\sfleft(I))|}{\epsilon|J^*_I|/(8m)}} - 1$.

Then we use the same procedure to obtain a partition $\calW_\sfR$ of $\sfright(I)$. Combining $\calW_\sfL$ and $\calW_\sfR$ gives us a partition $\calW$ of $I$. Then our $S^*_I$ contains the $|\calW|-1$ cutting points that form $\calW$. 
\begin{claim}
	\label{claim:properties-of-Sstar-I}
	For every $I \in \calI_\sftop$, the following holds.
	\begin{enumerate}[label=(\ref{claim:properties-of-Sstar-I}\alph*)]
		\item $S^*_I$ only contains integer multiplies of $\max\{2^{-h}|I|, 2^h\}$ in $(\sfbegin(I), \sfend(I))$, and $\sfcenter(I) \in S^*_I$.
		\item For every two adjacent numbers $t < t'$ in $S^*_I$, either $t' - t = \max\{2^{-h}|I|, 2^h\}$, or $\big| J^*_I \cap  \sigma^{*-1}((t', t])\big| \leq \frac{\epsilon}{8m}|J^*_I|$. \smallskip
		\item $|S^*_I| \leq \frac{16m}{\epsilon} - 1 = s - 1$.
	\end{enumerate}
\end{claim}
\begin{proof}
	The first two statements are easy to see. Note that $|\calW_\sfL| \leq 2 \cdot \ceil{\frac{|J^*_I \cap \sigma^{*-1}(\sfleft(I))|}{\epsilon|J^*_I|/(8m)}} - 1$  and $|\calW_\sfR| \leq 2 \cdot \ceil{\frac{|J^*_I \cap \sigma^{*-1}(\sfright(I))|}{\epsilon|J^*_I|/(8m)}} - 1$. Adding the two inequalities we have $|\calW|\leq 2\cdot \left(\frac{|J^*_I|}{\epsilon|J^*_I|/(8m)}+1\right) - 2  = \frac{16m}{\epsilon} =s$. This implies $|S^*_I| \leq  s-1$.
\end{proof}

\subsection{Virtually Valid Solutions }
	This section corresponds to Section~\ref{subsec:virtually-valid}.  We assume we are given a partial modified dyadic system $\bbJ = \big({J_\anc},  b^\anc \in [0, T]^{J_\anc}, e^\anc \in [0, T]^{J_\anc}, (J_I)_{I \in \calI[I^*]}, (S_{I})_{I \in \calI_\sftop[I^*]}\big)$; the definitions and lemmas are w.r.t this $\bbJ$. We define a window $(b^\bbJ_j, e^\bbJ_j]$ for each top job $j$ in $\bbJ$. The definition depends on the sets $S_I$, which is a key difference between the improved algorithm and the basic one.   We also define an ``extended window'' $(b'^\bbJ_j, e'^\bbJ_j]$ for $j$.
	\begin{definition}[$b^\bbJ_j$, $b'^\bbJ_j, e^\bbJ_j$ and $e'^\bbJ_j$ values]
		For every top job $j \in J_I$ for some $I \in \calI_\sftop$, we define the window and extended window for $j$ to be $(b^\bbJ_j, e^\bbJ_j]$ and $(b'^\bbJ_j, e'^\bbJ_j]$, where
		\begin{itemize}
			\item $b^\bbJ_j$ is the minimum number $b \in (\sfbegin(I), \sfcenter(I)] \cap S_{\supseteq I}$ such that there are no precedence constraints from $J_{\subseteq (b, \sfcenter(I)]}$ to $j$,
			\item $e^\bbJ_j$ is the maximum number $e \in [\sfcenter(I), \sfend(I)) \cap S_{\supseteq I}$ such that there are no precedence constraints from $j$ to $J_{\subseteq (\sfcenter(I), e]}$, 
			\item $b'^\bbJ_j$ is the maximum number in $((\sfbegin(I), \sfcenter(I)) \cap S_{\supseteq I} \cup \{\sfbegin(I)\}$ that is smaller than $b^\bbJ_j$, and
			\item $e'^\bbJ_j$ is the minimum number in $(\sfcenter(I), \sfend(I)) \cap S_{\supseteq I} \cup \{\sfend(I)\}$ that is larger than $e^\bbJ_j$.
		\end{itemize}
	\end{definition}
	Notice that $\sfcenter(I) \in S_I$, so $b^\bbJ_j$ and $e^\bbJ_j$ are well-defined. As $b^\bbJ_j > \sfbegin(I)$ and $e^\bbJ_j < \sfend(I)$ and thus $b'^\bbJ_j$ and $e'^\bbJ_j$ are also well-defined.  If we assume that $S_{\supseteq I}$ contain all integer multiples of $\max\{2^{-h}|I|, 2^h\}$ in $(\sfbegin(I), \sfend(I))$ (though this is impossible for large enough $T$ by Observation~\ref{obs:cutting-small}), then $b^\bbJ_j$ and $e^\bbJ_j$ coincide with the $b^\bbJ_j$ and $e^\bbJ_j$ in Definition~\ref{def:b-e}, $b'^\bbJ_j = b^\bbJ_j - \max\{2^{-h}|I|, 2^h\}$ and $e'^\bbJ_j = e^\bbJ_j + \max\{2^{-h}|I|, 2^h\}$.
	
	Claim~\ref{claim:simple-properties-of-b-e} still holds with the same proof and we simply copy it here:
	\begin{claim}[Copy of Claim~\ref{claim:simple-properties-of-b-e}]
			\label{claim:simple-properties-of-b-e-modified}
			For any top job $j \in J_I, I \in \calI_\sftop$, we have $\sfbegin(I) < b^\bbJ_j \leq \sfcenter(I)\leq e^\bbJ_j < \sfend(I)$. Moreover, there are no precedence constraints from $j$ to $J_{\subseteq (0, e^\bbJ_j]}$, or from $J_{\subseteq (b^\bbJ_j, T]}$ to $j$.
		\end{claim}
	
	Claim~\ref{lemma:b-e-respects-precedence} also holds, but requires a slight change in the proof to accommodate the new definition of windows.
	\begin{lemma}[Copy of Claim~\ref{lemma:b-e-respects-precedence}]
		\label{lemma:b-e-respects-precedence-modified}
		Let $j$ and $j'$ be two top jobs with $j \prec j'$.  Then we have $b^\bbJ_j \leq b^\bbJ_{j'}$ and $e^\bbJ_j \leq e^\bbJ_{j'}$. 
	\end{lemma}
	\begin{proof}
		Assume $j \in J_I$ and $j' \in J_{I'}$ for some $I, I' \in \calI_\sftop[I^*]$.  The analysis for the  cases where $I$ and $I'$ are disjoint and where $I = I'$ is the same as that in the proof of Lemma~\ref{lemma:b-e-respects-precedence}.  So, as before,  we only need to consider the case that $I \subseteq \sfleft(I')$ and only need to prove $b^\bbJ_j \leq b^\bbJ_{j'}$. Again refer to Figure~\ref{fig:proof-of-time} for illustration of time points used in this proof.
		
		If $b^\bbJ_{j'} \geq \sfcenter(I)$, then we have $b^\bbJ_j \leq \sfcenter(I) \leq b^\bbJ_{j'}$ and we are done. So, assume $b^\bbJ_{j'} < \sfcenter(I)$.  Since $j \in J_I \subseteq J\left[(\sfbegin(I), \sfcenter(I')]\right]$ and $j \prec j'$, we have $b^\bbJ_{j'}  > \sfbegin(I)$ by its definition. So $b^\bbJ_{j'}$ is an integer in $S_{\supseteq I'}\cap \big(\sfbegin(I), \sfcenter(I)\big)$. By the definition of $b^\bbJ_{j'}$, there are no precedence constraints from $J\left[(b^\bbJ_{j'}, \sfcenter(I')]\right]$ to $j'$. Since $J\left[(b^\bbJ_{j'}, \sfcenter(I')]\right] \supseteq J\left[(b^\bbJ_{j'}, \sfcenter(I)]\right]$, there will be no precedence constraints from $J\left[(b^\bbJ_{j'}, \sfcenter(I)]\right]$ to $j'$, implying that there will be no such constraints to $j$ as well.   As $b^\bbJ_{j'}$ is an integer in $S_{\supseteq I'} \subseteq S_{\supseteq I}$ strictly between $\sfbegin(I)$ and $\sfcenter(I)$, we have that $b^\bbJ_j \leq b^\bbJ_{j'}$ by its definition.
	\end{proof}

	The following claim is new in the improved result:
	\begin{claim}
		\label{coro:possibilities-for-b-e}
		For every $I \in \calI_\sftop$, we have that $\big|\{b^\bbJ_j: j \in J_I\}\big| + \big|\{e^\bbJ_j: j \in J_I\}\big| \leq sh$.
	\end{claim}
	\begin{proof}
		By Observation~\ref{obs:cutting-small}, we have $|S_{\supseteq I} \cap (\sfbegin(I), \sfend(I))| \leq (s-1)h$. Noticing that both $\{b^\bbJ_j: j \in J_I\}$ and $\{e^\bbJ_j: j \in J_I\}$ are subsets of $S_{\supset I} \cap (\sfbegin(I), \sfend(I))$ and can only share the element $\sfcenter(I)$, we have $\big|\{b^\bbJ_j: j \in J_I\}\big| + \big|\{e^\bbJ_j: j \in J_I\}\big| \leq (s-1)h  + 1 \leq sh$.
	\end{proof}
				
	With the windows for top jobs defined, we can then define a virtually-valid schedule for the partial modified dyadic system $\bbJ$ exactly as in Definition~\ref{def:virtually-valid}.
	
	\begin{definition}
		\label{def:virtually-valid-modified} The definition of a virtually-valid schedule for $\bbJ$ is the same as Definition~\ref{def:virtually-valid}.
	\end{definition}
	 
\subsection{Conversion between Valid and Virtually-Valid Schedules}
	This section corresponds to Section~\ref{sec:conversion} for the basic algorithm.  First, we show there is a virtually valid schedule $\sigma'^*$ for the modified dyadic system $\bbJ^*$, with a small number of discarded jobs. Then we show that given any valid schedule $\sigma''$ for a modified dyadic system $\bbJ$, we can efficiently convert it to a valid schedule with a small number of extra discarded jobs. 
	
	\subsubsection{From $\sigma^*$ to a Virtually-Valid Schedule $\sigma'^*$ for $\bbJ^*$} 
	This section corresponds to Section~\ref{subsec:valid-to-virtually-valid} for the basic algorithm.
	In the analysis for the existence of $\sigma'^*$ is different from that in the basic algorithm,  we need to use the new definition of windows.  For simplicity, we use $b^*_j, b'^*_j, e^*_j$ and $e'^*_j$ for a top job $j$ in $\bbJ^*$ to denote $b^{\bbJ^*}_j, b'^{\bbJ^*}_j, e^{\bbJ^*}_j$ and $e'^{\bbJ^*}_j$ respectively. 
	
	First, we show a lemma correspondent to Lemma~\ref{lemma:in-extended-window}. 
	\begin{lemma}[Counterpart of Lemma~\ref{lemma:in-extended-window}]
		\label{lemma:in-extended-window-modified}
		For every top job $j \in J^*_I$ for some $I \in \calI_\sftop[I^*]$, we have $\sigma^*_j\in (b'^*_j, e'^*_j]$.
	\end{lemma}
	\begin{proof}
		We prove that $\sigma^*_j > b'^*_j$.  By the definition of $b^*_j$, either $b^*_j$ is the smallest number in $(\sfbegin(I), \sfcenter[I]] \cap S_{\supseteq I}$ or $b^*_j$ is larger than the number and there is a job $j' \in J\left[(b'^*_j, \sfcenter(I)]\right]$ such that $j' \prec j$ (since otherwise $b^*_j$ would be smaller).  In the former case, $b'^*_j = \sfbegin(I)$ and the inequality follows from Claim~\ref{claim:simple-properties-of-b-e-modified}.  In the latter case, we must have $\sigma^*_{j'} > b'^*_j$. Since $j' \prec j$, we have $\sigma^*_j >b'^*_j$. Similarly, we can show that $\sigma^*_j\leq e'^*_j$. 
	\end{proof}	
	
	Now we show there is a virtually valid schedule $\sigma'^*$ for $\bbJ^*$ with a small number of jobs discarded.
	\begin{lemma}[Counterpart of Lemma~\ref{lemma:valid-to-virtually-valid}]
		\label{lemma:valid-to-virtually-valid-modified}
		There is a virtually-valid schedule $\sigma'^*$ for $\bbJ^*$ with at most $\frac{5\epsilon T}{8}$ jobs discarded.
	\end{lemma}
	\begin{proof} 
		Again as in the proof of Lemma~\ref{lemma:valid-to-virtually-valid}, we schedule the bottom jobs in ${\sigma'^*}$ in exactly the same way as they are in $\sigma^*$. It remains to show how to schedule  top jobs in $\bbJ^*$ to satisfy capacity and window constraints. 

		Similarly we fix some $I \in \calI_\sftop$ and focus on the set $J^*_I$ of jobs assigned to $I$ in $\bbJ^*$. We schedule $J^*_I \cap \sigma^{*-1}(\sfleft(I))$ in $\sfleft(I)$ and $J^*_I \cap \sigma^{*-1}(\sfright(I))$ in $\sfright(I)$ in the schedule $\sigma'^*$. We only show how to schedule $J^*_I \cap \sigma^{*-1}(\sfleft(I))$. We define $\sfcap(t) := |J^*_I  \cap {\sigma^{*-1}}(t) \big|$ to be number of available slots at time $t$, for every $t \in \sfleft(I)$. In $\sigma'^*$,  we only schedule $J^*_I \cap {\sigma^{*-1}}(\sfleft(I))$ using the available slots.
		
		Here comes a difference between this proof and the proof of Lemma~\ref{lemma:valid-to-virtually-valid}: We let $\calC$ be the partition of $\sfleft(I)$ using the points in  $(\sfbegin(I), \sfcenter(I)) \cap S^*_{\supseteq I}$. For every  $I' \in \calC$, let $\sfcap(I') := \sum_{t \in I'}\sfcap(t) =  \big|J^*_I  \cap {\sigma^{*-1}}(I')\big|$ be the number of available slots in $I'$.  Then the procedure for scheduling  jobs $J^*_I  \cap {\sigma^{*-1}}(\sfleft(I))$ will be the same as that in the proof of Lemma~\ref{lemma:valid-to-virtually-valid}, except that we use a different $\calC$.  Initially, let $\tilde J \gets \emptyset$. For every $I' \in \calC$ from left to right, we do the following: schedule $\min\{|\tilde J|, \sfcap(I')\}$  jobs in $\tilde J$ using the $\sfcap(I')$ available slots in $I'$, remove the scheduled  jobs from $\tilde J$, and add the $\sfcap(I')$ jobs $J^*_I  \cap {\sigma^{*-1}}(I')$ to $\tilde J$. Then we discard $\tilde J$ in the end.  Notice that if some $j \in J^*_I  \cap {\sigma^{*-1}}(I'), I' \in \calC$ is scheduled in ${\sigma'^*}$, then it must be scheduled in some $I'' \in \calC$ to the right of $I'$.  The window constraint for $j$ will be satisfied since by Lemma~\ref{lemma:in-extended-window-modified} we have $b'^*_{j} \leq \sfbegin(I')$, which implies $b^*_{j} \leq \sfbegin(I'')$, and $e^*_j  \geq \sfcenter(I) \geq \sfend(I'')$. 
		
		Again we can show that the number of discarded jobs is the maximum of $\sfcap(I')$ over all $I' \in \calC$.  Notice that by Property~(\ref{claim:properties-of-Sstar-I}b), this is at most $\max\big\{2^{-h}|I|m, 2^hm, \frac{\epsilon|J^*_I|}{8m}\big\}$. Considering jobs in $J^*_I \cap \sigma^{*-1}(\sfright(I))$, we discarded at most $\max\big\{2^{1-h}|I|m, 2^{1+h}m, \frac{\epsilon|J^*_I|}{4m}\big\} \leq 2^{1-h}|I|m + 2^{1+h}m +  \frac{\epsilon|J^*_I|}{4m}$ jobs in $J^*_I$. 
		
		We then bound $\sum_{I \in \calI_\sftop}2^{1-h}|I|m + \sum_{I \in \calI_\sftop}2^{1+h}m + \sum_{I \in \calI_\sftop} \frac{\epsilon|J^*_I|}{4m}$. As in the proof of Lemma~\ref{lemma:virtually-valid-to-valid}, the first term is at most $2^{1-h}mLT \leq \frac{2mLT\cdot\epsilon}{16m\log T} \leq \frac{\epsilon T}{8}$, the second term is at most $4mT\cdot 2^{-h'} \leq \frac{4mT\cdot \epsilon}{16m} = \frac{\epsilon T}{4}$.  The third term is at most $\frac{\epsilon|J^\circ|}{4m} \leq \frac{\epsilon mT}{4m} \leq \frac{\epsilon T}{4}$. Thus, the sum of the three terms is at most $\frac{5\epsilon T}{8}$.
	\end{proof}
	
	To obtain a better running time, we need the virtually-valid schedule for $\bbJ^*$ to satisfy more property stated in the following lemma. 
	\begin{lemma}
			\label{lemma:S''-star}
			There is a virtually valid schedule $\sigma''^*$ for $\bbJ^*$ with at most $\frac{3\epsilon T}{4}$ jobs discarded. Moreover, 
			\begin{enumerate}[label=(\ref{lemma:S''-star}\alph*)]
				\item for every $ I^* \in \calI_{\sftop} \cup \calI_{\sfmid}$ and $W \subseteq I^*$, we have $|\{j \in J^*_{\supsetneq I^*} \cap \sigma''^{*-1}(I^*): (b^*_j, e^*_j] \cap I^* = W\}|$ is an integer multiply of $\rho|I^*|$.
			\end{enumerate}
	\end{lemma}
	The lemma says that if we consider all the jobs $j$ assigned to strict ancestors of $I^*$ and scheduled in $I^*$ in $\sigma''^*$, and group them according to $(b^*_j, e^*_j] \cap I^*$, then the cardinality of each group is an integer multiply of $\rho|I^*|$; recall that $\frac{\epsilon}{2(Lsh)^2} < \rho \leq \frac{\epsilon}{(Lsh)^2}$ is an integer power of $2$.
		\begin{proof}[Proof of Lemma~\ref{lemma:S''-star}]
			Define $L' = \log T - \log (1/\rho)$. So for an interval $I \in \calI_{L'}$, we have $\rho |I| = \rho T 2^{-L'} = 1$. Then we only need to guarantee Property~(\ref{lemma:S''-star}a) for $I \in \calI_{<L'}$.  For every interval $\hat I \in \calI_{< L'}$, every interval $W \subseteq \hat I$, let $J^*_{\hat I}(W):=\{j \in J^*_{\hat I}: (b^*_j, e^*_j] = W\}$. It suffices to guarantee the following condition (*): 
			\begin{enumerate}[label=(*)]
				\item For every $\hat I \in \calI_{\leq L'-1}, W \subseteq \hat I$ and strict descendant $I^* \in \calI_{\leq L'}$ of $\hat I$, we have $|J^*_{\hat I}(W) \cap \sigma''^{*-1}(I^*)|$ is an integer multiply of $\rho|I^*|$.
			\end{enumerate}
			To see why (*) implies Property~(\ref{lemma:S''-star}a), notice that the set $\{j \in J^*_{\supsetneq I^*} \cap \sigma''^{*-1}(I^*): (b^*_j, e^*_j) \cap I^* = W\}$ is the disjoint union of $J^*_{\hat I}(W') \cap \sigma''^{*-1}(I^*)$ over all $\hat I \in \calI, \hat I \supsetneq I$ and $W' \subseteq \hat I$ with $W' \cap I^* = W$.

			We shall show how to construct schedule $\sigma''^*$ so as to satisfy (*), by discarding some top jobs in $\sigma'^*$. Let $\sigma''^* = \sigma'^*$ initially.  We fix any $\hat I \in \calI_{<L'}$ and a sub-interval $W \subseteq \hat I$ such that $J^*_{\hat I}(W) \neq \emptyset$. For every strict descendant $I^* \in \calI_{\leq L'}$ of $\hat I$ from bottom to top, we guarantee (*) for $I^*$ (and the fixed $\hat I$ and $W$) in that order.  The condition holds for intervals $I^*$ at level $L'$ since for such intervals we have $\rho|I^*| = 1$.  Now focus on an interval $I^*$ above level $L'$ and assume (*) holds for all descendants $I$ of $I^*$. In particular, $|J^*_{\hat I}(W) \cap \sigma''^{*-1}(\sfleft(I^*))|$ and $|J^*_{\hat I}(W) \cap \sigma''^{*-1}(\sfright(I^*))|$ are integer multiplies of $\rho|I^*|/2$, and so is $|J^*_{\hat I}(W) \cap \sigma''^{*-1}(I^*)|$. If $|J^*_{\hat I}(W) \cap \sigma''^{*-1}(I^*)|$ is not an integer multiply of $\rho|I^*|$, we need to discard $\rho|I^*|/2$ jobs in the set without violating condition (*) for descendants of $I^*$. This can be done using a simple recursive procedure: we distributed  jobs that we need to discard recursively to sub-intervals $I$ of $I^*$ in the tree $\bfT$, guaranteeing the invariant that the number of jobs we discard in $I$ is an integer multiply of $\rho|I|$ less than or equal to $|J^*_{\hat I}(W) \cap \sigma''^{*-1}(I)|$.  The recursions stop at level $L'$ for which we can discard the specified number of jobs directly. 
			
			
			We count how many jobs we discarded. First fix a $\hat I$ and $W$. To guarantee (*) for $I^* \in \calI[\hat I] \cap \calI_{\geq L'}$, we discarded at most $\rho|I^*|/2$ jobs. So, summing up the bound over all $I^*$, we discarded at most $\rho L |\hat I|/2$ jobs.  For a fixed $\hat I$, by Corollary~\ref{coro:possibilities-for-b-e}, there are at most $(sh)^2/4$ different windows $W$ for which $J^*_{\hat I}(W) \neq \emptyset$.  Therefore, summing up the bound over all $W$ gives us an upper bound of $\rho L(sh)^2|\hat I|/8$. Then summing up the bound over all $\hat I \in \calI_{\leq L'-1}$, we have that we discarded at most $\rho (Lsh)^2 T/8 \leq \frac{\epsilon}{(Lsh)^2} \cdot (Lsh)^2 T/8 \leq \frac{\epsilon T}{8}$ jobs from $\sigma'^*$ to $\sigma''^*$, by the definition of $\rho$. Thus, counting the discarded jobs in $\sigma'^*$, the number of discarded jobs in $\sigma''^*$ is at most $\frac{5\epsilon T}8 + \frac{\epsilon T}{8} = \frac{3\epsilon T}{4}$. 
		\end{proof}
	
	\subsubsection{Converting a Virtually-Valid Schedule to a Valid One}
	This section corresponds to Section~\ref{subsec:virtually-valid-to-valid} for the basic algorithm: we show that given a modified dyadic system $\bbJ = ((J_I)_{I \in \calI}, (S_I)_{I \in \calI_\sftop})$ and a virtually-valid schedule $\sigma''$ for $\bbJ$, we can efficiently construct a valid schedule $\sigma$ for $\bbJ$ with a small number of extra discarded jobs.  Almost all the arguments in Section~\ref{subsec:virtually-valid-to-valid} still hold, since we are not using the properties of $b^\bbJ_j$'s and $e^\bbJ_j$'s other than those stated in Claim~\ref{claim:simple-properties-of-b-e-modified} and Lemma~\ref{lemma:b-e-respects-precedence-modified}.    For Lemma~\ref{lemma:virtually-valid-to-valid}, we need to use the bound in Claim~\ref{coro:possibilities-for-b-e} for the number of different $b^\bbJ_j$ and $e^\bbJ_j$ values in a set $J_I$. This gives us a better bound, allowing us to use a larger $\delta$ and thus a smaller $p$. 
	
	We define $J^{\#}$, $J^{\#}_I$ for $I \in \calI_\sftop$, $I(j)$ for $j \in J^{\#}$, $\side_{\sigma'}$ for a schedule $\sigma'$ and $\sigmaless$ in the same way as they were Section~\ref{subsec:virtually-valid-to-valid}.  Lemma~\ref{lemma:canonicalize} still holds with an identical proof (except we need to refer to new versions of claims, lemmas, definitions and properties):
	\begin{lemma}[copy of Lemma~\ref{lemma:canonicalize}]
		\label{lemma:canonicalize-modified}
		We can efficiently find another virtually-valid schedule $\sigma'$ for $\bbJ$ with $\sigma'^{-1}(\discarded) = \sigma''^{-1}(\discarded)$ and $\sigma'_j = \sigma''_j$ for every bottom job $j$ in $\bbJ$. Moreover, for every two jobs $j, j' \in J^{\#}$, the following holds. 
		\begin{enumerate}[label=(\ref{lemma:canonicalize-modified}\alph*)]
			\item If $j \prec j'$, then $\sigma'_{j} \leq \sigma'_{j'}$.
			\item If $j \sigmaless j'$, then $\sigma'_{j} \leq \sigma'_{j'}$.
		\end{enumerate}
	\end{lemma}
	

	We state the counterpart of Lemma~\ref{lemma:virtually-valid-to-valid} and show the difference in the new proof:
	\begin{lemma}[Counterpart of Lemma~\ref{lemma:virtually-valid-to-valid}] 
		\label{lemma:virtually-valid-to-valid-modified}
		Given a virtually-valid schedule $\sigma'$ for $\bbJ$ satisfying  conditions in Lemma~\ref{lemma:canonicalize-modified}, we can efficiently construct a valid schedule $\sigma$ for $\bbJ$ with $|\sigma(\discarded) \setminus \sigma'(\discarded)| \leq \frac{\epsilon T}{4}$.
	\end{lemma}
	\begin{proof}
		We only give the difference between this proof and the proof of Lemma~\ref{lemma:virtually-valid-to-valid}. We can replace Inequality~\eqref{inequ:Delta-left} in the proof to
		\begin{align*}
			&\quad \sum_{I' \in \calI_\sfbot[\sfleft(I)]}\Delta(J^{\#}_I\cap \sigma'^{-1}(I'))
			\leq A_{I,\sfL}(\delta|J_I|+\delta'|I|) + 2^{-{h-1}}|I|,
		\end{align*}
		where $A_{I, \sfL} = |\{b^\bbJ_j: j \in J^{\#}_I\}|$ is the number of different $b^\bbJ_j$ values for $j \in J^{\#}_I$. Similarly, we can prove $\sum_{I' \in \calI_\sfbot[\sfright(I)]}\Delta(J^{\#}_I\cap \sigma'^{-1}(I')) \leq A_{I,\sfR}(\delta|J_I|+\delta'|I|) + 2^{-{h-1}}|I|$, where $A_{I, \sfR} = |\{e^\bbJ_j: j \in J^{\#}_I\}|$.  By Corollary~\ref{coro:possibilities-for-b-e} we have $A_{I, \sfL} + A_{I, \sfR} \leq sh$. Thus, 
		\begin{align*}
			\sum_{I' \in \calI_\sfbot[I]}\Delta(J^{\#}_I\cap \sigma'^{-1}(I'))
			\leq sh(\delta|J_I|+\delta'|I|) + 2^{-h}|I|,
		\end{align*}
		
		Using the same argument as in the proof of Lemma~\ref{lemma:virtually-valid-to-valid}, we can show that the number of extra jobs we discarded from $\sigma'$ to $\sigma$ is at most  $m\sum_{I' \in \calI_\sfbot}\Delta\left(J^{\#}\cap \sigma'^{-1}(I')\right) \leq sh\delta m^2T+ (sh\delta'+2^{-h})mLT = shm^2T\cdot  \frac{\epsilon}{8shm^2} + 2\cdot 2^{-h} mLT = \frac{\epsilon T}{8} + 2\cdot \frac{\epsilon}{16m \log T}\cdot mLT \leq \frac{\epsilon T}{8} + \frac{\epsilon T}{8} = \frac{\epsilon T}{4}$.
	\end{proof}

\section{Improved Algorithm: Recursive Algorithm for Constructing Modified Dyadic System and Virtually-Valid Schedule}
\label{sec:improved-algo}

		This section corresponds to Section~\ref{sec:algo} for the basic algorithm. In order to define our new algorithm, we need  one more definition.
		\begin{definition}
			Given an interval $I \in \calI_\sftop$, and a set $S$ of integers, we define $\calI^\cut_S[I]$ to be the set of intervals $I' \in \calI[I]$ with $(\sfbegin(I'), \sfend(I')) \cap S \neq \emptyset$.  Define $\calI^\uncut_S[I] = \calI[I] \setminus \calI^\cut_S[I]$ to be the set of intervals $I' \in \calI[I]$ with $(\sfbegin(I'), \sfend(I')) \cap S = \emptyset$. Let $\calI^\maxuncut_S[I]$ be the set of inclusion-wise maximal intervals in $\calI^\uncut_S[I]$. 
		\end{definition}
		
		Thus, intervals in $\calI^\cut_S[I]$ are ``cut'' by points in $S$ and intervals in $\calI^\uncut_S[I]$ are not cut by points in $S$.   An interval $I$ is in $\calI^\maxuncut_S[I]$ if it is in $\calI^\uncut_S[I]$ but its parent is not, or $I=I^* \in \calI^\uncut_S[I]$.  It is easy to see that $\calI^\maxuncut_S[I]$ is a partition of $I$.


		
		With the definitions, we can now describe our algorithm $\schedule\mhyphen\modified$ (described in Algorithm~\ref{alg:schedule-modified}) for the improved running time. We skip some arguments if they are easy generalizations of the counterparts for the basic algorithm $\schedule$. In the input, we are given $I^*, {J_\anc}, b^\anc, e^\anc \in [0, T]^{J_\anc}$ as before. We are also given a set ${S_\anc}$ of integers, a set $J_I$ for every $I \in \calI^\cut_{{S_\anc}}[I^*]$ and a set $K_I$ for every $I \in \calI^\maxuncut_{{S_\anc}}[I^*]$.   Thus, other than the set $S_\anc$, a key difference between $\schedule\mhyphen\modified$ and $\schedule$ is that in $\schedule\mhyphen\modified$ we do not have the $J_I$ information for all intervals $I$ in $\calI_{<h-1}[I^*]$. Instead, we only have the information for intervals that are cut by ${S_\anc}$ (this will be a subset of $\calI_{<h-1}[I^*]$ since it will be easy to see that ${S_\anc}$ only contains integer multiplies of $\max\{2^{-h+1}|I^*|, 2^h\}$).  If one assumes ${S_\anc}$ is the set of all integer multiples of $\max\{2^{-h+1}|I^*|, 2^h\}$ between $\sfbegin(I)$ and $\sfend(I)$, then the input parameters between $\schedule$ and $\schedule\mhyphen\modified$ become the same.\footnote{There is a minor notation discrepancy between $\schedule$ and $\schedule\mhyphen\modified$: The bottom intervals are always uncut but they belong to $\calI_{<h-1}[I^*]$ if $I^*$ is at below level $L-h+1$.} The output of $\schedule\mhyphen\modified$ is the same as that for $\schedule$, except that the modified dyadic system $\bbJ^\best$ now needs to contain $S_\anc$ (which is given as input) and $(S^\best_I)_{I \in \calI_\sftop[I^*]}$ (which is produced by the algorithm).  

		\begin{algorithm}[t!] 
			\caption{$\schedule\mhyphen\modified\big(I^*,  {J_\anc},    \tilde b^\anc,  \tilde e^\anc, S_\anc, (\tilde J_{I})_{I \in \calI^\cut_{{S_\anc}}[I^*]}, (\tilde K_{I})_{I \in \calI^{\maxuncut}_{{S_\anc}}[I^*]}\big)$}\label{alg:schedule-modified}
			\textbf{Output:} a partial modified dyadic system $\bbJ^\best$ and a virtually-valid schedule $\sigma^\best$ for $\bbJ$
			\begin{algorithmic}[1]
				\If{$|{J_\anc}| > m|I^*|$ or  $|\tilde J_{\subseteq I^*}| + |\tilde K_{\subseteq I^*}| > m|I^*|$} \Return $(\bot, \bot)$ \EndIf  \label{step:schedule-modified-check-size}
				\If{$I^* \in \calI_\sfbot$}
					\Return $\big(\bbJ^\best := (\tilde J_\anc, \tilde b^\anc, \tilde e^\anc, S_\anc, (), (\tilde K_{I^*})), \text{best virtually-valid schedule $\sigma^\best$ for } \bbJ^\best\big)$
				    \label{step:schedule-modified-terminating-start}
				\EndIf \vspace*{-10pt}
				\State $\sigma^\best \gets \bot$, copy $\tilde b^\anc$, $\tilde e^\anc$, $\tilde J_I$'s and $\tilde K_I$'s to $b^\anc, e^\anc$, $J_I$'s  and $K_I$'s \label{step:schedule-modified-init}
				\For{every $S_{I^*}$ satisfying Property (\ref{def:modified-dyadic-system}a) if $I^* \in \calI_\sftop$, or $S_{I^*} = \{\sfcenter(I^*)\}$ if $I^* \in \calI_\sfmid$} \label{step:schedule-modified-enumerating-S} 
					\State $\calI' \gets \calI^\cut_{{S_\anc} \cup S_{I^*}}[I^*] \cap \calI^\uncut_{{S_\anc}}[I^*]$ \label{step:schedule-modified-calI'}
					\For{every possible vector $(g_I)_{I \in \calI'}$ s.t. $g_I \in \{\sfL, \sfR\}^p$ if $I \in \calI_\sftop$ and $g_I \in \{\sfL, \sfR\}^{m|I|}$ if $I \in \calI_\sfmid$} \label{step:schedule-modified-enumerate-g}
						\State \textbf{for} every $I \in \calI'$ from top to bottom \textbf{do}  $\big( J_I,  K_{\sfleft(I)},  K_{\sfright(I)}\big) \gets \pushdown(I, K_I, g_{I})$ 
						\label{step:schedule-modified-call-pushdown}
						\For{every $j \in  J_{I^*}$} \label{step:schedule-modified-extend-b-e-up}
							define
							\begin{itemize}[leftmargin=40pt]
								\item $b^\anc_j$ to be minimum integer $b$ in $({S_\anc} \cup S_{I^*})\cap (\sfbegin(I^*),\sfcenter(I^*)]$ such that there are no precedence constraints from $J_{\subseteq (b, \sfcenter(I^*)]} \cup K_{\subseteq (b, \sfcenter(I^*)]}$ to $j$, and
								\item $b^\anc_j$ to be minimum integer $e$ in $({S_\anc}  \cup S_{I^*}) \cap [\sfcenter(I^*),\sfend(I^*))$ such that there are no precedence constraints from $j$ to $J_{\subseteq (\sfcenter(I^*),e]} \cup K_{\subseteq (\sfcenter(I^*),e]}$.
							\end{itemize}
						\EndFor
						\For{every good partition (see Definition~\ref{def:good-partition}) of $ {J_\anc} \cup  J_{I^*}$ into $ J_\anc^\sfL, J_\anc^\sfR$ and $J_\discarded$, keeping only one partition in every good equivalence class defined in Remark~\ref{remark:equivalence}} \label{step:schedule-modified-partition-Janc}
							\State $(\bbJ^\sfL, \sigma^\sfL) \gets \schedule\mhyphen\modified\Big(\sfleft(I^*),  J_\anc^\sfL,    b^\anc|_{ J_\anc^\sfL}, e^\anc|_{ J_\anc^\sfL}, $\newline \hspace*{0.4\textwidth} ${S_\anc} \cup S_{I^*}, (J_{I})_{I \in \calI^\cut_{{S_\anc} \cup S_{I^*}}[\sfleft(I^*)]}, (K_{I})_{I \in \calI^{\maxuncut}_{{S_\anc} \cup S_{I^*}}[\sfleft(I^*)]}\Big)$ \label{step:schedule-modified-recurse-left}
							\State $(\bbJ^\sfR, \sigma^\sfR) \gets \schedule\mhyphen\modified\Big(\sfright(I^*),  J_\anc^\sfR,    b^\anc|_{ J_\anc^\sfR}, e^\anc|_{ J_\anc^\sfR},$ \newline\hspace*{0.38\textwidth}${S_\anc} \cup S_{I^*},  (J_{I})_{I \in \calI^\cut_{{S_\anc} \cup S_{I^*}}[\sfright(I^*)]}, (K_{I})_{I \in \calI^{\maxuncut}_{{S_\anc} \cup S_{I^*}}[\sfright(I^*)]}\Big)$
							 \label{step:schedule-modified-recurse-right}
							\If{$\sigma^\sfL, \sigma^\sfR\neq \bot$ and \big($\sigma^\best=\bot$ or $|(\sigma^\sfL)^{-1}(\sfleft(I^*))|+|(\sigma^\sfR)^{-1}(\sfright(I^*))|  > |(\sigma^{\best})^{-1}(I^*)|$\big)} \vspace*{-10pt}
							\label{step:schedule-modified-check-if-better}
								\State $S^\best_I \gets S_{I}, \forall I \in \calI_\sftop[I^*]$;  $J^\best_{I^*} \gets J_{I^*}$ \label{step:schedule-modified-update-0}
								\State $J^\best_I \gets J^\sfL_I, \forall I \in \calI[\sfleft(I^*)]$; $J^\best_I \gets J^\sfR_I, \forall I \in \calI[\sfright(I^*)]$ \label{step:schedule-modified-update}
								\State let $\sigma^\best$ be obtained by merging $\sigma^\sfL$ and $\sigma^\sfR$ and discard $J_\discarded$	
								\label{step:schedule-modified-update-1}
							\EndIf
						\EndFor
					\EndFor
				\EndFor
				\State \textbf{if} $\sigma \neq \bot$ \textbf{then} \Return $\Big(\bbJ^\best:=\big(J_\anc, \tilde b^\anc,  \tilde e^\anc, S_\anc, (S^\best_I)_{I \in \calI_\sftop[I^*]}, (J^\best_I)_{I \in \calI[I^*]}\big), \sigma^\best\Big)$ \\\hspace*{38pt} \textbf{else} \Return $(\bot, \bot)$
			\end{algorithmic}
		\end{algorithm}
		
	In Step~\ref{step:schedule-modified-check-size} of $\schedule\mhyphen\modified$, we check if both the number of ancestor jobs and the number of jobs assigned to sub-intervals of $I^*$ are at most $m|I^*|$.  We can handle the case $I^*\in \calI_\sfbot$ directly. So we assume $I^* \in \calI_\sftop \cup \calI_\sfmid$. In Step~\ref{step:schedule-modified-init}, we initialize the variables as before. One key difference between $\schedule\mhyphen\modified$ and $\schedule$ comes from Step~\ref{step:schedule-modified-enumerating-S}, where we guess the set $S_{I^*}$: If $I^* \in \calI_\sftop$, then $S_{I^*}$ can be any set satisfying Property~(\ref{def:modified-dyadic-system}a); if $I^* \in \calI_\sfmid$, we fix $S_{I^*} = \{ \sfcenter(I^*)\}$, which is needed to make sure that $I^*$ will be partitioned by its center.  In Step~\ref{step:schedule-modified-calI'}, we construct a set $\calI'  = \calI^\cut_{S_\anc \cup S_{I^*}} \cap \calI^\uncut_{S_\anc}$, the set of new intervals $I$ for which we need to know the set $J_I$. In Loop~\ref{step:schedule-modified-enumerate-g}, we try to guess $g^*_I$ for every $I \in \calI'$ as before. Based on our guesses,  we expand the information about $J_I$'s in Step~\ref{step:schedule-modified-call-pushdown}, by calling $\pushdown$  for every $I \in \calI'$; we shall guarantee that $I$ won't be a bottom interval (See Claim~\ref{claim:schedule-modified-properties}).
	
	Once we have $S_{I^*}$, $J_I$'s for $I \in \calI^\cut_{S_\anc \cup S_{I^*}}$ and $K_I$'s for $I \in \calI^\maxuncut_{S_\anc \cup S_{I^*}}$, we can compute the $b^\bbJ$ and $e^\bbJ$ values for jobs in $J_{I^*}$.  This is done in Step~\ref{step:schedule-modified-extend-b-e-up}: For every $j \in J_{I^*}$, $b^\anc_j$ and $e^\anc_j$ will be the same as the $b^\bbJ_j$ and $e^\bbJ_j$ values for the constructed modified dyadic system $\bbJ$.  
	In Loop~\ref{step:schedule-modified-partition-Janc}, we guess how jobs in ${J_\anc} \cup J_{I^*}$ are split into $\sfleft(I^*)$ and $\sfright(I^*)$. Again, we keep one partition in every equivalence class, but we only consider good partitions:
	\begin{definition} \label{def:good-partition}
		In Step~\ref{step:schedule-modified-partition-Janc} of $\schedule\mhyphen\modified$, we say a partition $(J^\sfL_\anc, J^\sfR_\anc, J_\discarded)$ is good if for every interval $W \subseteq \sfleft(I^*)$, we have $|\{j \in J^\sfL_\anc: (b^\anc_j, e^\anc_j] \cap \sfleft(I^*) = W\}|$ is an integer multiply of $\rho |I^*|/2$, and for every interval $W \subseteq \sfright(I^*)$, we have $|\{j \in J^\sfR_\anc: (b^\anc_j, e^\anc_j] \cap \sfleft(I^*) = W\}|$ is an integer multiply of $\rho |I^*|/2$.
	\end{definition}


	Once we made the guess, we recursively and independently call the $\schedule$ procedure for $\sfleft(I^*)$ and $\sfright(I^*)$ (Step~\ref{step:schedule-modified-recurse-left} and \ref{step:schedule-modified-recurse-right}).  We maintain the best solution constructed so far (Step~\ref{step:schedule-modified-update-0} to~\ref{step:schedule-modified-update-1}) and return it in the end.
	
	In the main algorithm, we simply call $(\bbJ^\best, \sigma^\best) \gets \schedule\mhyphen\modified([T], \emptyset, (), (), \emptyset, (), (K_{[T]} = J^\circ))$ and return $\bbJ^\best$ and $\sigma^\best$.
		
\subsection{Analysis of Correctness}
	\label{subsec:algo-modifiedcorrectness}
	We now analyze the correctness of the algorithm. The following claim gives some simple properties about the input to each recursion of $\schedule$. 
		\begin{claim}[Counterpart of Claim~\ref{claim:schedule-properties}]
			\label{claim:schedule-modified-properties}
			At the beginning of a recursion of $\schedule\mhyphen\modified$ for some $I^* \in \calI$, the following holds.
			\begin{enumerate}[label=(\ref{claim:schedule-modified-properties}\alph*)]
				\item All integers in ${S_\anc}$ are integer multiplies of $\max\{2^{-h+1}|I^*|, 2^h\}$, and $|S_\anc \cap (\sfbegin(I^*), \sfend(I^*))| \leq (s-1)(h-1)$. This implies that $\calI^\cut_{{S_\anc}}[I^*]$ does not contain bottom intervals. \label{property:schedule-modified-input-Sanc}
				\item All sets in $\{{J_\anc}\} \cup \{\tilde J_I\}_{I \in \calI^\cut_{{S_\anc}}[I^*]}\cup\{\tilde K_I\}_{I \in \calI^\maxuncut_{{S_\anc}}[I^*]}$ are mutually disjoint. \label{property:schedule-modified-input-partition}
				\item For every $I \in \calI^\cut_{{S_\anc}}[I^*] \cap \calI_\sftop[I^*]$, we have $\Delta(\tilde J_I) \leq \delta|\tilde J_I| + \delta'|I|$.\smallskip \label{property:schedule-modified-input-chain-length}
				\item For every $I \in \calI^\cut_{{S_\anc}}[I^*] \cap \calI_\sfmid[I^*]$, we have $\tilde J_I = \emptyset$. \label{property:schedule-modified-input-middle-empty}
				\item The sequence $(\tilde J_{I} \text{ or }  \tilde K_{I})_{I \in \calI^\cut_{{S_\anc}}[I^*]\cup\calI^\maxuncut_{{S_\anc}}[I^*]}$ according to $\inorderless$ respects the precedence constraints, where $(\tilde J_{I} \text{ or }  \tilde K_{I})$ indicates either $\tilde J_I$ or $\tilde K_I$ depending on which one is given in the input. \label{property:schedule-modified-input-precedence}
				\item For every $j \in {J_\anc}$, we have that $\tilde b^\anc_j, \tilde e^\anc_j \in S_{\anc}$; moreover, for every $j \in {J_\anc}$, we have $\tilde b^\anc_j \leq \sfbegin(I^*)$ or $\tilde e^\anc_j \geq \sfend(I^*)$. \label{property:schedule-modified-input-b-e-multiplies}
			\end{enumerate}
		\end{claim}
		\begin{proof}
			So, Property~\ref{property:schedule-modified-input-Sanc} holds.  Step~\ref{step:schedule-modified-enumerating-S} we guaranteed that $S_I$ satisfy Property~\ref{property:modified-dyadic-S}, and ${S_\anc}$ is the union of $S_{\hat I}$'s for strict ancestors $\hat I$ of $I^*$.  Thus all integers in ${S_\anc}$ are integer multiplies of $\max\{2^{-h+1}|I^*|, 2^h\}$. By Observation~\ref{obs:cutting-small}, we have $|S_\anc \cap (\sfbegin(I^*), \sfend(I^*))| \leq (s-1)h$. 
			For Property~\ref{property:schedule-modified-input-b-e-multiplies}, notice that at every recursion of $\schedule\mhyphen\modified$, we guaranteed that for any $j \in J_{I^*}$ we have $\tilde b^\anc_j, \tilde e^\anc_j \in S_\anc \cup S_{I^*}$. 
			The analysis for the second half of  Property~\ref{property:schedule-modified-input-b-e-multiplies} and the other four properties are the same as that in Claim~\ref{claim:schedule-properties}.
		\end{proof}
	

	The following lemma shows the validity of the output for each recursion of $\schedule\mhyphen\modified$. Its proof is almost identical to that of Lemma~\ref{lemma:schedule-valid}, the only difference being that we need to check the valid of $S$ in the sets. 
		\begin{lemma}[Counterpart of Lemma~\ref{lemma:schedule-valid}]
			\label{lemma:schedule-modified-valid}
			Suppose some recursion of $\schedule\mhyphen\modified$ takes $\big(I^*,  J_\anc, \tilde b^\anc,  \break \tilde e^\anc,  {S_\anc}, (\tilde J_{I})_{I \in \calI^\cut_{{S_\anc}}[I^*]}, (\tilde K_{I})_{I \in \calI^\maxuncut_{{S_\anc}}[I^*]}\big)$ as input and returns $\big(\bbJ^\best,\sigma^\best\big) \neq (\bot, \bot)$.   Then $\bbJ^\best = (J_\anc, \tilde b^\anc, \allowbreak \tilde e^\anc, {S_\anc}, (S_I)_{I \in \calI_\sftop[I^*]}, (J^\best_I)_{I \in \calI[I^*]})$ is a {partial modified dyadic} system over $I^*$ and $\sigma^\best$ is a virtually-valid schedule for $\bbJ^\best$.  Moreover, $J^\best_I = \tilde J_I$ for every $I \in \calI^\cut_{S_\anc}[I^*]$, and $J^\best_{\subseteq I} = \tilde K_I$ for every $I \in \calI^\maxuncut_{S_\anc}[I^*]$.
		\end{lemma}
		\begin{proof}
			The lemma can be proved using mathematical induction and we only highlight the difference between the proof and that for Lemma~\ref{lemma:schedule-valid}. Consider the case $I^* \in \calI_\sftop \cup \calI_\sfmid$, and focus on the last iteration of the three nested loops in which  Step \ref{step:schedule-modified-update-0} to \ref{step:schedule-modified-update-1} are executed and the last moment of the iteration. The induction hypothesis for the two sub-recursions of $\schedule\mhyphen\modified$ made in Step~\ref{step:schedule-modified-recurse-left} and~\ref{step:schedule-modified-recurse-right} says
			\begin{itemize}
				\item $\bbJ^\sfL:=\big(J_\anc^\sfL,   b^\anc|_{J_\anc^\sfL}, e^\anc|_{J_\anc^\sfL}, S_\anc, (S_I)_{I \in \calI_\sftop[\sfleft(I^*)]}, (J^\sfL_I)_{I \in \calI[\sfleft(I^*)]}\big)$ is a {partial modified dyadic} system over $\sfleft(I^*)$ and $\sigma^\sfL$ is a valid schedule for $\bbJ^\sfL$. Moreover $J^\sfL_{\subseteq I} =  K_{I}$ for every $I \in \calI^\maxuncut_{S_\anc \cup S_{I^*}}[\sfleft(I^*)]$. 
				\item $\bbJ^\sfR:=\big(J_\anc^\sfR,   b^\anc|_{J_\anc^\sfL}, e^\anc|_{J_\anc^\sfR}, S_\anc, (S_I)_{I \in \calI_\sftop[\sfright(I^*)]}, (J^\sfR_I)_{I \in \calI[\sfright(I^*)]}\big)$ is a {partial modified dyadic} system over $\sfright(I^*)$ and $\sigma^\sfR$ is a valid schedule for $\bbJ^\sfR$. Moreover $J^\sfR_{\subseteq I} =  K_{I}$ for every $I \in \calI^\maxuncut_{S_\anc \cup S_{I^*}}[\sfright(I^*)]$. 
			\end{itemize}
			
			We first show that $J^\best_{\subseteq I} = \tilde K_I$ for every $I \in \calI^\maxuncut_{S_\anc}$.  Focus on such an interval $I$. Since points in $S_\anc$ do not cut $I$, whether a sub-interval $I'$ of $I$ is cut by $S_\anc \cup S_{I^*}$ is determined by whether it is cut by $S_{I^*}$.  By Property~\ref{property:pushdown-partition} of the procedure $\pushdown$, we have that $\union_{I' \in \calI^\cut_{S_{I^*}}[I]} J_{I'} \cup \union_{I' \in \calI^\maxuncut_{S_{I^*}}[I]}K_{I'} = K_I = \tilde K_I$. By the induction hypothesis, for every $I' \in \calI^\maxuncut_{S_{I^*}}[I]$, we have $J^\best_{\subseteq I'} =  K_{I'}$ by the induction hypothesis and the way we constructed the sets $J^\best_{I''}$'s. Also, $J_{I'} = J^\best_{I'}$ for every $I' \in \calI^\cut_{S_{I^*}[I]}$. Therefore, 
			\begin{align*}
				\tilde K_I = \union_{I' \in \calI^\cut_{S_{I^*}}[I]} J_{I'} \cup \union_{I' \in \calI^\maxuncut_{S_{I^*}}[I]}K_{I'} = \union_{I' \in \calI^\cut_{S_{I^*}}[I]} J^\best_{I'} \cup \union_{I' \in \calI^\maxuncut_{S_{I^*}}[I]}J^\best_{\subseteq I'} = J^\best_{\subseteq I}.
			\end{align*}
			
			Again, $\bbJ^\best$ satisfies all the properties in Definition~\ref{def:modified-dyadic-system}: they are implied by these properties for the two partial modified dyadic systems $\bbJ^\sfL$ and $\bbJ^\sfR$.  The virtual-validity of  $\sigma^\best$ is implied by the virtual-validity of $\sigma^\sfL$ and $\sigma^\sfR$, and  the $b^\anc_j$ and $e^\anc_j$ values $j \in J^\best_{I^*}$ are the same as their $b^{\bbJ^\best}_j$ and $e^{\bbJ^\best}_j$ values. 
		\end{proof}	
				
		The following lemma is the counterpart of Lemma~\ref{lemma:schedule-good}; a slight difference is that we use $\sigma''^*$ instead of $\sigma'^*$ in the lemma (Recall that $b^*_j = b^{\bbJ^*}_j$ and $e^*_j = b^{\bbJ^*}_j$ for every top job $j$ in $\bbJ^*$):
		\begin{lemma}[Counterpart of Lemma~\ref{lemma:schedule-good}]
			\label{lemma:schedule-modified-good}
			Suppose at the beginning of some recursion of $\schedule\mhyphen\modified$, we have $\{(\tilde b^\anc_j, \tilde e^\anc_j] : j \in  {J_\anc}\} = \{(b^*_j, e^*_j]: j \in \sigma''^{*-1}(I^*) \cap J^*_{\supsetneq I^*}\}$, $S_\anc = S^*_{\supsetneq I^*}$, $\tilde J_I = J^*_I$ for every $I\in \calI^\cut_{{S_\anc}}[I^*]$ and $ K_I = K^*_I =  J^*_{\subseteq I}$ for every $I \in \calI^\maxuncut_{{S_\anc}}[I^*]$.  Then the returned schedule $\sigma^\best$ has at least $|\sigma''^{*-1}(I^*)|$ jobs scheduled.
		\end{lemma}
		
		
		The proof of the lemma is very similar to that of Lemma~\ref{lemma:schedule-good} and thus we only highlight the difference. First, in the new algorithm we need to guess $S_{I^*}$, and thus in our analysis we focus on the iteration of the outermost loop in which we have $S_{I^*} = S^*_{I^*}$.  Second,  we only considered good partitions in Step~\ref{step:schedule-modified-partition-Janc}, but Lemma~\ref{lemma:S''-star} says that the partition according to $\sigma''^*$ is always good. 
		
%
		So, the main algorithm will return a modified dyadic system $\bbJ^\best$ and a valid schedule $\sigma^\best$ for $\bbJ^\best$ with at least $\sigma''^{*-1}([T])$ jobs scheduled, since the parameters passed to $\schedule\mhyphen\modified$ satisfy the conditions of Lemma~\ref{lemma:schedule-modified-good}.
\subsection{Analysis of Running Time}
	The following claim will be used to derive the improved running time:
	\begin{claim}
		\label{claim:calI'-small}
		The set $\calI'$ constructed in Step~\ref{step:schedule-modified-calI'} has $|\calI'| \leq (s-1)h$.  
	\end{claim}
	\begin{proof}
		Notice that $\calI' = \calI^\cut_{{S_\anc} \cup S_{I^*}}[I^*] \cap \calI^\uncut_{{S_\anc}}[I^*]$ is the set of intervals in $\calI[I^*]$ that are not cut by ${S_\anc}$ but cut by ${S_\anc} \cup S_{I^*}$. If $I^* \in \calI_\sfmid$, then $S_{I^*} = \{\sfcenter(I^*)\}$ and $\calI'$ is either $\emptyset$ or $\{I^*\}$. Now assume $I^* \in \calI_\sftop$. By Observation~\ref{obs:cutting-small}, each point in $S_{I^*}$ can cut at most $h$ intervals in $\calI[I^*]$. The claim holds since $|S_{I^*}| \leq s-1$.
	\end{proof}
	
	\label{subsec:algo-modified-running-time}
		Finally we analyze the running time of the algorithm.
		\begin{lemma}
			\label{lemma:schedule-modified-run-time}
			The running time of $\schedule\mhyphen\modified$ for $I^* = [T]$ is at most $n^{O\left(\frac{m^4}{\epsilon^3}\log^3\log n\right)}$.
		\end{lemma} 
		\begin{proof} Again, for every $\ell \in [0, L)$, we define $R_\ell$ to be the maximum number of times we call $\schedule\mhyphen\modified$ as sub-recursions in a recursion of $\schedule\mhyphen\modified$ for some $I^* \in \calI_\ell$.  Let $R_{L}$ be the worst case running time for a recursion of $\schedule\mhyphen\modified$ for some $I^* \in \calI_{L}$.  Again, it suffices to bound $\poly(n)\prod_{\ell=0}^{L}R_\ell$.

		First, we bound $R_\ell$ for $\ell < L$ and focus on any $I^* \in \calI_\ell$. By Claim~\ref{claim:calI'-small}, we have $|\calI'| \leq (s-1)h$ before Step~\ref{step:schedule-modified-enumerate-g}.  If $\ell \leq L-h' - h$, $(g_I)_{I \in \calI'}$ has a total length of $(s-1)hp$.  If $\ell \geq L-h'-h+1$ but $\ell \leq L-h$, the total length is at most $(s-1)h \max\{p, m 2^{-(L - h')}T \} =(s-1)h \max\{p, m 2^{h+h'}\}  \leq 2^{h+h'}(s-1)h m$. (Notice that for the improved algorithm, we have $p = O(shm^2/\epsilon \cdot \log \log T) = O(\frac{m^3}{\epsilon^2}\log^2\log T)$ and $2^{h + h'} = \Omega(\frac{m^2\log T}{\epsilon^2})$). If $\ell \geq L-h+1$, then the length is $0$. 
		
		Now we consider the number of different ways to split ${J_\anc} \cup J_{I^*}$ into $J_\anc^\sfL, J_\anc^\sfR$ and $J_\discarded$. For each $j \in {J_\anc} \cup J_{I^*}$, we have $b^\anc_j \leq \sfbegin(\sfleft(I^*))$ or $e^\anc_j \geq \sfend(\sfleft(I^*))$, and we also have $b^\anc_j \leq \sfbegin(\sfright(I^*))$ or $e^\anc_j \geq \sfend(\sfright(I^*))$. By Property~\ref{property:schedule-modified-input-Sanc} and \ref{property:schedule-modified-input-b-e-multiplies}, we have $\big|\{b^\anc_j: j \in J_\anc \cup J_{I^*}\}\big| + \big|\{e^\anc_j: j \in  J_\anc \cup J_{I^*}\}\big| \leq |S_{\anc} \cap S_{I^*}| + 1 \leq (s-1)h + 1 \leq sh$. 
		Thus there are at most $2\cdot sh$ distinct elements in $\{(b^\anc_j, e^\anc_j] \cap \sfleft(I^*): j \in {J_\anc} \cup J_{I^*}\} \cup \{(b^\anc_j, e^\anc_j] \cap \sfright(I^*): j \in {J_\anc} \cup J_{I^*}\}$.  Since we only consider good partitions, in which multiplicity of each interval in $\{(b^\anc_j, e^\anc_j] \cap \sfleft(I^*): j \in J^\sfL_\anc \cup J_{I^*}\}$ or $\{(b^\anc_j, e^\anc_j] \cap \sfright(I^*): j \in J^\sfR_\anc \cup J_{I^*}\}$ is a multiply of $\rho|I^*|/2$, the number of partitions we consider is at most $(2m|I^*|/\rho + 1)^{2sh}$. 
		
		Therefore, fore $\ell \leq L-h-h'$, we have 
			\begin{align*}
				&\quad \log R_\ell \leq 1  + shp+ 2sh \log \big((2m/\rho) + 1\big) \leq 1 + shp + O(sh \log \log T) \leq O(shp)\\ 
				&= O\left( sh \cdot \frac1\delta \log (m/\delta') \right)= O\left(\frac{(shm)^2\log \log T}{\epsilon}\right) = O\left(\frac{m^4\log^3\log T}{\epsilon^3}\right).
			\end{align*}
			If $\ell \geq L - h' -h +1$, we have 
			\begin{align*}
				\log R_\ell &\leq 1 + 2^{h+h'}shm + 2sh \log \big((2m/\rho) + 1\big) \leq O(2^{h+h'}shm) = O\left(\frac{m^4\log T\log\log T}{\epsilon^3}\right).
			\end{align*}

			Now we bound $\log R_{L}$ can still be bounded by $O\left(m2^h\log 2^h\right) = O\left(\frac{m^2\log n \log\log n}{\epsilon}\right)$.
			
			Therefore, we have 
			\begin{align*}
				\sum_{\ell = 0}^L\log R_L &\leq (L-h-h'+1)\cdot O\left(\frac{m^4\log^3\log T}{\epsilon^3}\right) + (h+h'-1)\cdot O\left(\frac{m^4\log T\log\log T}{\epsilon^3}\right)\\
				&\quad + O\left(\frac{m^2\log n \log\log n}{\epsilon}\right) \qquad = \qquad O\left(\frac{m^4 \log n \log^3 \log n}{\epsilon^3}\right).
			\end{align*}
			
		Above we used that $h + h' = O(\log \log T)$.
		Thus, the running time of $\schedule\mhyphen\modified$ for $I^* = [T]$ is at most $n^{O\left(\frac{m^4}{\epsilon^3}\log^3\log n\right)}$.
		\end{proof}
				
		\paragraph{Wrapping Up} Running the main algorithm, we can obtain a {partial modified dyadic} system $\bbJ$ and a virtually-valid schedule $\sigma''$ for $\bbJ$ with  $|\sigma''^{-1}(\discarded)| \leq |\sigma''^{*-1}(\discarded)| \leq \frac{3\epsilon T}{4}$ by Lemma~\ref{lemma:S''-star}. By Lemma~\ref{lemma:canonicalize-modified} and \ref{lemma:virtually-valid-to-valid-modified}, we can convert $\sigma''$ to a valid schedule $\sigma$ for $\bbJ$ with $|\sigma^{-1}(\discarded) \setminus \sigma''^{-1}(\discarded)| \leq \frac{3\epsilon T}{4}$. Thus, we have  $|\sigma^{-1}(\discarded)| \leq \frac{\epsilon T}{4} + \frac{3\epsilon T}{4} = \epsilon T$.  The running time of the whole algorithm is $n^{O\left(\frac{m^4}{\epsilon^3}\log^3\log n\right)}$ by Lemma~\ref{lemma:schedule-modified-run-time}.  Therefore we proved Theorem~\ref{thm:main}.

%% file: discussion.tex
\section{Discussion}
We showed how to obtain a $(1+\epsilon)$-approximation for $Pm|\text{prec}, j_j = 1|C_{\max}$ in running time $n^{O_{m, \epsilon}(\log^3\log n)}$, by using a novel combinatorial algorithm based on making guesses  about the optimum solution.  Though we have the improved running time, obtaining a PTAS for the problem remains open. We believe our framework has the potential to achieve this goal. Currently the $\poly\log\log n $ factors in the exponent come from the number of interesting levels in one recursion of the algorithm. It is possible that our framework with a more careful analysis of the number of discarded jobs can lead to a PTAS for the problem.  It is also interesting to see if a Sherali-Adams hierarchy based algorithm can give a result similar to ours. 